\documentclass[11pt]{article}

\usepackage{amsthm}
\usepackage{graphicx} %
\usepackage{array} %

\usepackage[bb=boondox]{mathalpha}

\usepackage{amsmath, amssymb, amsfonts, verbatim, soul}
\usepackage{hyphenat,epsfig,subcaption,multirow}
\usepackage[font=small,labelfont=bf]{caption}
\usepackage{nicefrac}
\usepackage{paralist}

\usepackage[usenames,dvipsnames]{xcolor}
\usepackage[ruled]{algorithm2e}

\DeclareFontFamily{U}{mathx}{\hyphenchar\font45}
\DeclareFontShape{U}{mathx}{m}{n}{
      <5> <6> <7> <8> <9> <10>
      <10.95> <12> <14.4> <17.28> <20.74> <24.88>
      mathx10
      }{}
\DeclareSymbolFont{mathx}{U}{mathx}{m}{n}
\DeclareMathSymbol{\bigtimes}{1}{mathx}{"91}

\usepackage{tcolorbox}
\tcbuselibrary{skins,breakable}
\tcbset{enhanced jigsaw}

\usepackage[normalem]{ulem}
\usepackage[compact]{titlesec}

\definecolor{DarkRed}{rgb}{0.5,0.1,0.1}
\definecolor{DarkBlue}{rgb}{0.1,0.1,0.5}

\usepackage{nameref}
\definecolor{ForestGreen}{rgb}{0.1333,0.5451,0.1333}
\definecolor{Red}{rgb}{0.9,0,0}
\usepackage[linktocpage=true,
	pagebackref=true,colorlinks,
	linkcolor=DarkRed,citecolor=ForestGreen,
	bookmarks,bookmarksopen,bookmarksnumbered]
	{hyperref}
\usepackage[noabbrev,nameinlink]{cleveref}
\crefname{property}{property}{Property}
\creflabelformat{property}{(#1)#2#3}
\crefname{equation}{eq}{Eq}
\creflabelformat{equation}{(#1)#2#3}

\usepackage{bm}
\usepackage{url}
\usepackage{xspace}
\usepackage[mathscr]{euscript}

\usepackage{tikz}
\usetikzlibrary{arrows}
\usetikzlibrary{arrows.meta}
\usetikzlibrary{shapes}
\usetikzlibrary{backgrounds}
\usetikzlibrary{positioning}
\usetikzlibrary{decorations.markings}
\usetikzlibrary{decorations.pathreplacing} %
\usetikzlibrary{patterns}
\usetikzlibrary{calc}
\usetikzlibrary{fit}
\usetikzlibrary{decorations}

\usepackage[framemethod=TikZ]{mdframed}

\usepackage[noend]{algpseudocode}
\makeatletter
\def\BState{\State\hskip-\ALG@thistlm}
\makeatother

\usepackage{cite}
\usepackage{enumitem}
\setlist[itemize]{leftmargin=20pt}
\setlist[enumerate]{leftmargin=20pt}

\usepackage[margin=1in]{geometry}

\usepackage{thmtools}
\usepackage{thm-restate}

\newtheorem{theorem}{Theorem}
\newtheorem{lemma}{Lemma}[section]
\newtheorem{proposition}[lemma]{Proposition}

\newtheorem{claim}[lemma]{Claim}
\newtheorem{fact}[lemma]{Fact}

\newtheorem{definition}[lemma]{Definition}

\newtheorem*{claim*}{Claim}
\newtheorem*{assumption*}{Assumption}
\newtheorem*{proposition*}{Proposition}
\newtheorem*{lemma*}{Lemma}

\newtheorem{observation}[lemma]{Observation}

\newtheorem*{observation*}{Observation}

\newtheorem*{theorem*}{Theorem}

\crefname{lemma}{Lemma}{Lemmas}
\crefname{claim}{claim}{claims}
\crefname{property}{Property}{Properties}
\crefname{invariant}{Invariant}{Invariants}

\newtheorem{mdresult}{Result}
\newenvironment{result}{\begin{mdframed}[backgroundcolor=lightgray!40,topline=false,rightline=false,leftline=false,bottomline=false,innertopmargin=2pt]\begin{mdresult}}{\end{mdresult}\end{mdframed}}

\theoremstyle{definition}

\newenvironment{ourbox}{\begin{mdframed}[hidealllines=false,innerleftmargin=10pt,backgroundcolor=white!10,innertopmargin=10pt,innerbottommargin=5pt,roundcorner=10pt]}{\end{mdframed}}

\newenvironment{cbox}{\begin{mdframed}[backgroundcolor=ForestGreen!10,topline=false,bottomline=false, innerbottommargin=5pt,innertopmargin=5pt]}{\end{mdframed}}

\newtheorem{Definition}[lemma]{Definition}
\renewenvironment{definition}{\begin{cbox}\begin{Definition}}{\end{Definition}\end{cbox}}

\newtheorem{mdalgorithm}{Algorithm}

\allowdisplaybreaks

\renewcommand{\qed}{\nobreak \ifvmode \relax \else
      \ifdim\lastskip<1.5em \hskip-\lastskip
      \hskip1.5em plus0em minus0.5em \fi \nobreak
      \vrule height0.75em width0.5em depth0.25em\fi}

\newcommand{\Qed}[1]{\rlap{\qed$_{\textnormal{~~\Cref{#1}}}$}}

\newcommand*\samethanks[1][\value{footnote}]{\footnotemark[#1]}

\newcommand{\logstar}[1]{\ensuremath{\log^{*}\!{#1}}}

\renewcommand{\leq}{\leqslant}
\renewcommand{\geq}{\geqslant}

\newcommand{\rs}{{{Ruzsa-Szemerédi}}\xspace}

\newcommand{\Gbase}{G_{\textnormal{\textsc{base}}}}
\newcommand{\Ebase}{E_{\textnormal{\textsc{base}}}}
\newcommand{\Mbase}{M^{\star}_{\textnormal{\textsc{base}}}}

\newcommand{\GG}{\ensuremath{\mathcal{G}}}

\newcommand{\belong}[2]{\ensuremath{\textnormal{\texttt{Belong}}(#1,#2)}}

\newcommand{\Ot}{\ensuremath{\widetilde{O}}}
\newcommand{\eps}{\ensuremath{\varepsilon}}

\newcommand{\bracket}[1]{\left[#1\right]}
\newcommand{\paren}[1]{\ensuremath{\left(#1\right)}\xspace}
\newcommand{\card}[1]{\left\vert{#1}\right\vert}

\newcommand{\set}[1]{\ensuremath{\left\{ #1 \right\}}}
\newcommand{\poly}{\mbox{\rm poly}}

\DeclareMathOperator*{\Exp}{\ensuremath{{{E}}}}
\DeclareMathOperator*{\Prob}{\ensuremath{\textnormal{Pr}}}
\renewcommand{\Pr}{\Prob}

\newcommand{\Ex}{\Exp}

\newenvironment{tbox}{\begin{tcolorbox}[
		enlarge top by=5pt,
		enlarge bottom by=5pt,
		 breakable,
		 boxsep=0pt,
                  left=4pt,
                  right=4pt,
                  top=10pt,
                  arc=0pt,
                  boxrule=1pt,toprule=1pt,
                  colback=white
                  ]%
	}
{\end{tcolorbox}}

\newcommand{\supp}[1]{\ensuremath{\textnormal{\text{supp}}(#1)}}
\newcommand{\distribution}[1]{\ensuremath{\textnormal{dist}(#1)}\xspace}

\newcommand{\II}{\ensuremath{\mathbf{I}}}
\newcommand{\HH}{\ensuremath{\mathbf{H}}}

\newcommand{\mireal}[1][]{
  \ifx\relax#1\relax%
    \II(\mione \,; \mitwo)%
  \else%
    \II(\mione \,; \mitwo\mid #1)%
  \fi
}
\newcommand{\en}[1]{\ensuremath{\HH(#1)}}

\newcommand{\itfacts}[1]{\Cref{fact:it-facts}-(\ref{part:#1})\xspace}

\newcommand{\Grs}{\ensuremath{G_{\textnormal{\sc ers}}}}
\newcommand{\Lrs}{\ensuremath{L_{\textnormal{\sc ers}}}}
\newcommand{\Rrs}{\ensuremath{R_{\textnormal{\sc ers}}}}
\newcommand{\Ers}{\ensuremath{E_{\textnormal{\sc ers}}}}

\title{Semi-Streaming Matching in a Single Pass I: \\ A New Framework for Lower Bounds via Blueprints\footnote{A preliminary version of this paper appeared in the \emph{ACM Symposium on Theory of Computing, STOC 2026}. \newline}}
\author{%
  \begin{tabular}{c}
    Sepehr Assadi\footnote{(sepehr@assadi.info, max.jiang@uwaterloo.ca, mxiang@uwaterloo.ca) Supported in part by an NSERC Discovery Grant, a Faculty of Mathematics Research Chair grant and an Undergraduate Research Fellowship from the University of Waterloo.}  \qquad Max Jiang\samethanks \qquad Mars Xiang\samethanks \\[4ex]
    {University of Waterloo}
  \end{tabular}%
}

\date{}

\begin{document}
\maketitle


\begin{abstract}

\medskip

In the semi-streaming model, we have an $n$-vertex graph $G=(V,E)$ whose edges arrive in an arbitrary order in a stream. The goal is to make one or a few passes over the stream, 
use a limited memory of $\Ot(n) := O(n \cdot \poly\!\log{n})$ bits, and output a solution to the problem at hand at the end. A central open question in this area is to determine 
the best approximation ratio possible for the maximum matching problem via \emph{single-pass} semi-streaming algorithms. 

\medskip

This problem admits a simple $0.5$-approximation algorithm---by maintaining a maximal 
matching greedily---which, despite extensive efforts, has remained the state of the art. Lower bounds for this problem have also been few and far between with best known bounds ruling 
out better than $1/(1+\ln{(2)}) \approx 0.590$ approximation, using a highly complicated construction motivated by the literature on \rs (RS) graphs from extremal graph theory. 

\medskip

We develop a new framework for proving lower bounds for the semi-streaming matching problem. Our framework abstracts out the extremal graph theory and information theoretic arguments in the lower bounds, and reduces
the problem to constructing certain \emph{constant-size} graphs, which we call \textbf{blueprints}. Not only existing lower bounds 
can be captured by these blueprints---leading to far simpler and more concise arguments---but also we can design new blueprints that can be used to rule out 
better than $(8-2\sqrt{10})/3 \approx 0.558$-approximation for the semi-streaming matching problem. We believe this approach can be of its own independent interest
and lead to further improvements on this tantalizing open question. 

\end{abstract}

\vspace{1cm}

\paragraph{Follow up work.} More recently, we built on this framework to rule out any single-pass semi-streaming algorithm with approximation ratio strictly better than half~\cite{AssadiJX26b}. This shows that the simple greedy algorithm for the problem is already optimal, settling the central open question at the heart of this work. The remainder of the current paper provides the original  version of this work with no mention of that new result.

\pagenumbering{roman}

\clearpage
\setcounter{tocdepth}{2}
\tableofcontents
\clearpage
\pagenumbering{arabic}
\setcounter{page}{1}


\newcommand{\tapprox}{\ensuremath{A^*}}
\newcommand{\tvalue}{\ensuremath{V^*}}

\section{Introduction}\label{sec:intro} 

The semi-streaming model---formalized by~\cite{FeigenbaumKMSZ05}---is defined as follows: we have an $n$-vertex graph $G=(V,E)$ whose edges are given
to the algorithm in some arbitrarily ordered stream; the algorithm makes one or a few passes over this stream, uses a memory of $\Ot(n):= O(n \cdot \poly\!\log{(n)})$ bits, and at the end of the stream, 
outputs a solution to the problem at hand. In this paper, we study \textbf{single-pass} semi-streaming algorithms for the \textbf{maximum matching} problem. 

Already two decades ago,~\cite{FeigenbaumKMSZ05} observed that this problem admits a simple semi-streaming algorithm achieving a $0.5$-approximation: greedily maintain a {maximal} matching of the arriving edges in the stream. 
Despite significant effort, this simple algorithm has remained the state of the art for this problem\footnote{It is only known that one can break this approximation in $o(n^2)$ space~\cite{AssadiBKL23}, namely, in space asymptotically less than storing the entire 
input explicitly (although due to the reliance of the algorithm on the triangle removal lemma~\cite{Fox11}, the improvement over the trivial $n^2$ bound is only by a $2^{\Theta(\logstar{(n)})}$ factor.)}. Indeed, whether or not the approximation ratio of this simple 
algorithm can be improved has since become one of the most central problems in the graph streaming literature; see, e.g.~\cite{KonradMM12,GoelKK12,AssadiKL17,Kapralov21,AssadiBKL23,FeldmanS24} and references therein (see
also~\cite{McGregor05,AhnG11,PazS17,AssadiB21,FischerMU22,AssadiS23,KonradN24,Bernstein24,Assadi25,AssadiBKNS25} for pointers to some other related variants of the problem including multi-pass, random-order, or dynamic streaming algorithms).  

The progress on this question from the lower bound side has also been limited so far.~\cite{FeigenbaumKMSZ05} proved that finding an exact maximum matching is not possible 
via semi-streaming algorithms. Almost a decade after,~\cite{GoelKK12} proved that even a $2/3$-approximation is not possible, which soon after was followed by a $(1-1/e) \sim 0.632$-approximation lower bound of~\cite{Kapralov13}.
Almost another decade passed until~\cite{Kapralov21} improved this to a $1/(1+\ln{(2)}) \sim 0.590$-approximation lower bound, which remains the state of the art. 
It is worth noting that the lower bounds of~\cite{Kapralov13} and~\cite{Kapralov21} are based on extending prior lower bounds of~\cite{KarpVV90} and~\cite{EpsteinLSW13} in the \emph{online preemptive} model---with the same bounds
as above, respectively---to semi-streaming algorithms.  

The slow progress on lower bounds can be partially attributed to the extreme technicality of the existing arguments. The starting point of these arguments 
is an already complex construction of \emph{\rs (RS)} graphs~\cite{RuzsaS78} due to~\cite{FischerLNRRS02}. These constructions 
 are then used in a (completely) white-box manner to provide a ``geometric interpretation'' of prior online lower bound instances in~\cite{KarpVV90,EpsteinLSW13} to
construct hard instances  for semi-streaming algorithms (see~\cite[Section 2]{Kapralov21}); this step is the heart of these proofs. Finally, these hard instances are analyzed using information theory tools 
to establish the final lower bound. The end results are thus highly formidable arguments which in the case of~\cite{Kapralov21} span over 150 pages. 

\subsection{Our Contributions}\label{sec:contribution} 

In this paper, we present an improved lower bound for the semi-streaming matching problem. 

\begin{result}\label{res:main}
	Any single-pass semi-streaming algorithm (deterministic or randomized) for the maximum bipartite matching problem cannot achieve an approximation ratio better than
	\[
		(8-2\sqrt{10})/3 \sim 0.55848155\cdots.  
	\]
\end{result}

The main contribution of our work is not the numerical improvement over~\cite{Kapralov21}, but rather the new \textbf{lower bound framework} introduced in its proof, which we explain next. 

\subsection*{A New Framework for Lower Bounds via Blueprints} 
Our lower bound framework entirely abstracts out the extremal graph theory and information theory arguments of prior work, and reduces the task of proving lower bounds to constructing certain \emph{constant-size} graphs, which we call \textbf{blueprints}
(since our framework receives them as input and uses them as the `blueprint' of hard instances it creates). 

Defining blueprints or providing proper intuition behind their use will already get too technical and we postpone it to later. For now, we just mention that 
at a (very) high level, a blueprint is a bipartite graph on a constant number of \emph{labeled} vertices and edges\footnote{To put this in context, re-proving the $2/3$-approximation lower bound of~\cite{GoelKK12} requires a blueprint on $4$ vertices, 
and getting a better-than-$2/3$ lower bound can be done with a blueprint on $64$ vertices (see~\Cref{fig:blueprints}). The blueprint in~\Cref{res:main} is however less explicit and we do not bound its number of vertices beyond showing it is finite.}. 
The edges form a \emph{matching} and the goal is to maximize the \emph{relative} size of this matching compared to the number of vertices. 
The constraints are of the following type: for any pairs of vertices $u,v$, there exists a fixed set $S_{uv}$, depending only on the labels of vertices and the edge $(u,v)$; whenever the edge $(u,v)$ exists, 
for each pair $(x,y) \in S_{u,v}$, at most one of $x$ or $y$ can have any incident edges. 

\begin{figure}[h!]
	\centering
	\includegraphics[scale=0.5]{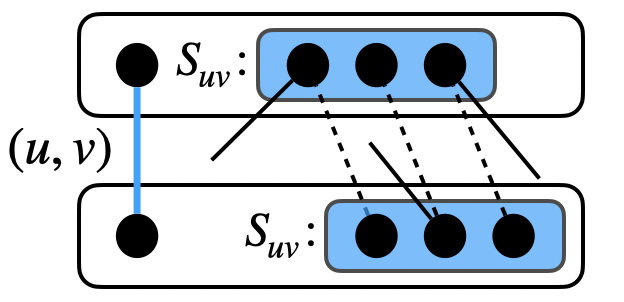}
	\caption{A high level illustration of a blueprint: existence of each edge $(u,v)$ introduces a constraint between some pairs of vertices that forces at most one vertex per pair to have any edge.
	Here, the banned pairs of a single edge $(u,v)$ are shown with dashed lines; only one endpoint of each dashed line has an edge.}\label{fig:blueprint-intro}
\end{figure}
\noindent
After the definition of these blueprints, the proof of our lower bound framework goes in three steps: 
\begin{itemize}
	\item \textbf{Step 1:} Introducing a new generalization of RS graphs, called \emph{ERS} graphs (short for \emph{extended} RS graphs). ERS graphs can be seen
	as a somewhat-dense analogue of \emph{Hamming graphs}\footnote{A Hamming graph $H(d,q)$ has vertices $[q]^d$ and edges between vertices that differ on exactly one coordinate.}, 
	the same way RS graphs of~\cite{FischerLNRRS02,GoelKK12} are somewhat-dense analogues of hypercube graphs\footnote{See~\cite[Section 4.3]{Raskhodnikova03} that fully spells out the steps of going from a hypercube to RS graphs in~\cite{FischerLNRRS02}.}. 
	\item \textbf{Step 2:} Defining a relatively simple graph product that takes any blueprint and an ERS graph and constructs a \emph{hard} instance for the semi-streaming matching problem. 
	While this step is inspired by the prior work of~\cite{Kapralov13,Kapralov21}, it almost entirely bypasses the intricacies and highly complicated nature of those arguments given its \emph{blackbox} reliance on ERS graphs 
	as opposed to the white-box applications of the ideas of~\cite{FischerLNRRS02,GoelKK12} used in~\cite{Kapralov13,Kapralov21}. 
	\item \textbf{Step 3:} An information-theoretic analysis of these hard instances to prove a lower bound for semi-streaming algorithms. This step is generally not complicated in prior work either but in our work has become particularly easy
	given the modular structure of hard instances following the blueprint. Along the way, we show in a formal way that, on our hard instances, we only need to analyze algorithms that store edges of the graph explicitly in their memory and at the end output
	a large matching of the stored edges.  
\end{itemize}

Equipped with this framework, we can now \emph{solely} focus on designing ``good'' blueprints, which is a standalone optimization problem with its own rules and constraints. 
This constitutes the second main part of the paper. 

\subsection*{Construction of Blueprints}

We first show a construction of blueprints that rules out a $(2-\sqrt{2}) \sim 0.585$-approximation which already (slightly) improves upon the prior lower bound of~\cite{Kapralov21}. This matches
the \emph{target} bound of best known online preemptive matching results in~\cite{HuangPTTWZ19}, put forward by~\cite{Kapralov21} who stated that their techniques  ``can probably be extended'' to 
mimic hard instances of~\cite{HuangPTTWZ19} as well. Our blueprints now seamlessly recover this stronger bound
using far simpler and more direct arguments that do not even rely on online preemptive matching lower bounds directly.  

We then build on this blueprint to create a more complex recursive one that proves~\Cref{res:main}
and improves the prior best bounds of~\cite{Kapralov21} considerably. We note that, to our knowledge, our lower bound also improves the state of the art for the online preemptive matching problem as a corollary but we have not explored that direction in this paper.

The general strategy of the proofs is as follows: we consider $\tapprox$ as the \emph{infimum} of approximation ratios possible by all finite blueprints
and in each result, prove an upper bound on $\tapprox$. These upper bounds are proven via combinatorial arguments based on the rules and constraints of blueprints and have the following flavor: given any blueprint with approximation ratio above the target upper bound, we show how to add a \emph{gadget} to the blueprint that decreases its approximation ratio. 

We conclude this section by remarking that \emph{if} one attempts to recover 
the hard instances of the semi-streaming matching problem from our blueprints (i.e., ``unroll'' the steps in our framework), the end result again would be a highly complicated construction as is the case of~\cite{Kapralov13,Kapralov21}. 
But, the power of our framework is that it \emph{abstracts away} all these intricacies. This allows us to design even more complex 
hard instances compared to prior lower bounds, and yet use far simpler arguments as we solely need to focus on constant-size blueprints and their own rules and constraints.  
We believe this approach based on blueprints can be of its own independent interest
and lead to further improvements on understanding the semi-streaming matching problem.  

\subsection{Organization of the Rest of This Paper} 
Our arguments in this paper are quite concise on one hand and at the same time require careful technical definitions. Given this, we forego having a lengthy informal technical overview of the entire paper upfront and instead in each section, provide ample intuition and discussions throughout the proofs after providing the formal definitions. The rest of the paper, after a short preliminaries section, is organized as follows.  
\begin{itemize}
	\item In~\Cref{sec:ers}, we define \emph{ERS} graphs and present a somewhat-dense construction of them with $n^{1+\Omega(1/\log\log{n})}$ edges (which is the by-now familiar bound of RS graphs in~\cite{FischerLNRRS02,GoelKK12}
	as well as hard instances of~\cite{Kapralov13,Kapralov21}).
	\item In~\Cref{sec:blueprints}, we formally define blueprints and provide an \emph{expansion} operation that turns them into large graphs for forming hard instances of the semi-streaming matching problem. We also prove some key properties of blueprints in this section. 
	\item In~\Cref{sec:blueprint-lower}, we use our expansion operation of blueprints to show that given any blueprint, we can prove a semi-streaming lower bound that rules out approximation ratios 
	determined by the quality of the blueprint. This constitutes our lower bound framework. 
	
	\item In~\Cref{sec:construct-blueprints}, we present some constructions of blueprints and use them, alongside our framework, to establish lower bounds for the semi-streaming matching problem and prove~\Cref{res:main}. 
\end{itemize}
Finally, we wrap up the paper in~\Cref{sec:concluding} with concluding remarks and pointers for future work.

\clearpage


\newcommand{\PS}[1]{\ensuremath{\texttt{Player}^{(#1)}}}
\newcommand{\ES}[1]{\ensuremath{E^{(#1)}}}

\newcommand{\msg}{\ensuremath{\texttt{msg}}}

\newcommand{\conc}{\circ}

\newcommand{\EE}[1]{\ensuremath{E}^{(#1)}}
\newcommand{\FF}[1]{\ensuremath{F}^{(#1)}}
\newcommand{\LL}[1]{\ensuremath{L}^{(#1)}}
\newcommand{\RR}[1]{\ensuremath{R}^{(#1)}}
\newcommand{\LLL}[1]{\ensuremath{L'}^{(#1)}}
\newcommand{\RRR}[1]{\ensuremath{R'}^{(#1)}}
\newcommand{\EEE}[1]{\ensuremath{E'}^{(#1)}}
\newcommand{\LRS}{\ensuremath{L_{RS}}}
\newcommand{\RRS}{\ensuremath{R_{RS}}}

\newcommand{\expand}[3]{\ensuremath{\textnormal{\texttt{Expansion}}(#1,#2,#3)}\xspace}
\newcommand{\sexpand}[3]{\ensuremath{\textnormal{\textsc{Strong-Expand}}_{#3}(#1,#2)}\xspace}
\newcommand{\inputs}[3]{\ensuremath{\textnormal{\textsc{Input}}_{#3}(#1,#2)}\xspace}

\newcommand{\bG}{\mathbb{G}}
\newcommand{\bV}{\mathbb{V}}
\newcommand{\bE}{\mathbb{E}}
\newcommand{\bL}{\mathbb{L}}
\newcommand{\bR}{\mathbb{R}}
\newcommand{\bv}{\mathbb{v}}
\newcommand{\bu}{\mathbb{u}}
\newcommand{\be}{\mathbb{e}}
\newcommand{\bff}{\mathbb{f}}
\newcommand{\bg}{\mathbb{g}}
\newcommand{\bw}{\mathbb{w}}

\newcommand{\bLL}[1]{\ensuremath{\bL}^{(#1)}}
\newcommand{\bRR}[1]{\ensuremath{\bR}^{(#1)}}
\newcommand{\bEE}[1]{\ensuremath{\bE}^{(#1)}}

\newcommand{\Mstar}{\ensuremath{M^{\star}}}

\newcommand{\nrs}{\ensuremath{n_{\textsc{ers}}}}

\newcommand{\Lnew}{\overline{L}}
\newcommand{\Rnew}{\overline{R}}

\newcommand{\sizes}{\vec{s}}

\newcommand{\extend}[3]{\ensuremath{\textnormal{\textsc{Extend}}(#1,#2,#3)}\xspace}
\renewcommand{\LLL}[1]{\ensuremath{L'}^{(#1)}}
\renewcommand{\RRR}[1]{\ensuremath{R'}^{(#1)}}
\renewcommand{\EEE}[1]{\ensuremath{E'}^{(#1)}}

\newcommand{\bLLL}[1]{\ensuremath{\bL}'^{(#1)}}
\newcommand{\bRRR}[1]{\ensuremath{\bR}'^{(#1)}}
\newcommand{\bEEE}[1]{\ensuremath{\bE}'^{(#1)}}
\newcommand{\hbG}{\ensuremath{\widehat \bG}}
\newcommand{\hbL}{\ensuremath{\widehat \bL}}
\newcommand{\hbR}{\ensuremath{\widehat \bR}}
\newcommand{\hbE}{\ensuremath{\widehat \bE}}
\newcommand{\hbe}{\ensuremath{\widehat \be}}
\newcommand{\hbLL}[1]{\ensuremath{\hbL}^{(#1)}}
\newcommand{\hbRR}[1]{\ensuremath{\hbR}^{(#1)}}
\newcommand{\hbEE}[1]{\ensuremath{\hbE}^{(#1)}}
\newcommand{\hP}{\ensuremath{\widehat P}}
\newcommand{\hB}{\ensuremath{\widehat B}}
\newcommand{\hC}{\ensuremath{\widehat C}}
\newcommand{\hbw}{\ensuremath{\widehat \bw}}

\newcommand{\hsizes}{\ensuremath{\widehat \sizes}}

\newcommand{\ssize}[1]{\ensuremath{size\left(#1\right)}}
\newcommand{\rratio}[1]{\ensuremath{ratio\left(#1\right)}}
\newcommand{\sscore}[1]{\ensuremath{score\left(#1\right)}}
\newcommand{\sscores}[1]{\ensuremath{score_s\left(#1\right)}}

\section{Preliminaries}

\paragraph{Notation.} We often denote a bipartite graph as $G=(L,R,E)$ with the bipartition of vertices $L$ and $R$. We always assume $\card{L} = \card{R}$ unless explicitly stated otherwise. 
For an edge $e \in E$, we use $L(e)$ and $R(e)$ to denote the endpoint of $e$ in $L$ and $R$, respectively. 

For a tuple $x = (x_1,\ldots,x_t)$ and $i \in [t]$, we use $x_{<i}$ to denote $(x_1,\ldots,x_{i-1})$; we define $x_{>i}$ and $x_{-i}$ analogously. Moreover, for two tuples $x=(x_1,\ldots,x_a)$ and $y=(y_1,\ldots,y_b)$, 
we use $x \conc y$ to denote the concatenated $(a+b)$-tuple $(x_1,\ldots,x_a,y_1,\ldots,y_b)$. Finally, when considering tuples $(x_1,\ldots,x_t) \in [C]^t$ for some $C$ depending on the context, 
we may represent a set of tuples by using $*$-notation in some indices to denote all possible choices from $[C]$ for those indices, e.g., 
\[
(*,x_2,\ldots,x_{t-1},*) = \set{(a,x_2,\ldots,x_{t-1},b) \mid a,b \in [C]}.
\]

\paragraph{Notation for blueprints.} Recall that a \emph{blueprint} for us is a constant-size graph, which shall be defined formally in~\Cref{sec:blueprints}. 
In order to easily distinguish between blueprints and our ``actual''  graphs, throughout, we always use `blackboard bold' letters to denote blueprints and their vertices and edges, e.g., write $\bG=(\bV,\bE)$,
and use `normal' letters for other  graphs, e.g., write $G=(V,E)$.

\subsection{A Communication Game for Approximate Matching}\label{sec:cc-game}

We consider approximating matchings in the following one-way number-in-hand communication model with shared blackboard: 
\begin{ourbox}
\begin{itemize}[leftmargin=10pt]
	\item An input graph $G=(V,E)$ is edge-partitioned between $P$ players $\PS{1},\ldots,\PS{P}$. We denote the edges given to the player $p \in [P]$ by $\ES{p}$ and thus have $E = \ES{1} \sqcup \ldots \sqcup \ES{P}$. 
	\item The communication is done using a shared blackboard. First, $\PS{1}$ writes a message $\msg_1$ based on $\ES{1}$ on the shared blackboard which will be visible to all subsequent players. 
	Then, $\PS{2}$ writes the next message $\msg_2$ based on $\ES{2}$ and $\msg_1$. The players continue like this until $\PS{P}$ writes the last message $\msg_P$ which is a function of $\ES{P}$ and $\msg_{<P}$.
	\item The goal of the players is to output an approximate maximum matching of $G$ by $\PS{P}$ writing it on the shared blackboard as the message $\msg_P$. 
	\item We assume that the input of players is sampled from some distribution and against this distribution, the players are following a deterministic protocol. 
\end{itemize}
\end{ourbox}
The \textbf{communication cost} of a protocol used by the players is defined as the worst-case number of bits written by any one player on the blackboard on any input.

We also note that by Yao's minimax principle, now that we are working on a distributional input, the assumption that the protocol is deterministic is without loss of generality and any communication lower bound here 
readily extends to randomized protocols as well. 

We have the following well-known result---dating back to the introduction of the streaming model in~\cite{AlonMS96}---that relates communication cost lower bounds to streaming ones. 

\begin{proposition}\label{prop:cc-stream}
	Suppose for any $P \geq 2$, there is some distribution on $n$-vertex graphs as inputs to the $P$-player communication game of the approximate matching problem such that
	any protocol with communication cost $s(n)$ can only achieve an approximation ratio of $\alpha(n)$ with probability at most $1-\delta(n)$. Then, any streaming algorithm using $s(n)$ space
	can only achieve an approximation ratio of $\alpha(n)$ for the maximum matching problem with probability at most $1-\delta(n)$.  
\end{proposition}

\subsection{A Simple Lower Bound for Lossy Compression}\label{sec:comp-lem}

We also need a basic information theory result, which allows us---for the particular distribution of our input graphs---to effectively treat any communication protocol as a one wherein players are only communicating edges of the graph. 
In the following, think of a compression scheme as the protocol each player in a communication game runs to compute its messages based on their input and messages of prior players. Consider the following setup: 
\begin{ourbox}
\begin{itemize}[leftmargin=10pt]
	\item We have an arbitrary graph $\Gbase = (V, \Ebase)$ on a set $V$ of $n$ vertices and $m$ edges $\Ebase$. %
	\item For any $\delta \in (0,1/2)$, define the distribution $\GG := \GG(\Gbase,\delta)$ over subgraphs of $\Gbase$ by removing $\delta$ fraction of edges in $\Gbase$ uniformly at random to 
	obtain a sample $G = (V,E)$. 
	\item  A \textbf{compression scheme} of size $s \geq 1$ for $\GG$ is any mapping $\Phi : \supp{\GG} \rightarrow \set{0,1}^s$ that maps any graph $G$ in the support of $\GG$ to a $s$-bit string $\Phi(G)$. We refer 
	to each $\phi \in \set{0,1}^s$ in the image of $\Phi$ as a \textbf{summary}. 
	\item Fix a compression scheme $\Phi$. For any edge $e \in \Gbase$ and any summary $\phi \in \set{0,1}^s$, we say that the edge $e$ \textbf{belongs to} $\phi$ iff the following holds: 
	\[
		\Pr_{G \sim \GG}\paren{\text{$e$ is an edge in $G$} \mid \Phi(G) = \phi} \geq 1 - \delta/2. 
	\]
	Informally speaking, this means that an edge $e$ belongs to $\phi$ if ``a vast majority'' of the graphs mapped to $\phi$ under $\Phi$ contain the edge $e$. 
	For any summary $\phi$, we use $\belong{\Ebase}{\phi}$ to denote the set of edges in $\Ebase$ that belong to $\phi$. 
\end{itemize}
\end{ourbox}

The following result is the only information-theoretic tool we need in this paper. Roughly speaking, it states that if we sample a graph $G$ from some $\GG(\Gbase,\cdot)$ and compute its summary
using any fixed compression scheme of size $s$, then, in expectation, only $O(s)$ edges can belong to this summary; in other words, given a summary we can only ``recover'' $O(s)$ edges of 
the input graph in expectation. 

\begin{proposition}\label{prop:compression}
	Fix any graph $\Gbase$ on $m$ edges, any $\delta \in (0,1/2)$, and any $s \geq \log{(m+1)}$. Consider any compression scheme $\Phi$ of size $s$ for $\GG = \GG(\Gbase,\delta)$. 
	For $G$ sampled from $\GG$, 
	\[
		\Exp_{G \sim \GG} \card{\belong{\Ebase}{\Phi(G)}} \leq  \frac{16s}{\delta}. 
	\]
\end{proposition}
\noindent
We present the proof of this result, using standard tools, in~\Cref{app:info}.

\clearpage


\newcommand{\MIXY}[3]{\ensuremath{M_{{#1},{#2},{#3}}}\xspace}
\newcommand{\Mixy}{\MIXY{i}{x}{y}}
\newcommand{\MI}[1]{\ensuremath{M_{#1}}}
\newcommand{\Mi}{\MI{i}}

\newcommand{\weight}{\ensuremath{weight}}
\newcommand{\ccolor}{\ensuremath{color}}
\newcommand{\hweight}{\ensuremath{\widehat{\weight}}}

\section{Extended \rs Graphs}\label{sec:ers}

We start the technical part of the paper by presenting a key ingredient needed for our main results, an extremal graph family inspired by \rs (RS) graphs. 
Recall that an RS graph is a graph whose edges can be partitioned into \emph{induced} matchings, each with at least a certain size~\cite{RuzsaS78}. Here, by induced matching we mean that the subgraph induced on vertices of the matching
only contains the edges of this matching. 

We introduce a new family of graphs inspired by RS graphs, which we call \emph{ERS} graphs (short for \emph{extended} \rs graphs). It helps to provide some context first. 

\paragraph{Hamming graphs.} Consider the following \emph{bipartite hamming} graph $G=(L,R,E)$ where $L,R$ are labeled by $[C]^d$ and there is an edge from $x \in L$ to $y \in R$ iff $x,y$ differ in exactly one coordinate. 
In this graph, fixing any $i \in [d]$ partitions $L$ and $R$ into $C$ groups 
\[
L_{i,a}:= (\underbrace{*,\ldots,*}_{i-1},a,\underbrace{*,\ldots,*}_{d-i}) \quad \text{and} \quad R_{i,a} := (\underbrace{*,\ldots,*}_{i-1},a,\underbrace{*,\ldots,*}_{d-i})
\]
 for $a \in [C]$. It can be easily verified that these graphs have two very interesting properties (the reason behind our interest in these properties will become clear in~\Cref{sec:expand,sec:blueprint-lower}): 
 \begin{itemize}
 	\item For any $i \in [d]$ and $a \neq b \in [C]$, there is a perfect matching $M_{i,a,b}$ between $L_{i,a}$ and $R_{i,b}$; 
	\item For any $i \in [d]$ and $a \neq b \in [C]$, the matching $M_{i,a,b}$ is an \emph{induced} matching. 
 \end{itemize}
 The catch here is that these graphs are very sparse and can only have $\Theta(n\log{n})$ edges for any constant $C \geq 2$. Thus, we will not be able 
 to use them directly when designing our hard instances for semi-streaming algorithms (as the algorithm can just store such a graph directly in its memory). 

This is where our definition of ERS graphs comes into picture. A reader familiar with \rs graphs, specifically the ones of~\cite{FischerLNRRS02,GoelKK12}, may know that they can be seen as a somewhat-dense \emph{relaxation} of bipartite hypercube graph (namely, a hamming graph with $C=2$). In the same vein, we define our ERS graphs as a relaxation of Hamming graphs; this relaxation allows us to create a somewhat-dense version of them in~\Cref{prop:base-c-rs}; see also~\Cref{fig:base-rs}.

\begin{definition}\label{def:base-rs}
For any $r,t,C \geq 1$, a \textnormal{\textbf{$(C,r,t)$-ERS graph}} is a bipartite graph $G=(L, R, E)$ with $t \cdot C^2$ matchings indexed by $\Mixy$ for $i \in [t]$ and $x,y \in [C]$ satisfying the following conditions (throughout, we define $\Mi := \bigcup_{x,y \in [C]} \Mixy$ as a \textbf{matching-group}): 

\begin{enumerate}
    \item \textnormal{\textbf{Size:}} For each $i \in [t]$ and $x, y \in [C]$, we have $\card{M_{i,x,y}} = r$.
    \item \textnormal{\textbf{Disjointness:}} For any $i \neq j \in [t]$, we have $M_i \cap M_j = \emptyset$.
    \item \textnormal{\textbf{Decomposition:}} For each $i \in [t]$, there exist $C$ disjoint groups
    \[
   L_i := L_{i,1} \sqcup \ldots \sqcup L_{i,C} \qquad \text{and} \qquad R_i := R_{i,1} \sqcup \ldots \sqcup R_{i,C}
    \]
    where for each $x, y \in [C]$, $M_{i,x,y}$ is a matching between $L_{i,x}$ and $R_{i,y}$. %
    \item \textnormal{\textbf{Inducedness:}} For any $i \in [t]$, any edge of $E \setminus M_i$ that is between vertices $L_i$ and $R_i$ of $M_i$ can only be between some pair $L_{i,x}$ and $R_{i,x}$ for $x \in [C]$.
\end{enumerate}
\end{definition}

\begin{figure}[h!]
	\centering
	\subcaptionbox{A matching-group $M_i$ consisting of disjoint matchings $M_{i,x,y}$ for $x,y \in [4]$, each denoted by a single line.}%
	[.48\linewidth]{
		\includegraphics[scale=0.4]{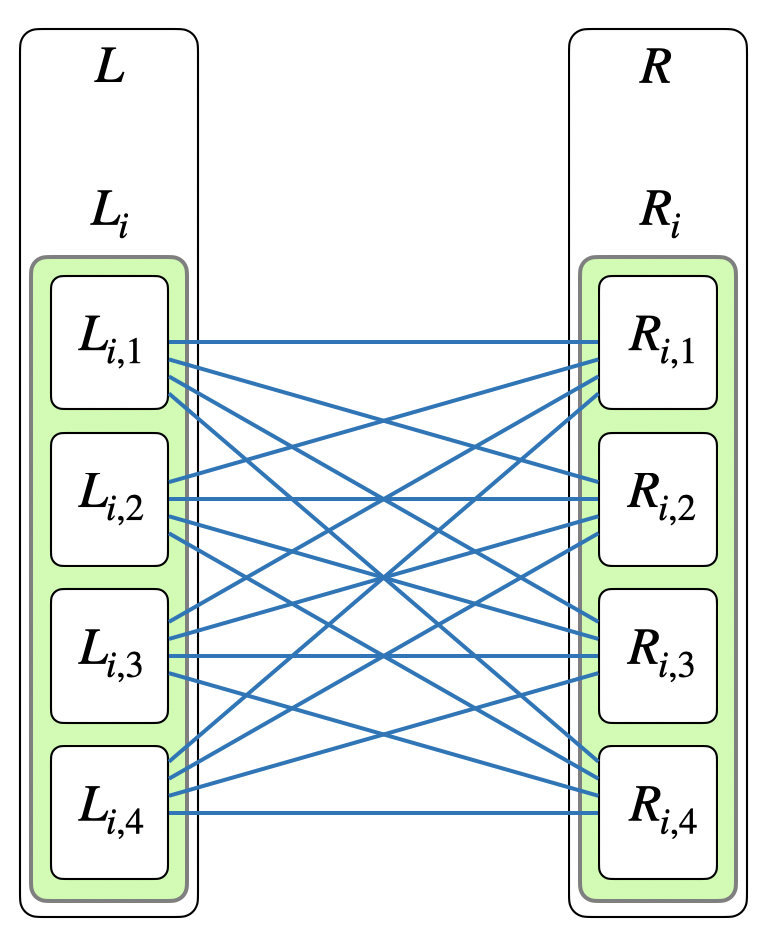}
	} 
	\hspace{0.2cm} 
	\subcaptionbox{All other edges of $G \setminus M_i$ between $L_i$ and $R_i$, denoted by dashed lines.}%
	[.48\linewidth]{
	\includegraphics[scale=0.4]{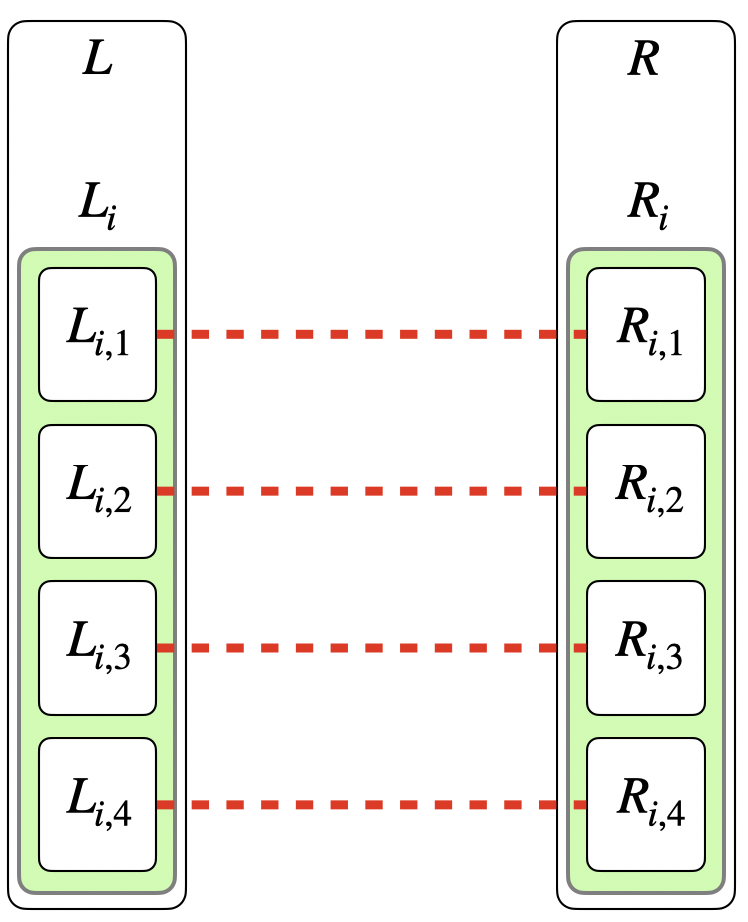}
	}
	\caption{An illustration a $(4,r,t)$-ERS graph and a single matching-group $M_i$ inside it.}
	\label{fig:base-rs}
\end{figure}

We emphasize that in~\Cref{def:base-rs} each matching-group $M_i$ is not a matching but rather a union of $C^2$ disjoint matchings, and moreover---unlike Hamming graphs mentioned earlier---each 
$M_{i,x,y}$ is \emph{not} necessarily a perfect matching between $L_{i,x}$ and $R_{i,y}$. 
We also note that these ERS graphs are \emph{not} \rs graphs\footnote{Although a $(C,r,t)$-ERS graph can be turned into a \rs graph by removing
matchings $M_{i,x,x}$ for any $x \in [C]$: any remaining matching $M_{i,x,y}$ for $x \neq y \in [C]$ will now be an induced matching in the entire graph.} but satisfy the following similar property: each $M_i$ is a group of large matchings instead of a single matching; and, the induced subgraph of 
$G$ on vertices of $M_i$ may contain (many) other edges, but those edges can only be between certain pairs of vertices. 

We present a construction of $(C,r,t)$-ERS graphs for all constant values of $C$, using ideas inspired by existing RS graph constructions in~\cite{FischerLNRRS02,GoelKK12}.

\begin{theorem}\label{prop:base-c-rs}
	Fix any integer $C \geq 2$ and parameter $\delta \in (0,1/(40C))$. There exists an integer $n_0 := n_0(C,\delta)$ such that for infinitely many values of $n \geq n_0$, 
	there are $(C,r,t)$-ERS graphs with $n$ vertices on each side of the bipartition with parameters $t = n^{\Omega(1/\log\log{n})}$ and $r = n/C - \delta \cdot n$. 
\end{theorem}

We note that in~\Cref{prop:base-c-rs}, we think of $C$ and $\delta$ as being constants (albeit $C$ can be arbitrarily large and $\delta$ can be arbitrarily small), and then let $n$ go to infinity and thus we hide dependence on $C$ and $\delta$ in the asymptotic notation. We should also point out that in the ERS graphs of~\Cref{prop:base-c-rs}, each matching-group spans all but $\delta$ fraction of vertices of the graph. 

The proof of~\Cref{prop:base-c-rs} is orthogonal to the rest of our ideas in the paper and follows the RS graph constructions of~\cite{FischerLNRRS02,GoelKK12}. As such, we postpone 
this proof to~\Cref{app:ers}, but we note that it does \emph{not} follow directly from~\cite{FischerLNRRS02,GoelKK12} and requires further modifications. 

\medskip
Before moving on, we note that in recent years, multiple other variants and generalizations of RS graphs have also been introduced in streaming and related areas, 
including \emph{lopsided RS graphs}~\cite{GoelKK12}, \emph{perfectly matchable RS graphs}~\cite{KonradN21}, \emph{full RS graphs}~\cite{PyneSW25}, \emph{ordered RS graphs}~\cite{BehnezhadG24,Pratt25}, \emph{cluster packing graphs}~\cite{AssadiKNS24,AssadiSY26}, and \emph{matching contractors}~\cite{AssadiBLLW25}. Our ERS graphs are orthogonal to all these and extend RS graphs in yet another direction (the techniques behind their constructions also seem to be closest to
those of the original RS graph constructions in~\cite{FischerLNRRS02,GoelKK12} compared to these more recent variants).

\clearpage


\section{Blueprints}\label{sec:blueprints}

We now present the central concept of our work: \emph{blueprints}. %
In the following, we first define these blueprints formally and then show how we can ``expand'' them into larger graphs. We conclude this section by providing some additional properties and equivalences for blueprints.

\subsection{Blueprints Definition}\label{sec:blueprints-def}

We start by defining the blueprints formally below. We note that the definition of blueprints may seem daunting at first due to the number of parameters involved.
To build intuition, we encourage the reader to first consider the case where $B_L = B_R = 1$, $\sizes_L$ and $\sizes_R$ are the all-ones tuples, and the fractional assignment maps every edge to one;
this is what we call a \emph{simple} blueprint below\footnote{We will later show (in~\Cref{sec:equiv-blueprints}) that
such blueprints can capture the entire power of all blueprints and thus without loss of generality, one can always work with those; that being said, specifying blueprints at that level will be quite cumbersome for
our constructions and we need the flexibility of the more general, albeit daunting, definition.}.

\begin{definition}\label{def:blueprint}
    For integers $P,C,B_L,B_R \geq 1$ and tuples $\sizes_L,\sizes_R$ of dimensions $B_L,B_R$, respectively,
    we define a \textbf{blueprint} with parameters $(P,C,B_L,B_R,\sizes_L,\sizes_R)$ as any bipartite graph $\bG=(\bL,\bR,\bE,\bw)$ with a fractional assignment $\bw$ to edges $\bE$
    with the following properties:
    \begin{itemize}
    \item Vertices in $\bL$ (resp. $\bR$) are partitioned into $B_L$ (resp. $B_R$) equal-size \textbf{blocks}:
    \[
    	\bL = \bLL{1} \sqcup \cdots \sqcup \bLL{B_L}, \quad \bR = \bRR{1} \sqcup \cdots \sqcup \bRR{B_R},
    \]
    where a block consists of $C^P$ vertices labeled by $[C]^P$. Specifically, for $b \in [B_L]$, $b' \in [B_R]$, we denote these vertices as:
    \begin{align*}
  	  \bLL{b} = \set{\bLL{b}(x) \mid x \in {[C]^P}} \quad \text{and} \quad \bRR{b'} := \set{\bRR{b'}(x) \mid x \in {[C]^P}}.
    \end{align*}
    We refer to $\sizes_L(b)$ (resp. $\sizes_R(b')$) as the \textbf{relative size} of $\bLL{b}$ (resp. $\bRR{b'}$).
    \item Edges of $\bG$ are also partitioned into $P$ groups:
    $
      \bE = \bEE{1} \sqcup \cdots \sqcup \bEE{P}.
    $
    \end{itemize}
    We say a blueprint is \textbf{simple} iff $B_L = B_R = 1$, $\sizes_L = \sizes_R = (1)$, and $\bw(\be) = 1$ for all $\be \in \bE$.  Thus, a simple blueprint can be identified
    with $\bG=(\bL,\bR,\bE)$ and $(P,C)$. Moreover, we denote vertices of a simple blueprint directly with $\bL(x),\bR(x)$ for $x \in [C]^P$ (i.e., drop the superscript). 
    
    \noindent
    \textbf{Note:} We require parameters and fractional assignments of blueprints to be \underline{rational} numbers. 
\end{definition}

These definitions may sound quite abstract at this point, but let us provide a brief intuition. The key parameter of a blueprint is $P$: a
 blueprint with parameter $P$ is designed for a $(P+1)$-player communication game. The remaining parameters are all there to provide further flexibility in the design of the blueprint and its eventual expansion into an actual graph (roughly speaking, and
 not that accurately, each pair $\bLL{b},\bRR{b'}$ will be expanded into $P$ many $(C,\cdot,\cdot)$-ERS graphs defined in \Cref{def:base-rs} using the edges of $\bG$ between these two blocks).

In order for us to be able to use a blueprint, we need it to satisfy additional properties captured by the following definition. We again encourage the reader, for the sake of intuition, to first
consider simple blueprints in the following definition which in particular simplifies the fractional assignment $\bw$ below to an integral matching that is the same as $\bE$.

\begin{definition}\label{def:proper-blueprint}
	We say a blueprint $\bG=(\bL,\bR,\bE, \bw)$ with parameters $(P,C,B_L,B_R,\sizes_L,\sizes_R)$ is \textbf{proper} if it additionally satisfies the following constraints:
    \begin{itemize}
        \item \textnormal{\textbf{Matching constraint:}} For any $b \in [B_L]$ (resp. $b' \in [B_R]$) and any vertex $\bu \in \bLL{b}$ (resp. $\bv \in \bRR{b'}$), either $\bu$ (resp. $\bv$) is
        \underline{not} incident on any edge of $\bE$ \underline{or}
        \[
        		\sum_{\be \ni \bu} \bw(\be) = \sizes_L(b) \qquad \text{(resp.} \qquad \sum_{\be \ni \bv} \bw(\be) = \sizes_R(b')~).
        \]
        In other words, the assignment $\bw$ forms a fractional matching wherein each vertex is either not matched at all or is matched equal to the relative size of the block it belongs to.

        \item \textnormal{\textbf{Ban constraints:}} For any $b \in [B_L],b' \in [B_R]$ and $x,y \in [C]^P$ and vertices $\bLL{b}(x)$ and $\bRR{b'}(y)$ incident on the same edge $\be \in \bEE{p}$ for some $p \in [P]$:
	\begin{quote}
		for any $z \in [C]^{P-p+1}$, at least one of vertices $\bLL{b}(x_{<p} \conc z)$ or $\bRR{b'}(y_{<p} \conc z)$ has no edges in $\bG$; we
		say these two vertices are \textbf{banned together} by the edge $\be$.
	\end{quote}
    \end{itemize}
\end{definition}

We note that the ban constraint is the most important part of the definition of a proper blueprint and yet we need to postpone explaining the intuition behind it to later when
we will be in a better position to say something intelligible about it. But, see~\Cref{fig:bans} for an example of ban constraints.

\begin{figure}[h!]
	\centering
	\subcaptionbox{An edge in $\bEE{1}$ will ban all $(\bLL{1}(x),\bRR{1}(x))$ pairs.}%
	[.48\linewidth]{
		\includegraphics[scale=0.35]{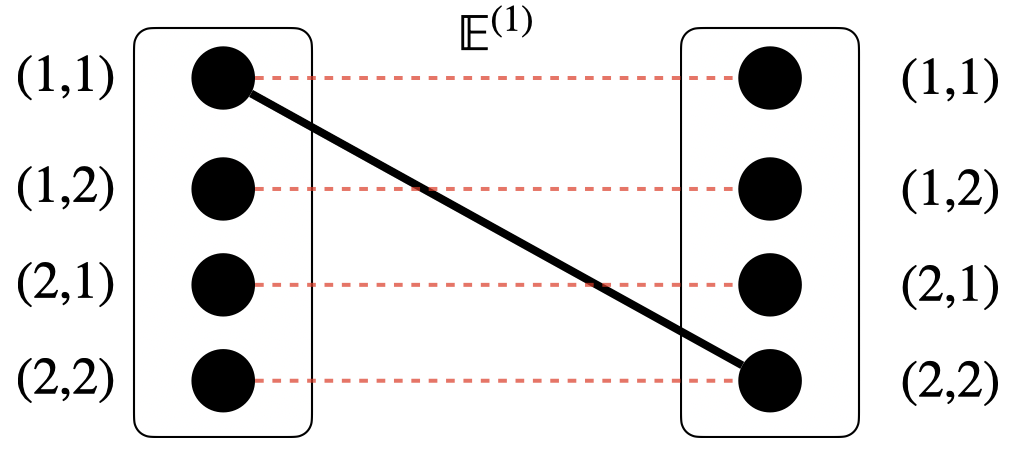}
	}
	\hspace{0.2cm}
	\subcaptionbox{The same edge in $\bEE{2}$ will ban a different but half the fraction of pairs.}%
	[.48\linewidth]{
	\includegraphics[scale=0.35]{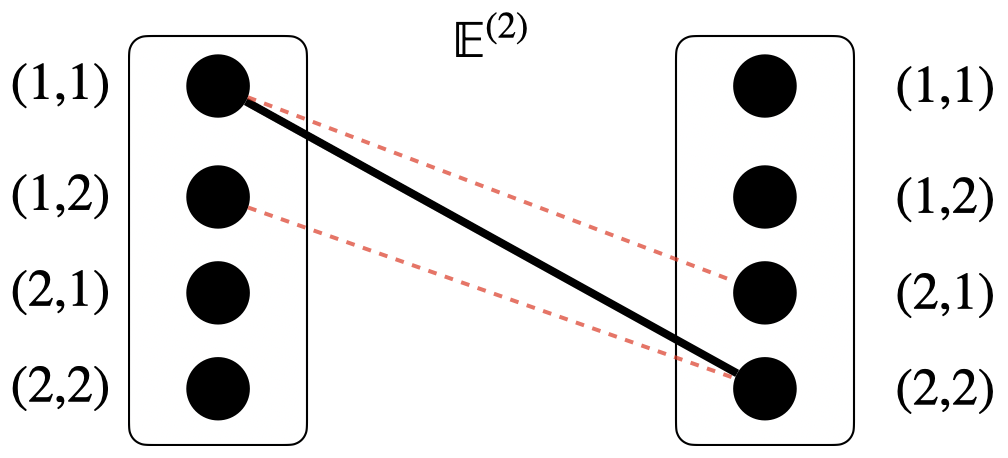}
	}
	\caption{A simple blueprint with parameters $P=2$ and $C=2$. Solid (black) lines show the single edge in each blueprint and dashed (red) lines show the pairs of vertices banned by the edge.}
	\label{fig:bans}
\end{figure}

Finally, the following quantities determine how ``good'' a blueprint is for our purpose.
\begin{definition}\label{def:blueprint-value-approx}
	Define the \textnormal{\textbf{value}} of a blueprint $\bG = (\bL,\bR,\bE,\bw)$ as:
 	 \[
    		value(\bG) := \frac{2 \cdot \sum_{\be \in \bE} \bw(\be)}{C^{P} \cdot \paren{\sum_{b \in [B_L]} \sizes_L(b) + \sum_{b' \in [B_R]} \sizes_R(b')}}.
    	\]
	For simple blueprints, this translates to $value(\bG) = \card{\bE}/\card{\bL} = \card{\bE}/\card{\bR} = \card{\bE}/C^P$.

	\smallskip
	\noindent
	Similarly, define the \textnormal{\textbf{approximation (ratio)}} of the blueprint as:
 	 \[
 		   \alpha(\bG) := \frac{2-2\,value(\bG)}{2-value(\bG)}.
    	\]
\end{definition}
We note that since all parameters and fractional weights of blueprints are rational numbers, values and approximation ratios are also always rational. 

The approximation ratio of a blueprint directly  translates to the eventual approximation ratio lower bound we can prove using this blueprint; thus, after establishing how to prove lower bounds using blueprints,
our task simply becomes finding proper blueprints with smallest approximation ratio, or equivalently, largest value.
See~\Cref{fig:blueprints} for examples of some proper blueprints. 

\medskip

\begin{figure}[h!]
	\centering
	\subcaptionbox{A simple blueprint with $P=1$ and $C=2$ with value $1/2$ and approximation ratio $2/3$. This blueprint can be expanded to reconstruct the lower bound of~\cite{GoelKK12}.\label{fig:2/3-blueprint}}%
	[.48\linewidth]{
		\includegraphics[scale=0.4]{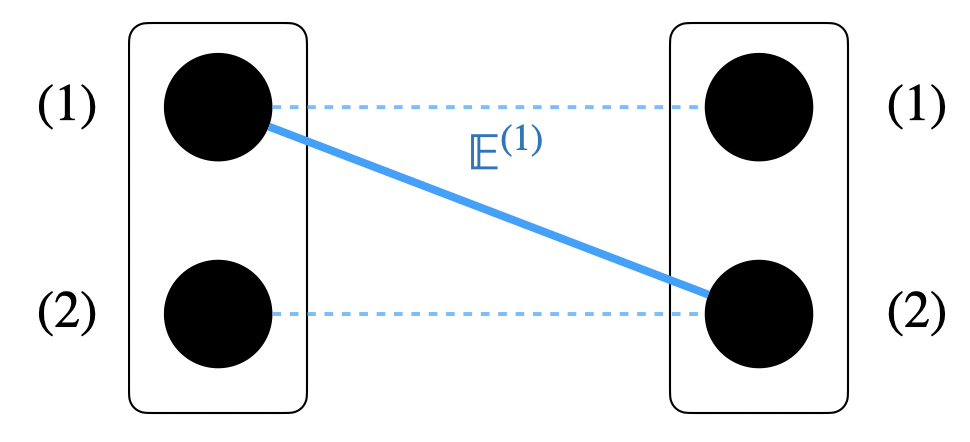}
	}
	\hspace{0.2cm}
	\subcaptionbox{A blueprint with $P=B_L=B_R=C=2$ and $s_L = s_R = (1,1)$ which has approximation ratio $2/3$.}%
	[.48\linewidth]{
	\includegraphics[scale=0.3]{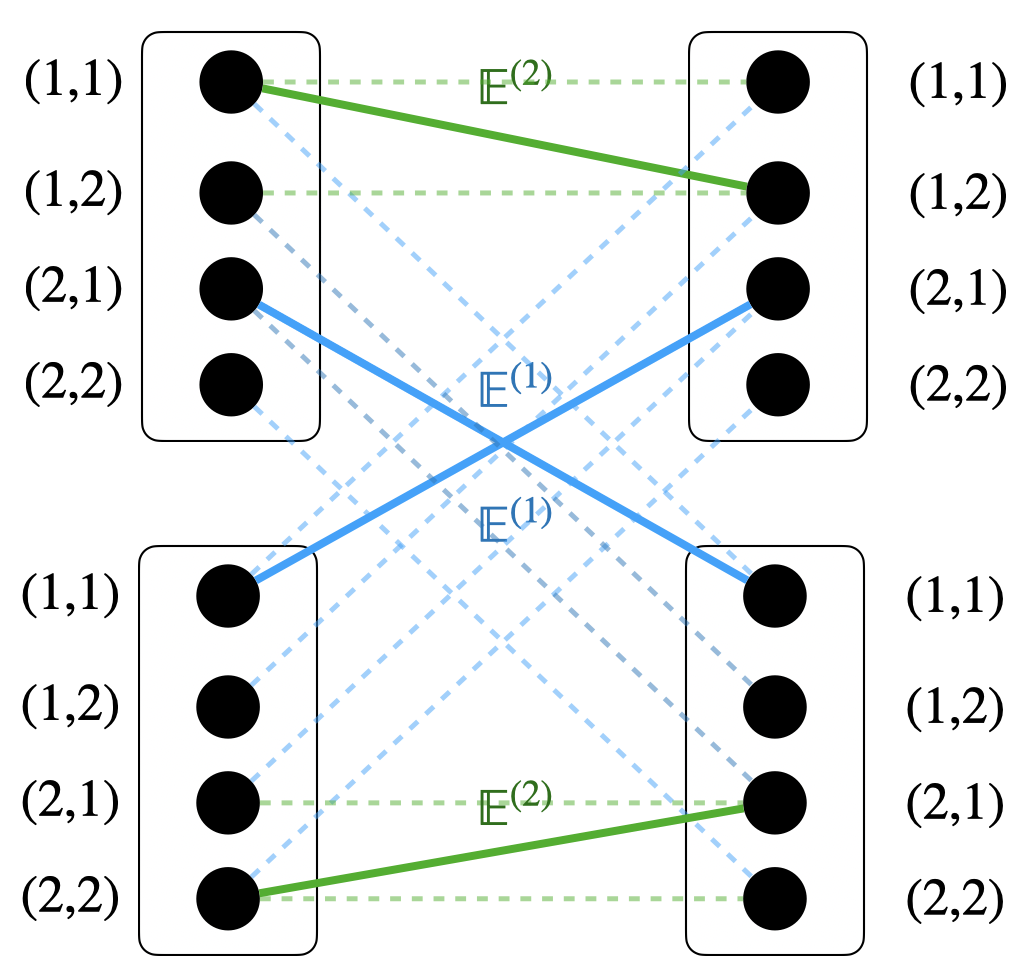}
	}
	
	\vspace{0.5cm}
	\subcaptionbox{
	A blueprint with $P=B_L=B_R=2$, $C=4$, and $s_L = s_R = (1,2/3)$. The fractional assignment on each of edges of both $\bEE{1}$ and $\bEE{2}$ is $1/6$.
Here, instead of drawing all $16$ vertices in each block, the tuples next to each part lists indices of vertices in $[C]^P = [4]^2$ belonging to that part. The value of
this blueprint is $(4 \cdot 1/8 + 3/8) / (1+2/3) = 21/40$ and the approximation ratio is $38/59 \sim 0.644$. This blueprint can be used to prove a better-than-$2/3$ approximation lower bound
for semi-streaming matching.\label{fig:blueprint3}}%
	[.9\linewidth]{
	\includegraphics[scale=0.4]{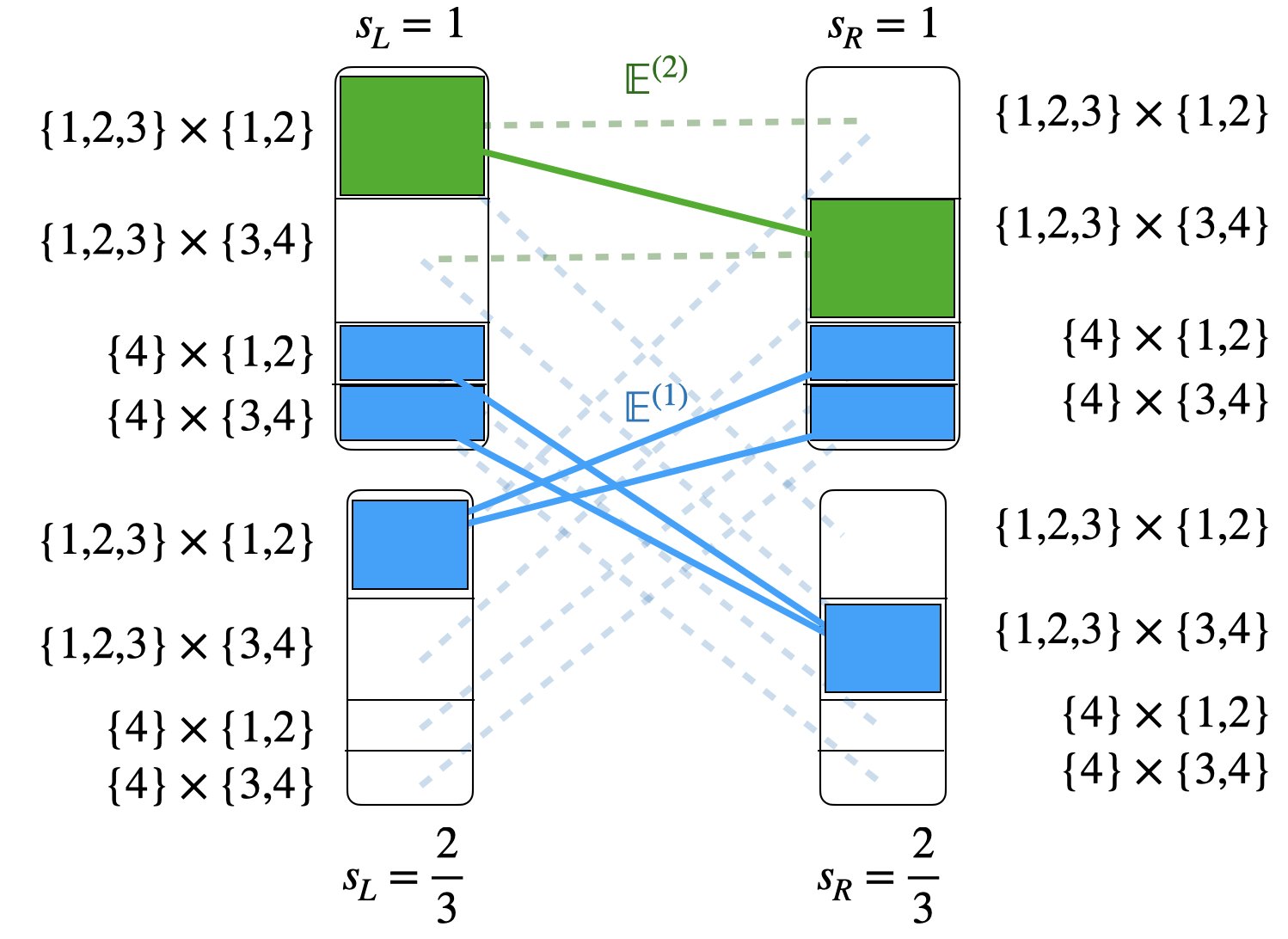}
	}
	\caption{Examples of different proper blueprints, their edges (solid lines) and the pairs each edge bans (dashed lines). One can verify that all blueprints are maximal, namely, no other edge can be added to them without violating one of ban or matching constraints.}
	\label{fig:blueprints}
\end{figure}

\subsection{Expanding Simple Proper Blueprints}\label{sec:expand}

We now describe how to \emph{expand} a proper blueprint, namely, turn it into a ``large'' graph (with number of vertices that goes to infinity for any fixed choice of parameters of the blueprint)
with the help of ERS graphs defined earlier. 

We will solely focus on expanding \emph{simple} proper blueprints. Our techniques can also be used to expand any proper blueprint (not only simple ones) but that generalization is quite tedious. 
Instead, in~\Cref{sec:equiv-blueprints}, we show that any proper blueprint can be turned into a simple one without any loss in its value (but different parameters $P$ and $C$);
thus, we can expand any non-simple proper blueprint by first turning it simple and then using the results in this subsection.

\paragraph{A high-level connection between blueprints and ERS graphs.} To motivate our expansion, let us relate a simple proper blueprint $\bG$ with parameters $P=1$ and $C \geq 2$ to a $(C,r,t)$-ERS graph. 
For any matching-group $M_i$ of the ERS graph, when looking at the decomposition of $M_i$, we ``see'' a $[C] \times [C]$ grid that can be thought of as vertices of $\bG$. 
Then, the ban constraint of the blueprint coincides nicely with the inducedness property of the ERS graphs: the pairs of vertices banned by each other in $\bG$ 
correspond precisely the pair of groups, for every $i \in [t]$, in the ERS graph that may have an edge between them from outside of edges in $M_i$, namely, pairs that 
do \emph{not} satisfy an inducedness property in the ERS graph. 

Now, suppose for any $i \in [t]$, we only pick matchings of $\Grs$ that correspond to the edges of the blueprints between the groups. Then, whenever we look at the decomposition of a matching-group $M_i$, 
we ``see'' a copy of the blueprint (with edges of $\bG$ replaced by matchings of $\Grs$) with only extra ``noisy'' edges between the banned pairs. In our lower bound instances in~\Cref{sec:blueprint-lower}, we will show that 
one can handle these noisy edges, for a randomly chosen $i^* \in [t]$, through an external gadget\footnote{By basically connecting the endpoint of these pairs of groups that correspond to a singleton vertex in the blueprint to
some external vertices. This ensures that in \emph{every large} matchings of the instance, 
these vertices always match to external ones and can be effectively ignored in the original graph.}. 
This allows us to essentially ignore these noisy edges and see our hard instance as one ``large copy'' of the blueprint -- the edges of the blueprint will correspond precisely to the edges of the graph
that a low space algorithm will not be able to find on our hard instances.

\medskip

We now formally define the expansions. In addition to the blueprint, our expansions use a ``host'' graph---which is an ERS graph defined in~\Cref{def:base-rs}---as well as an auxiliary tuple.

\begin{definition}\label{def:expansion}
	Let:
	\begin{itemize}
	\item $\bG = (\bL,\bR,\bE)$ be a simple {proper} blueprint with parameters $(P,C)$;
	\item $\Grs = (\Lrs,\Rrs,\Ers)$ be a $(C,r,t)$-ERS graph for some parameters $r,t$;
	\item $J = (J_1,\ldots,J_P) \in [t]^P$ be a tuple.
	\end{itemize}
	We define the \textbf{expansion} of $\bG$ with respect to $(\Grs,J)$ as the graph $G=(L,R,E)$ such that:
	\begin{itemize}
		\item Vertices of $G$ are $L := ({\Lrs})^P$ and $R := ({\Rrs})^P$;
		\item Edges of $G$ are partitioned into $P$ sets $E := \EE{1} \sqcup \ldots \sqcup \EE{P}$ to be defined in~\Cref{def:expansion-edge}.
	\end{itemize}
	We emphasize that vertices of $G$ do not depend on the choice of $J$.
\end{definition}

The main part of the construction is to define the edges of $G$, which requires some auxiliary definitions first. We first make 
some general definitions about the edges \emph{without} respect to $J$ (hence, in the following, we instead use $K \in [t]^P$ to avoid confusion with $J$). We emphasize that
not all these edges will appear in $G$ and for now, we are simply defining certain \emph{types} of edges.

\begin{itemize}[leftmargin=10pt]
	\item Fix some edge $\be$ of the blueprint, say, $\be = (\bL(x),\bR(y))$ for $x,y \in [C]^P$.
	Moreover, let $K \in [t]^P$ be any tuple. Define:
	\begin{align}
		E(\be,K) := &~\text{all edges $((\Lrs(e_1), \ldots, \Lrs(e_P)), (\Rrs(e_1), \ldots, \Rrs(e_P)))$} \notag \\
		& \qquad \text{where} \qquad e_{p} \in M_{K_{p},x_{p},y_{p}}~\text{for all}~p \in [P]. \label{eq:E-K}
	\end{align}
	(Recall that $\Lrs(e)$ and $\Rrs(e)$ for $e \in \Grs$ denote the endpoints of $e$ in $\Lrs$ and $\Rrs$.)

	Note that here the choice $K$ dictates which matching-group will be picked from each $\Grs$ and the choice of $\be$ dictates the choice of matchings inside these matching groups.
\end{itemize}
\noindent
We claim that $E(\be,K)$ is always a matching.
 \begin{claim}\label{clm:same-edge-matching}
For any choice of $\be \in \bE$ and $K \in [t]^P$, $E(\be, K)$ is a matching.
\end{claim}
  \begin{proof}
        Any two separate edges $e=(e_1,\ldots,e_P),f=(f_1,\ldots,f_P) \in E(\be, K)$ must have at least one value of $p \in [P]$ in which $e_{p} \neq f_p$ (otherwise $e$ and $f$ will be the same edge).
        Since $M_{K_p,x_p,y_p}$ is a matching (where $(x,y)$ is defined by $\be$), we obtain that $e$ and $f$ are vertex disjoint. \Qed{clm:same-edge-matching}
        
 \end{proof}

Using~\Cref{clm:same-edge-matching}, we can next prove that in fact for a fixed $K \in [t]^P$, the entire
set of edges $E(\be,K)$ for $\be \in \bE$ will be a matching.

\begin{claim}\label{clm:different-edge-matching}
	For any choice of $K \in [t]^P$, the set of edges $\bigcup_{\be \in \bE} E(\be,K)$ is a matching.
    \end{claim}
    \begin{proof}
    It suffices to prove that for any pairs of edges $\be \neq \bff \in \bE$, the vertices of $E(\be,K)$ and $E(\bff,K)$ are vertex-disjoint. The rest follows immediately from~\Cref{clm:same-edge-matching} as a union of vertex-disjoint matchings is
    also clearly a matching.

    Consider any pairs of edges $e=(e_1,\ldots,e_P) \in E(\be,K)$ and $f = (f_1,\ldots,f_P) \in E(\bff,K)$. We only prove that $e$ and $f$ do not intersect in $L$, i.e., $L(e) \neq L(f)$; the other case for $R$ can be proven by symmetry.
    Let $\bL(x_e)$ and $\bL(x_f)$ be the endpoints of $\be$ and $\bff$ in $\bL$ respectively. By the matching constraint of proper blueprints (\Cref{def:proper-blueprint}), we know that these
    two vertices should be different. Thus, we should have $x_e \neq x_f$ and suppose they differ on coordinate $p \in [P]$.
    Then, by the decomposition condition of $\Grs$ (\Cref{def:base-rs}), we also have that $L(e)$ and $L(f)$ are disjoint on the $p$-th coordinate. \Qed{clm:different-edge-matching}
    
    \end{proof}

  Using the definition of edges in~\Cref{eq:E-K}, we can now further define the following types of edges that also take into account the partitioning $\EE{1} \sqcup \ldots \sqcup \EE{p}$ of $E$ described in~\Cref{def:expansion}:
  \begin{itemize}[leftmargin=10pt]
  	\item For any $K \in [t]^P$, we define
	\begin{align}
		\EE{p}(K) := \bigcup_{\be \in \bEE{p}} E(\be,K), \label{eq:p-K}
	\end{align}
	namely, $\EE{p}(K)$ collects all edges $E(\be,K)$ for all $\be$ in the $p$-th partition $\bEE{p}$ of the blueprint.
	\item For any $p \in [P]$ and $K_{<p} \in [t]^{p-1}$, we further define
	\begin{align}
		\EE{p}(K_{<p}) := \bigcup_{\substack{K' \in [t]^P \\ K'_{<p}=K_{<p}}} \EE{p}(K'), \label{eq:p-K<p}
	\end{align}
	namely, $\EE{p}(K_{<p})$ collects every $E(\be,K')$ for $\be \in \bEE{p}$ and $K'$ with the same $(p-1)$-prefix as $K$.
  \end{itemize}

  We now show that edges produced by different choices of $(p, K)$ cannot collide.

\begin{claim}\label{clm:disjoint-partition}
	For any $(p, K) \neq (p', K')$ with $p,p' \in [P]$ and $K, K' \in [t]^P$, the edge sets $\EE{p}(K)$ and $\EE{p'}(K')$ are disjoint.
\end{claim}
\begin{proof}
	Consider any edge $e = (e_1,\ldots,e_P) \in E(\be, K)$ for some $\be \in \bEE{p}$. By~\Cref{eq:E-K}, at each coordinate $q \in [P]$, the ERS edge $e_{q}$ belongs to the matching-group $M_{K_{q}}$. By the disjointness condition of $\Grs$ (\Cref{def:base-rs}), $e$ can only belong to $E(\bff, K')$ for some $\bff$ if $K = K'$. By~\Cref{clm:different-edge-matching}, this further implies $\be = \bff$. Since $\bE = \bEE{1} \sqcup \ldots \sqcup \bEE{P}$ is a partition, $\be \in \bEE{p}$ and $\bff \in \bEE{p'}$ with $\be = \bff$ forces $p = p'$. \Qed{clm:disjoint-partition}
\end{proof}

  With these definitions and claims at hand, we can now conclude the definition of the edges that appear in the expansion.

\begin{definition}[Continuation of~\Cref{def:expansion}]\label{def:expansion-edge}
	For any $p \in [P]$, the edge-set $\EE{p}$ in the graph $G=(L,R,\EE{1} \sqcup \ldots \sqcup \EE{P})$ which is an expansion of $\bG$ with respect to $(\Grs,J)$ is
	\begin{align*}
		\EE{p} :=  \EE{p}(J_{<p}).
	\end{align*}
\end{definition}

We note that the partition $\EE{1} \sqcup \ldots \sqcup \EE{P}$ is indeed disjoint by~\Cref{clm:disjoint-partition}.
It is worth reiterating that the edges in $\EE{p}$ only depend on the $(p-1)$-prefix of $J$; put differently, even after fixing $\EE{p}$, the indices $J_{p},\ldots,J_{P}$ can still be chosen independently.

\bigskip

Finally, we define a \emph{canonical matching} in the expansions which plays a crucial role in our lower bound proofs -- roughly speaking, this is the ``only'' input matching the players need to find but it is well ``hidden'' inside
the expansion.

\begin{definition}\label{def:canonical}
	For an expansion $G$ of $\bG$ with respect to some $(\Grs,J)$, we define the \textbf{canonical matching} of $G$ as
	\[
		\Mstar := \Mstar(G) = \bigcup_{p=1}^P \EE{p}(J).
	\]
\end{definition}
It is easy to see $\Mstar$ indeed belongs to the graph $G$ given that $\EE{p}(J) \subseteq \EE{p}(J_{<p})$ by~\Cref{eq:p-K<p} and that $\EE{p}(J_{<p})$ is in $G$ by~\Cref{def:expansion-edge}.
In the following, we present the main properties of the canonical matching $\Mstar$. We note that the ban constraint of blueprints is only defined to satisfy the third property (weak inducedness) below.

\begin{proposition}\label{prop:Mstar}
	Let $G=(L,R,E)$ be the expansion of $\bG$ with respect to $(\Grs,J)$ where $\bG$ is a simple proper blueprint with parameters $(P,C)$, $\Grs$ is a $(C,r,t)$-ERS graph, and $J$ is in $[t]^P$.
	The canonical matching $\Mstar$ of $G$ satisfies the following properties:
	\begin{itemize}
	\item \textnormal{\textbf{Matching:}} the edges in $\Mstar$ form a matching and $\Mstar$ does not have duplicate edges;
	\item \textnormal{\textbf{Size:}} the number of edges of $\Mstar$ is
	\[
		\card{\Mstar} = value(\bG) \cdot \paren{\frac{r \cdot C}{\card{\Lrs}}}^P \cdot \card{L}.
	\]
	\item \textnormal{\textbf{Weak inducedness:}} let $e$ be an edge in $G$ belonging to some $\EE{p}(K)$ such that $p \in [P]$ and $K_{<p}=J_{<p}$ but for all $p' \geq p$, we have $J_{p'} \neq K_{p'}$.
	Then, $e$ cannot be part of the induced subgraph of $G$ on $\Mstar$.
	\end{itemize}
\end{proposition}
\begin{proof}
	We prove each part separately.

	\paragraph{Matching:} This follows immediately from~\Cref{clm:different-edge-matching} since
	\[
		\Mstar = \bigcup_{p=1}^P \EE{p}(J) = \bigcup_{p=1}^{P} \bigcup_{\be \in \bEE{p}} E(\be,J) = \bigcup_{\be \in \bE} E(\be,J),
	\]
	and the RHS is a matching by the claim and all the sets are disjoint and contain no duplicates.

	\paragraph{Size:} We have
	\begin{align*}
		\card{\Mstar} &= \sum_{\be \in \bE} \card{E(\be,J)} = \card{\bE} \cdot r^P = value(\bG) \cdot C^P \cdot r^P;
	\end{align*}
	here, the first equality is by the disjointness of $E(\be,J)$ for each $\be$ (\Cref{clm:different-edge-matching}), the second equality is by the definition of
	$E(\be,J)$ and since each matching of $\Grs$ is of size $r$, and the third is by the definition of $value(\bG)$ in~\Cref{def:blueprint-value-approx}. Noting that $\card{L} = \card{\Lrs}^P$ concludes the proof.

	\paragraph{Weak inducedness:} Let $e=(e_1,\ldots,e_P)$ be as in the statement and suppose $e$ belongs to $E(\be,K)$ for some $\be=(\bL(x),\bR(y))$ with $x,y \in [C]^P$.

	Suppose towards a contradiction that
	$e$ belongs to the induced subgraph of $G$ on $\Mstar$ and let $f$ and $g$ be two edges in $\Mstar$ such that $L(e) = L(f)$ and $R(e) = R(g)$. Further, let $f$ (resp. $g$) be part of some $E(\bff,J)$ with
	$\bff = (\bL(x'),*)$ and $\bg = (*, \bR(y'))$ where $*$ simply ignores the name of the corresponding vertex as it is not relevant for this proof.

	To conclude the proof, we need the following two claims.

	\begin{claim}\label{clm:x>p=y>p}
		$x'$ and $y'$ are equal to each other on all coordinates $\geq p$.
	\end{claim}
	\begin{proof}
	Consider any $p' \geq p$. By~\Cref{eq:E-K}, we have $e_{p'}$ belongs to $M_{K_{p'}}$ of $\Grs$ and $f_{p'}$ and $g_{p'}$ belong to $M_{J_{p'}}$ of $\Grs$.
	Since $J_{p'} \neq K_{p'}$, by the disjointness condition of $\Grs$ (\Cref{def:base-rs}), $e_{p'} \notin M_{J_{p'}}$, and thus by the inducedness condition of $\Grs$ (\Cref{def:base-rs}),
	we have that edges in $M_{K_{p'}}$ cannot be between $({\Lrs})_{J_{p'},x'_{p'}}$ and $({\Rrs})_{J_{p'},y'_{p'}}$ unless $x'_{p'} = y'_{p'}$. \Qed{clm:x>p=y>p}

	\end{proof}

        	\begin{claim}\label{clm:x<p=x'<p}
		$x'$ and $x$ are equal to each other on all coordinates $< p$ and similarly for $y$ and $y'$.
	\end{claim}
	\begin{proof}
        Suppose $x_{p'} \neq x'_{p'}$ for some $p' < p$ (the other case for $y$ and $y'$ is symmetric).
        We know that $J_{p'} = K_{p'}$  since $e \in \EE{p}(K)$ and thus $K_{<p} = J_{<p}$. This implies that both $e_{p'}$ and $f_{p'}$ belong to $M_{J_{p'}}$ of $\Grs$. But, if $x_{p'} \neq x'_{p'}$, by the decomposition condition of
        $\Grs$ (\Cref{def:base-rs}), then $\Lrs(e_{p'}) \neq \Lrs(f_{p'})$, contradicting our earlier assumption that $L(e) = L(f)$. \Qed{clm:x<p=x'<p}

        \end{proof}

        Finally, putting together~\Cref{clm:x>p=y>p,clm:x<p=x'<p} and our earlier definitions, we get that $x_{<p} = x'_{<p}$, $y_{<p} = y'_{<p}$, and $x'_{\geq p} = y'_{\geq p}$ and there are edges
        $\be \in \bEE{p}$, $\bff$, and $\bg$ in the blueprint with $\bL$-endpoints of $\be$ and $\bff$ labeled by $x$ and $x'$ belonging to $\bL$,
        and $\bR$-endpoints of $\be$ and $\bg$ labeled by $y$ and $y'$ belonging to $\bR$. But this is a contradiction given that the ban constraint of $\bG$ (\Cref{def:proper-blueprint}) implies that
         pairs
         \[
         \bL(x') = \bL(x'_{<p} \conc x'_{\geq p}) = \bL(x_{<p} \conc x'_{\geq p}) \quad \text{and} \quad \bR(y') = \bR(y'_{<p} \conc y'_{\geq p}) =  \bR(y_{<p} \conc x'_{\geq p})
         \]
         are banned together because of $\be$, i.e., not both of them can have a non-zero degree in $\bG$. This concludes the proof. \Qed{prop:Mstar}

\end{proof}

Finally, before moving on, we note that using our construction of ERS graphs in~\Cref{prop:base-c-rs}, we can create a $(C,r,t)$-ERS graph $\Grs=(\Lrs,\Rrs,\Ers)$ where $r = \card{\Lrs} \cdot (1/C - o(1))$. Plugging in this bound in~\Cref{prop:Mstar} implies that
\[
	\frac{\card{\Mstar}}{\card{L}} = value(\bG) \cdot \paren{\frac{r \cdot C}{\card{\Lrs}}}^P = value(\bG) \cdot \paren{1-o(1)},
\]
using the fact that $P$ is a constant with respect to the size of the graph.
Thus, matching $\Mstar$ in $G$ will essentially have the same ``value'' as that of $\bE$ in the blueprint. This will be crucial in
relating the approximation of streaming algorithms on $G$ to approximation ratio of the blueprint $\bG$.

\subsection{Equivalence of Simple and Not-Simple Proper Blueprints}\label{sec:equiv-blueprints}

In this subsection, we show that we can turn any not-simple proper blueprint into a simple one without reducing its value.

\begin{proposition}\label{prop:simplify-blueprint}
Let $\bG = (\bL,\bR,\bE,\bw)$ be a proper blueprint with parameters $(P,C,B_L,B_R,\sizes_L,\sizes_R)$.
Then, there exists a \underline{simple} proper blueprint $\hbG=(\hbL,\hbR,\hbE)$ with parameters $(P+1,\hC)$
such that $value(\hbG) = value(\bG)$.
\end{proposition}

To prove~\Cref{prop:simplify-blueprint}, we also show some other transformation on blueprints that ``simplify'' them without reducing their values. The proofs of all these lemmas are simple but not particularly short exercises
in verifying blueprint constraints and we postpone them to~\Cref{app:omitted-blueprints}. 

The first lemma shows that we can multiply the parameter $C$ of a blueprint with any integer and obtain a blueprint with the same value (and keep all other parameters intact). This will be useful in future transformations where we will want to assume certain divisibility conditions on $C$.

\begin{lemma}\label{lem:increase-C-blueprint}
Let $\bG = (\bL,\bR,\bE,\bw)$ be a proper blueprint with parameters $(P,C,B_L,B_R,\sizes_L,\sizes_R)$ and $k \geq 1$ be an integer.
Then, there exists a proper blueprint $\hbG=(\hbL,\hbR,\hbE,\hbw)$ with parameters $(P,k \cdot C,B_L,B_R,\sizes_L,\sizes_R)$ such that $value(\hbG) = value(\bG)$.
\end{lemma}

The next lemma shows that we may assume, without loss of any generality, that the relative sizes $\sizes_L$, $\sizes_R$ and the fractional assignment $\bw$ are integral.

\begin{lemma}\label{lem:integer-weights}
    Let $\bG = (\bL,\bR,\bE,\bw)$ be a proper blueprint with parameters $(P,C,B_L,B_R,\sizes_L,\sizes_R)$ and $k \geq 1$ be a rational. Then, the blueprint $\hbG$, obtained from $\bG$ by changing nothing except scaling $\bw$, $\sizes_L$, and $\sizes_R$ each by $k$, is proper and satisfies $value(\hbG) = value(\bG)$.
\end{lemma}

We are now ready to ditch two components of non-simple blueprints, the relative sizes of blocks and the assignment to edges, and obtain an essentially simple blueprint with the exception that parameters $B_L,B_R$ may still 
not be equal to one. 

\begin{lemma}\label{lem:remove-sizes-blueprint}
Let $\bG = (\bL,\bR,\bE,\bw)$ be a proper blueprint with parameters $(P,C,B_L,B_R,\sizes_L,\sizes_R)$.
Then, there exists a proper blueprint $\hbG=(\hbL,\hbR,\hbE,\hbw)$ with parameters $(P,C,\hB_L,\hB_R,\hsizes_L,\hsizes_R)$, where $\hsizes_L = (1, \ldots, 1)~,~ \hsizes_R = (1, \ldots, 1)$, and $\hbw(\hbe) = 1$ for all edges $\hbe \in \hbE$, such that 
$value(\hbG) = value(\bG)$.
\end{lemma}

Finally, armed with the reductions above, we are ready to prove \Cref{prop:simplify-blueprint}. We shall note that in all the above lemmas, the reduction still kept the parameter $P$ untouched. 
However, for the next proof, as advertised in~\Cref{prop:simplify-blueprint}, we inevitably have to increase $P$ by one. 

\begin{proof}[Proof of~\Cref{prop:simplify-blueprint}]
    Let us first assume that $B_L = B_R$. This is without loss of generality, as if $B_L \neq B_R$, we can duplicate the graph of $\bG$ and swap the left and right sides of the copy, and this new blueprint will satisfy $B_L = B_R$ and have the same value as $\bG$. Define $B = B_L = B_R$.

    By \Cref{lem:increase-C-blueprint,lem:remove-sizes-blueprint}, we may assume that $B \mid C$ and that $\sizes_L = (1, \ldots, 1)$, $\sizes_R = (1, \ldots, 1)$, and $\bw(\be) = 1$ for any $\be \in \bE$. If $\bG$ did not satisfy these conditions to begin with, we can replace $\bG$ with another proper blueprint that does satisfy these conditions and has the same value.

    Now we begin to define the simple proper blueprint $\hbG = (\hbL, \hbR, \hbE)$ with parameters $(\hP, \hC)$ that will have same value as $\bG$. The parameters we use will be $\hP = P + 1$ and $\hC = C$; we will use $C$ for $\hC$ throughout the proof. Define the integer $k := C / B$. We use integers in $[C]$ and pairs in $[B] \times [k]$ interchangeably for the ease of notation (using any arbitrary bijection). For every $p \in [P]$ define
    \begin{align*}
		\hbEE{p + 1} &:= ~\text{all edges $\paren{\hbL((i, b) \conc x), \hbR((i, b') \conc y)}$} \\
		& \qquad \text{where} \quad i \in [k], \quad x, y \in [C]^P, \quad b, b' \in [B],~\text{and} \quad (\bLL{b}(x), \bRR{b'}(y)) \in \bEE{p}.
	\end{align*}
    In addition, define $\hbEE{1} = \emptyset$. For each edge of $\bEE{p}$, we are adding $k$ edges to $\hbEE{p + 1}$. As $\hbG$ also has exactly $k$ times more vertices than $\bG$, we have $value(\hbG) = value(\bG)$.

    Under this definition, $\hbG$ is a simple blueprint, so it remains to verify that $\hbG$ is a proper.

    \begin{itemize}
        \item \textnormal{\textbf{Matching constraint:}} In a simple blueprint, the matching constraint is equivalent to asserting the edges form a matching in the standard sense. For each $\be \in \bE$, let $\hbE_{\be}$ be the set of edges in $\hbG$ that are added as a result of $\be$ in the definition above. It is easy to see $\hbE_{\be}$ is a matching in $\hbG$. Moreover, by the matching constraints of $\bG$, any distinct $\be_1, \be_2 \in \bE$ share no endpoints, which also means that $\hbE_{\be_1}$ and $\hbE_{\be_2}$ share no endpoints. Hence, $\hbE$ is a union of vertex disjoint matchings, and thus it itself is a matching.
        \item \textnormal{\textbf{Ban constraints:}} Consider any edge $\hbe \in \hbEE{p + 1}$ for $p \in [P]$. Suppose
		\[
		\hbe = (\hbL((i, b) \conc x), \hbR((i, b') \conc y)),
		\]
        where $b, b' \in [B]$ and $x, y \in [C]^P$ and $i \in [k]$. The ban constraints induced by $\hbe$ is that for any values of $z_p, \ldots, z_P \in [C]$, the vertices
        \begin{align*}
			&\hbL((i, b), x_1, \ldots, x_{p-1}, z_p \ldots, z_P) \text{ and } \\
			&\hbR((i, b'), y_1, \ldots, y_{p-1}, z_p \ldots, z_P)
		\end{align*}
        are banned together. The former vertex has an incident edge in $\hbG$ if and only if $\bLL{b}(x_{<p} \conc (z_p, \ldots, z_P))$ has an incident edge in $\bG$, and the latter vertex has an incident edge in $\hbG$ if and only if $\bRR{b'}(y_{<p} \conc (z_p, \ldots, z_P))$ has an incident edge in $\bG$. These cannot happen simultaneously due to the ban constraint in $\bG$ from $(\bLL{b}(x), \bRR{b'}(y)) \in \bE$.
    \end{itemize}
    This proves that $\hbG$ is a simple proper blueprint, as desired. \Qed{prop:simplify-blueprint}
\end{proof}

\clearpage


\section{Communication and Streaming Lower Bounds from Blueprints}\label{sec:blueprint-lower}

We are finally at a stage that can reduce proving lower bounds for the maximum matching problem to constructing blueprints. The following theorem captures our lower bound framework. 

\begin{theorem}\label{thm:framework}
     Suppose there is a proper blueprint $\bG$ with approximation ratio $\alpha(\bG)$. 
    Then, for any fixed $\delta > 0$, and any sufficiently large $n$, as a function of $\bG$ and $\delta$, the following holds. Any single-pass semi-streaming algorithm 
    for maximum bipartite matchings on $n$-vertex graphs cannot achieve an $(\alpha(\bG)+\delta)$-approximation with probability more than $1-\delta$. 
\end{theorem}

To prove~\Cref{thm:framework}, we establish a lower bound on the communication game defined for matching in~\Cref{sec:cc-game}. By~\Cref{prop:simplify-blueprint}, it suffices to focus on simple proper blueprints. 

\begin{lemma}\label{lem:framework}
    Suppose there is a \underline{simple} proper blueprint $\bG$ with parameters $(P,C)$ and approximation ratio $\alpha(\bG)$. 
    Then, for any fixed $\delta > 0$, and any sufficiently large $n$ as a function of $(P,C,\delta)$, the following holds. 
    There is a distribution on $n$-vertex graphs in the one-way $(P+1)$-player communication game for approximate matching, wherein any protocol with $n^{1+o(1/\log\log{n})}$ communication 
    cannot achieve an $(\alpha(\bG)+\delta)$-approximation with probability more than $1-\delta$. 
\end{lemma}

\Cref{thm:framework} follows immediately from~\Cref{lem:framework},~\Cref{prop:cc-stream}, and~\Cref{prop:simplify-blueprint}. In the rest of this section, we focus on proving~\Cref{lem:framework}. 

\subsection{A Hard Input Distribution}

Let $\bG$ be the proper blueprint with parameters $(P,C)$ specified in~\Cref{lem:framework}. Pick $\Grs$ to be a $(C,r,t)$-ERS graph  on $2\nrs$ vertices with parameters 
\[
	r = \nrs/C - \delta \cdot \nrs \qquad \text{and} \qquad t = {\nrs}^{1/\Omega(\log\log{\nrs})},
\]
which is guaranteed to exists for infinitely many choices of $\nrs$ by~\Cref{prop:base-c-rs}. 
We now define our hard input distribution for $(P+1)$ players as follows (see~\Cref{fig:dist-graph} also). 

\begin{ourbox}
	\textbf{A hard input distribution $\GG$ for approximate matching given a blueprint $\bG$.} 
	\begin{enumerate}
		\item Sample a tuple $J \in [t]^P$ uniformly at random. 
		\item Let $\Gbase=(L,R,\Ebase)$ be the expansion of $\bG$ with respect to $(\Grs,J)$ (\Cref{def:expansion}). 
		\item Let $G=(L,R,E)$ be a subgraph of $\Gbase$ whose edges are a union of $\EE{1} \sqcup \ldots \sqcup \EE{P}$ partitioned between the players defined in the following. 
		\item Let the input of $\PS{p}$ for $p \in [P]$ be:
		\begin{quote}
			$\EE{p}:$ a randomly chosen $(1-\delta)$ fraction of ${\Ebase}^{(p)}$ where ${\Ebase}^{(p)}$ is as defined in~\Cref{def:expansion} for the expansion $\Gbase$. 
		\end{quote}
		\item Let $\Mbase$ be the canonical matching of $\Gbase$ (\Cref{def:canonical}). Add $\card{L}=\card{R}$ new vertices to $L$ and $R$ each, denoted by $\Lnew$ and $\Rnew$. Let 
		the input of $\PS{P+1}$ be $\EE{P+1}$ which is some arbitrary matching that matches $L \setminus L(\Mbase)$ to $\Rnew$ and $R \setminus R(\Mbase)$ to $\Lnew$. %
		\item The final input graph of the players will become $G=(L \sqcup \Lnew, R \sqcup \Rnew, \EE{1} \sqcup \ldots \sqcup \EE{P+1})$. 
	\end{enumerate}
\end{ourbox}

We emphasize that the choice of the graph $\Grs$ in the distribution is fixed and deterministic and is known to all players of the communication game. 

\begin{figure}[t]
	\centering
	\includegraphics[scale=0.5]{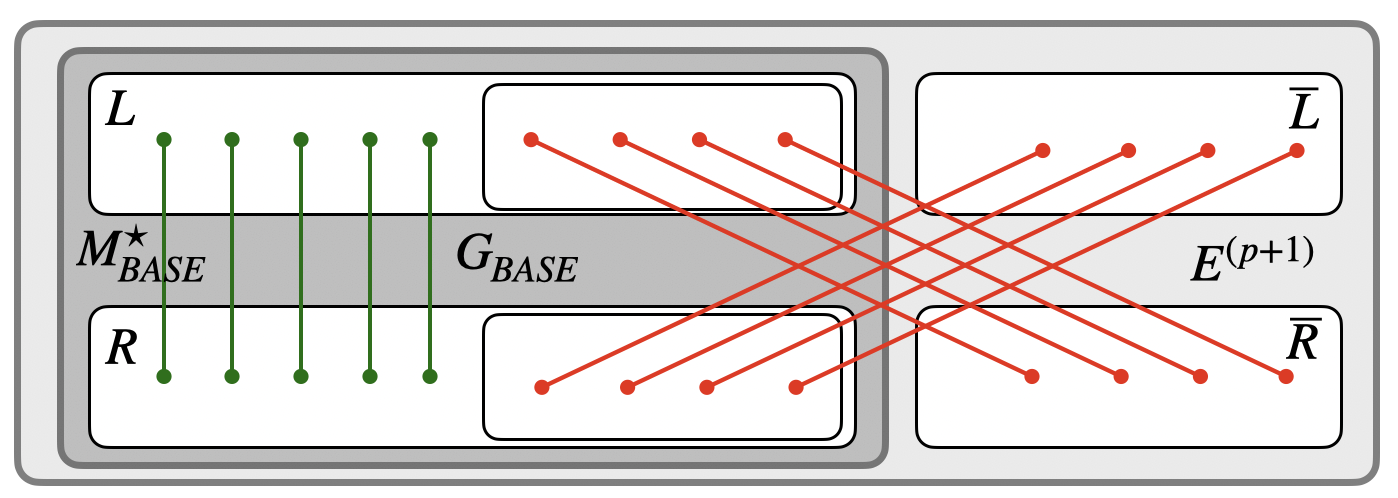}
	\caption{An illustration of the graph $G$ sampled from our hard distribution. A $\delta$-fraction of edges in $\Gbase$ will be removed in the graph $G$.}\label{fig:dist-graph}
\end{figure}

Let us first examine ``large'' matchings and their structure in graphs of the distribution $\GG$. 

\begin{lemma}\label{lem:large-matching}
	For a graph $G$ sampled from the distribution $\GG$: 
	\begin{enumerate}
		\item With probability at least $1-o(1)$, $G$ contains a matching of size at least
		\[
			\paren{2-value(\bG)} \cdot \card{L} - 2\cdot \delta \cdot \card{L}. 
		\]
		\item Any matching in $G$ that does not use any edge from the induced subgraph of $G$ on $\Mbase$ has size at most 
		\[
			\paren{2-2\cdot value(\bG)} \cdot \card{L} + P \cdot \delta \cdot \card{L}. 
		\]
	\end{enumerate}
\end{lemma}
\begin{proof}
	We prove each part separately. 
	\begin{enumerate}
		\item Let $\Mstar$ be the matching in $G$ consisting of the edges in $\Mbase$ which are still part of $G$ (so $\card{\Mbase} - 2\delta \card{\Mbase}$ edges with high probability using concentration bounds for sampling without replacement), 
		plus the matching $\EE{p+1}$ which is vertex-disjoint from $\Mbase$ by construction.  The size of this matching, with high probability, is at least 
		\begin{align*}
			\card{\Mbase} - 2\delta \card{\Mbase} + \card{\EE{p+1}} &= \card{\Mbase} - 2\delta \card{\Mbase} + 2 \cdot \paren{\card{L} - \card{\Mbase}} \\
			&= 2\card{L} - \card{\Mbase} - 2\delta \card{\Mbase} \\
			&\geq 2\card{L} - \paren{value(\bG) \cdot \paren{\frac{r \cdot C}{\card{\Lrs}}}^P \cdot \card{L}} - 2\delta \card{\Mbase} \tag{by~\Cref{prop:Mstar}} \\
			&\geq 2\card{L} - {value(\bG) \cdot \card{L}} - 2\delta \card{\Mbase}, \tag{given the value of $r$ in~\Cref{prop:base-c-rs}}
			\end{align*} 
			which implies the desired bound given $\card{\Mbase} \leq \card{L}$. 
		\item Consider the subgraph of $G$ that does not contain any edge between $L(\Mbase)$ and $R(\Mbase)$. All edges of this subgraph (including the ones in $\EE{P+1}$) are incident on vertices $L \setminus L(\Mbase)$ and $R \setminus R(\Mbase)$
		and thus the largest matching in this subgraph is of size 
		\begin{align*}
			2\card{L} - 2\card{\Mbase} &= 2\card{L} - 2\paren{value(\bG) \cdot \paren{\frac{r \cdot C}{\card{\Lrs}}}^P \cdot \card{L}} \tag{by~\Cref{prop:Mstar}}\\
			&\leq  2\card{L} - 2\paren{value(\bG) \cdot \paren{1-\delta}^P \cdot \card{L}} \tag{given value of $r$ in~\Cref{prop:base-c-rs}} \\
			&\leq 2\paren{\card{L} - value(\bG) \cdot \card{L}} + P \cdot \delta \cdot \card{L}, \tag{as $(1-x)^a \geq 1-a x$ for $x \in [0,1]$} \\
		\end{align*}
		concluding the proof. \Qed{lem:large-matching}
	\end{enumerate}
\end{proof}

As established in~\Cref{lem:large-matching}, graphs $G \sim \GG$ have a large matching and additionally, the edges in the induced subgraph of $G$ between $L(\Mbase)$ and $R(\Mbase)$
play a crucial role in \emph{every} large matchings of $G$. With this in mind, we define: 
\begin{itemize}
	\item An edge $e$ in $\Gbase$ is called \textbf{special} iff $e$ is between some vertices in $L(\Mbase)$ and $R(\Mbase)$. 
\end{itemize}
The lower bound strategy is to now prove that very few special edges will be discovered by the players in any low communication cost protocol.

Observe that---as pointed out after~\Cref{def:expansion-edge}---the set of edges ${\Ebase}^{(p)}$ for any $\PS{p}$ is only a function of $J_{<p}$ and not all of $J$; thus, 
we can meaningfully talk about the inputs of the first $p$-players with respect to only $J_{<p}$ and conditioned on this, $J_{\geq p}$ is still chosen uniformly at random. 
We use this to bound the fraction of special edges from the perspective of each player (the decision of whether an edge becomes special or not is only a function of $J$). 

\begin{lemma}\label{lem:irrelevant}
Fix any $p \in [P]$ and $J_{<p} \in [t]^{p-1}$. Let $e$ be any edge in ${\Ebase}^{(p)}$. Over a uniformly at random choice of $J_{\geq p}$,
\[
	\Pr_{J_{\geq p}}\paren{\text{$e$ is a special edge} \mid J_{<p}} \leq \frac{P}{t}.
\]
\end{lemma}
\begin{proof}
Fix any edge $e \in {\Ebase}^{(p)}$ and suppose $e \in \bEE{p}(K)$ for some $K \in [t]^P$. Given that $e$ belongs to ${\Ebase}^{(p)}$, we know that $J_{<p} = K_{<p}$ (by~\Cref{def:expansion-edge}). 
By the weak inducedness property of expansions in~\Cref{prop:Mstar}, for $e$ to become a special edge, we need to have that for some $p' \in [p:P]$, $K_{p'} = J_{p'}$. But as stated, the choice of $J_{\geq p}$ 
is independent and uniform at this point; in particular, each $J_{p'}$ is chosen uniformly from $[t]$. Thus, the probability of this event happening is at most $P/t$ by union bound, concluding the proof. \Qed{lem:irrelevant}
\end{proof}

\subsection{The Lower Bound Proof} 

We are now ready to finalize the proof of the lower bound. For this part, it helps to recall the definitions in~\Cref{sec:comp-lem}. 

\begin{lemma}\label{lem:special-count}
	Any protocol $\pi$ with communication cost $s \geq \log\paren{\card{L}+1}$, cannot output more than 
	\[
		 \frac{32s}{\delta^2} \cdot \frac{P^2}{t} 
	\]
	special edges as part of its returned matching with probability more than $1-\delta/2$. 
\end{lemma}
\begin{proof}
	Suppose in the communication game, to each $\PS{p}$ for $p \in [P+1]$, we also present the choice of $J_{< p}$ additionally; any lower bound proven with this extra information clearly holds against the original version of the problem as well. We partition
	the analysis from the perspective of the last player versus all other players. 
	
	\paragraph{Player $P+1$.} This player is responsible for outputting the answer as a function of $\msg_1,\ldots,\msg_P$, $J$, and $\EE{P+1}$. Note that $\EE{P+1}$ is deterministically fixed by $J$ so it does not provide any further information. 
	We say that an edge $e \in \Ebase$ \textbf{belongs} to $\msg_1,\ldots,\msg_P,J$ iff 
	\[
		\Pr_{G}\paren{e \in G \mid \msg_1,\ldots,\msg_P,J} \geq 1-\delta/2. 
	\]
	Given that the output of $\PS{P+1}$ is a deterministic function of $(\msg_1,\ldots,\msg_P,J)$, if this player outputs an edge that does \emph{not} belong to these messages, then, with probability 
	at least $\delta/2$ over the choice of the input (conditioned on everything), this edge will not be part of the graph. Thus, whenever $\PS{P+1}$ outputs an edge that does not belong to received messages and $J$, 
	the protocol errs with probability at least $\delta/2$. Hence, in the following, we only need to focus on the special edges that belong to the messages and bound their total number. 
	
	Fix some $p \in [P]$ and consider an edge $e \in {\Ebase}^{(p)}$. We have 
	\begin{align*}
		\Pr_{G}\paren{e \in G \mid \msg_1,\ldots,\msg_P,J} &= \Pr_{\EE{p}} \paren{e \in \EE{p} \mid \msg_1,\ldots,\msg_P,J} \tag{because $e$ is in ${\Ebase}^{(p)}$, we only need to consider randomness of $\EE{p} \subseteq {\Ebase}^{(p)}$} \\
		&=  \Pr_{\EE{p}} \paren{e \in \EE{p} \mid \msg_{1},\ldots,\msg_p,J} \tag{conditioned on $J$, players' inputs are independent and thus $\EE{p} \perp \msg_{p+1},\ldots,\msg_{P} \mid \msg_{p}, J$} \\
		&= \Pr_{\EE{p}} \paren{e \in \EE{p} \mid \msg_{1},\ldots, \msg_p,J_{<p}} \tag{because again $\EE{p} \perp J_{\geq p}$ conditioned on the rest}.
	\end{align*}
	Hence, for any edge $e \in {\Ebase}^{(p)}$, the edge $e$ belongs to $\msg_1,\ldots,\msg_P,J$ iff the following equivalent condition holds: 
	\begin{align}
		\Pr_{\EE{p}} \paren{e \in \EE{p} \mid \msg_{<p},J_{<p},\msg_p} \geq 1-\delta/2. \label{eq:our-task-belong}
	\end{align}
	In the following, we focus on bounding the number of special edges that satisfy~\Cref{eq:our-task-belong} for each player separately. 
	
	\paragraph{Player $p \in [P]$.} Conditioned on $\msg_{<p},J_{<p}$, the input of $\PS{p}$ is a random subgraph of ${\Ebase}^{(p)}$ obtained by removing $\delta$-fraction of edges uniformly. 
	Consider the message $\msg_p$ written by $\PS{p}$ on the blackboard. We can think of $(\pi,\msg_{<p},J_{<p})$ as being a compression scheme for subgraphs of ${\Ebase}^{(p)}(J_{<p})$, 
	which uses $\msg_p$ as its summary. 
	
	Recall the definition of an edge belonging to a summary in~\Cref{sec:comp-lem} and note that it is equivalent with~\Cref{eq:our-task-belong}. We can thus 
	define $B(\msg_{<p},J_{<p},\msg_p)$ to be the set of edges that belong to $\msg_p$ under the compression scheme $(\pi,\msg_{<p},J_{<p})$, and 
	note that this is the same as $\belong{{\Ebase}^{(p)}(J_{<p})}{\msg_p}$ defined in~\Cref{sec:comp-lem}. Thus, by~\Cref{prop:compression}, 
	\[
		\Exp_{\EE{p}} \card{B(\msg_{<p},J_{<p},\msg_p) \mid \msg_{<p},J_{<p}} \leq  \frac{16s}{\delta}. 
	\]
	At the same time, an $e \in {\Ebase}^{(p)}(J_{<p})$ can only become special given the randomness of $J_{\geq P}$, which is entirely independent of $\msg_{<p},J_{<p},\msg_p$ as argued earlier. Thus, 
	for any choice of $\msg_{<p},J_{<p},\msg_p$ and even $\EE{p}$, 
	\[
		\Exp_{J \geq p}\bracket{\text{\# of special edges in $B(\msg_{<p},J_{<p},\msg_p)$} \mid \EE{p},\msg_{<p},J_{<p},\msg_p} \leq \frac{P}{t} \cdot \card{B(\msg_{<p},J_{<p},\msg_p)}. 
	\]
	by~\Cref{lem:irrelevant}. Combining the above two equations implies that 
	\begin{align}
		\Exp_{\EE{p},J_{\geq p}} \bracket{~\text{\# of special edges in $B(\msg_{<p},J_{<p},\msg_p)$} \mid \msg_{<p},J_{<p}} \leq  \frac{16s}{\delta} \cdot \frac{P}{t}. \label{eq:our-task-final}
	\end{align}
	
	\paragraph{Concluding the proof.} By linearity of expectation and~\Cref{eq:our-task-final}, we obtain that over the randomness of the entire input, 
	\[
		\Exp\bracket{~\text{\# of special edges that belong to messages of players}~} \leq P \cdot \frac{16s}{\delta} \cdot \frac{P}{t}. 
	\]
	Thus, by Markov bound, 
	\[
		\Pr\paren{\text{$\PS{P+1}$ can output more than $\frac{32s}{\delta^2} \cdot \frac{P^2}{t}$ special edges that belong}} \leq \frac{\delta}{2}. 
	\]
	Thus, either the protocol outputs an edge that does not belong to some message and errs with probability at least $\delta/2$, or cannot find enough special edges to output with probability at least $\delta/2$, 
	concluding the proof. \Qed{lem:special-count}
\end{proof}

Finally, we are ready to conclude the proof of~\Cref{thm:framework}. 

\begin{proof}[Proof of~\Cref{thm:framework}]
	Consider the distribution defined above and let 
	\[
	s = \frac{\delta^3 \cdot n \cdot t}{32P^2}
	\]
	 be the communication cost of the protocol. Let $M_\pi$ denote the returned matching of the protocol. By~\Cref{lem:special-count}, with probability at least $\delta/2$, the matching $M_{\pi}$
	can only have
	\[
		 \frac{32s}{\delta^2 } \cdot \frac{P^2}{t} =  \frac{32\delta^3  \cdot n \cdot t}{32P^2 \cdot \delta^2} \cdot \frac{P^2}{t} = \delta n
	\]
	special edges. By~\Cref{lem:large-matching}, this implies that 
	\[
		\card{M_\pi} \leq \delta n + \paren{2-2\cdot value(\bG)} \cdot \card{L} + P \cdot \delta \cdot \card{L} \leq  \paren{2-2\cdot value(\bG)} \cdot \card{L} + (P+1) \cdot \delta n. 
	\]
	On the other hand, again by~\Cref{lem:large-matching}, with probability $1-o(1)$, graph $G$ has a matching of size 
	\[
	\paren{2-value(\bG)} \cdot \card{L} - 2\delta \cdot \card{L} \geq \paren{2-value(\bG)} \cdot \card{L} - 2\delta n. 
	\]
	Thus, with probability at least $\delta/2-o(1)$, the approximation ratio of the protocol is at most 
	\[
		\frac{{2-2\cdot value(\bG)}}{2-value(\bG)} + 10P \delta = \alpha(\bG) + 10P \delta,
	\]
	by the definition of approximation ratio $\alpha(\bG)$ of the blueprint in~\Cref{def:blueprint-value-approx}. 
	
	Finally, we can re-parameterize $\delta \leftarrow 10P\delta$ and since both parameters $P$ and $\delta$ are constant with respect to $n$, and using the value of $t=n^{\Omega(1/\log\log{n})}$ in~\Cref{prop:base-c-rs}, 
	we get that having 
	\[
		s = n^{1+o(1/\log\log{n})}
	\]
	communication implies that success probability of the protocol in outputting a $(\alpha(\bG)+\delta)$-approximation is at most $1-\delta$, concluding the proof. \Qed{thm:framework}
\end{proof}

\clearpage


\newcommand{\Old}{\ensuremath{\textsc{Old}}}
\newcommand{\New}{\ensuremath{\textsc{New}}}

\newcommand{\Top}{\ensuremath{\textsc{Top}}}
\newcommand{\Bot}{\ensuremath{\textsc{Bot}}}

\newcommand{\Free}{\ensuremath{\textsc{Free}}}
\newcommand{\Match}{\ensuremath{\textsc{Match}}}

\newcommand{\bA}{\mathbb{A}}
\newcommand{\bB}{\mathbb{B}}

\newcommand{\bAL}{\mathbb{AL}}
\newcommand{\bBL}{\mathbb{BL}}
\newcommand{\bCL}{\mathbb{CL}}
\newcommand{\bDL}{\mathbb{DL}}

\newcommand{\bAR}{\mathbb{AR}}
\newcommand{\bBR}{\mathbb{BR}}
\newcommand{\bCR}{\mathbb{CR}}
\newcommand{\bDR}{\mathbb{DR}}
\newcommand{\bOL}{\mathbb{OL}}
\newcommand{\bOR}{\mathbb{OR}}

\newcommand{\bX}{\mathbb{X}}
\newcommand{\bY}{\mathbb{Y}}

\section{Blueprint Constructions}\label{sec:construct-blueprints}

Equipped with \Cref{thm:framework}, we
can now switch to constructing blueprints with small approximation ratio or alternatively large value. This is the topic of this section. 

Let us define $\tvalue$ as the largest possible 
$value$ in proper blueprints, i.e.,\footnote{Considering $\set{value(\bG) : \text{a proper blueprint $\bG$}}$ as a subset of \emph{real} numbers in $[0,1]$, 
by the least upper bound property of bounded real intervals,  this set does indeed admit a supremum and thus $\tvalue$ exist and is well-defined.}
\begin{align}
	\tvalue := \sup\set{value(\bG):\text{a proper blueprint $\bG$}}. \label{eq:target-value}
\end{align}

By~\Cref{thm:framework} and \Cref{def:blueprint-value-approx}, no single-pass semi-streaming algorithm for the maximum matching problem achieves an approximation
ratio $\tapprox + \delta$ for any fixed $\delta > 0$, 
where 
\begin{align}
	\tapprox := \frac{2-2\cdot \tvalue}{2-\tvalue} \qquad \paren{=\inf\set{\alpha(\bG):\text{a proper blueprint $\bG$}}}; \label{eq:target-approx}
\end{align}
the existence of $\delta$ also ``takes care'' of the supremum term to obtain a \emph{finite} proper blueprint\footnote{We note that we could have directly worked with $\tapprox$ instead of $\tvalue$ but in general
 $value$ of a blueprint is a more natural property of the blueprint and is easier to work with in the calculations.}. 

Our main result in this section is the following. 

\begin{theorem}\label{thm:main-blueprint}
	For $\tvalue = \sup\set{value(\bG):\text{a proper blueprint $\bG$}}$ defined in~\Cref{eq:target-value}, 
	\[
		\tvalue \geq \frac{5-\sqrt{10}}{3}. 
	\]
	Alternatively, for $\tapprox$ defined in~\Cref{eq:target-approx}, we have $\tapprox \leq (8-2\sqrt{10})/3 \sim 0.558$. 
\end{theorem}

We start with a warm-up that proves $\tvalue \geq 2-\sqrt{2}$ or equivalently $\tapprox \leq 2-\sqrt{2} \sim 0.585$; this already improves the state-of-the-art lower bounds for the streaming matching problem in~\cite{Kapralov21}. We then build on this to prove~\Cref{thm:main-blueprint} which concludes the proof of~\Cref{res:main}.

\subsection{Warm-up: A $(2-\sqrt{2})$ Approximation Blueprint}\label{sec:2-sqrt-2}

We start with our warm-up construction. 

\begin{proposition}\label{thm:2-sqrt2-blueprint}
	Let $\tvalue = \sup\set{value(\bG):\text{a proper blueprint $\bG$}}$ be as in~\Cref{eq:target-value}. Then, 
	\[
		\tvalue \geq 2-\sqrt{2}. 
	\]
\end{proposition}

The key to proving~\Cref{thm:2-sqrt2-blueprint} is the following \emph{amplification} step.

\begin{lemma}\label{lem:2-sqrt2-amplification}
	For any small $\delta \in (0,1/10)$ and any proper blueprint $\bG$ with $value(\bG) < 2-\sqrt{2} - \delta$, there exists another proper blueprint $\hbG$ such that 
	\[
	value(\hbG) > value(\bG) + \delta^3.
	\] 
\end{lemma}
We note that the dependence on $\delta$ in~\Cref{lem:2-sqrt2-amplification} is inconsequential and the only important part is that there is some fixed function of $\delta$ in the RHS. 
\Cref{thm:2-sqrt2-blueprint} follows easily from~\Cref{lem:2-sqrt2-amplification}.  

\begin{proof}[Proof of~\Cref{thm:2-sqrt2-blueprint}]
	Suppose by contradiction that $\tvalue = 2-\sqrt{2}-\beta$ for some $\beta > 0$. Then, for any fixed $\theta > 0$, 
	there must exist a  (finite) proper blueprint with value at least $2-\sqrt{2}-\beta-\theta$. By~\Cref{lem:2-sqrt2-amplification}, we can find 
	another proper blueprint with value at least 
	\[
	2-\sqrt{2}-\beta-\theta + (\beta+\theta)^3 \geq 2-\sqrt{2}-\beta-\theta + \beta^3 > 2-\sqrt{2}-\beta,
	\] 
	by taking $\theta < \beta^3$. This leads to a contradiction, implying that $\tvalue \geq 2-\sqrt{2}$. \Qed{thm:2-sqrt2-blueprint}
\end{proof}

We now prove~\Cref{lem:2-sqrt2-amplification}. 
Throughout the proof, we fix $\eps > 0$ to be a sufficiently small constant as a function of $\delta$ such that $1/\eps$ is an integer.
We first mention several without-loss-of-generality assumptions that will be used in the proof.

\paragraph{Initial assumptions.}  By~\Cref{prop:simplify-blueprint}, we can assume without loss of generality that the starting blueprint $\bG$ in~\Cref{lem:2-sqrt2-amplification} is simple. Thus, 
the blueprint is $\bG = (\bL,\bR,\bE)$ with only two parameters $(P,C)$. We can also apply~\Cref{lem:increase-C-blueprint} to $\bG$ to ensure that its $C$-parameter is such that $\eps \cdot C$ is an integer (given $1/\eps$ is an integer). 
Finally, we also have $value(\bG) \geq 1/2$ because otherwise we can simply replace $\bG$ with any simple blueprint of that value (say, the one in~\Cref{fig:2/3-blueprint}) and use that one instead. 

\subsubsection{A high level intuition of the construction}

We first give a high-level intuition of what our construction ``attempts'' to achieve -- the actual construction requires some technical deviation but the conceptual ideas are the same. 
For this discussion, it will be helpful to follow~\Cref{fig:sqrt2-step-by-step} step by step when reading the construction below.

\begin{itemize}
	\item \textbf{Step 1}: Create $\hbG$ by starting from a copy of $\bG$ and its edges. 
	\item  \textbf{Step 2}: Increase the $P$-parameter of $\hbG$, which amounts to introducing a new dimension to the tuples defining the vertices. Consider vertices with new-dimension-index in
 $\set{1,\ldots,(1-\eps)C}$ as \emph{Old} vertices and the ones in $\set{(1-\eps)C+1,\ldots,C}$ as \emph{New} vertices\footnote{We note that implementing this step via blueprints is not directly possible and requires a different route.}. 
	\item \textbf{Step 3:} Add a new block of relative size $s = \Theta(\eps)$ to be fixed later; use the same partitioning between \emph{Old} and \emph{New} vertices for this block as well. 
	\item \textbf{Step 4:} Connect the \emph{New} vertices of the top blocks to the \emph{Old} vertices of the bottom blocks, using  the ``mirror image'' of unmatched vertices of top blocks to pick the bottom block vertices. 
\end{itemize}

\begin{figure}[h!]
	\centering
	\subcaptionbox{\textbf{Step 1}: start with the original blueprint $\bG$.}%
	[.22\linewidth]{
		\includegraphics[scale=0.4]{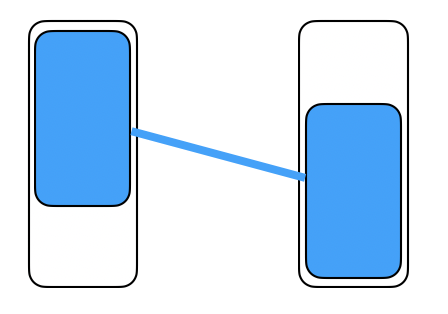}
	} 
	\hspace{0.2cm} 
	\subcaptionbox{\textbf{Step 2}: add a new dimension, partition vertices by a $1-\eps$ to $\eps$ ratio.}%
	[.22\linewidth]{
		\includegraphics[scale=0.35]{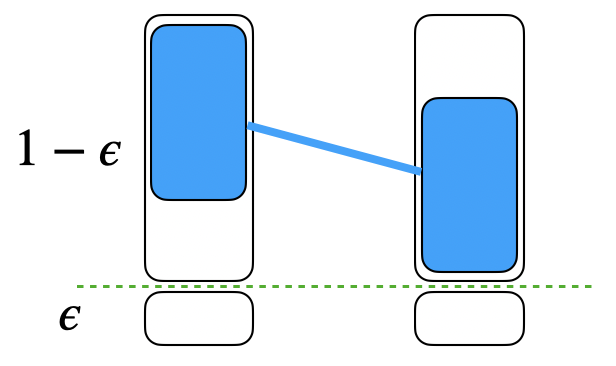}
	} 
	\hspace{0.2cm}
	\subcaptionbox{\textbf{Step 3}: add a  new block with (much) smaller relative size and partition it the same.}%
	[.22\linewidth]{
		\includegraphics[scale=0.35]{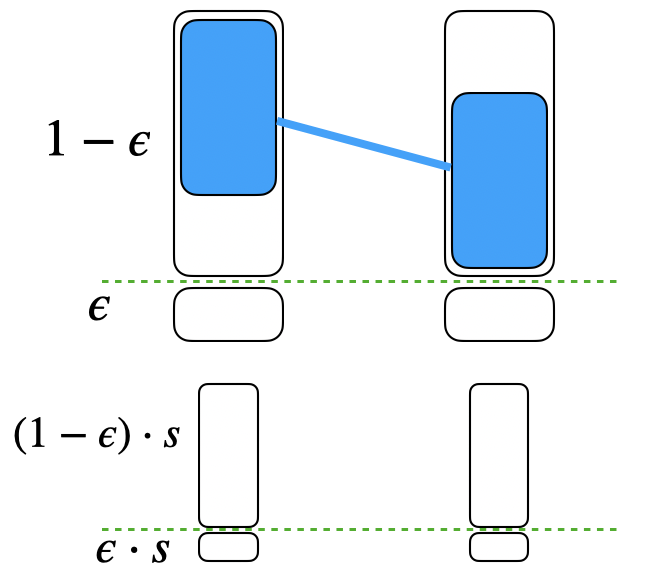}
	} 
	\hspace{0.2cm} 
	\subcaptionbox{\textbf{Step 4}: connect New vertices at the top to the specified subset of Old vertices at the bottom.}%
	[.22\linewidth]{
		\includegraphics[scale=0.35]{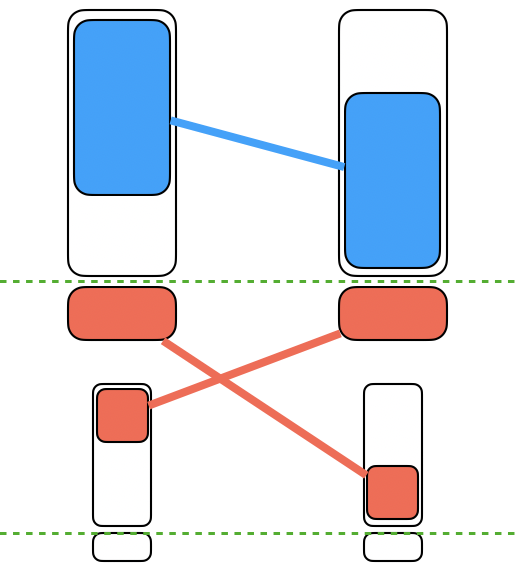}
	} 
	\caption{The intuition behind the blueprint $\hbG$ in~\Cref{lem:2-sqrt2-amplification}. 
	Shaded (blue) vertices at the top are matched in $\bG$ and the new shaded (red) vertices at the bottom are additionally matched in $\hbG$.}
	\label{fig:sqrt2-step-by-step}
\end{figure}

\noindent
To show $\hbG$ is a proper blueprint, we need to ensure it satisfies constraints of~\Cref{def:proper-blueprint}: 
\begin{itemize}
\item \textbf{Matching constraint:} the original edges of $\bG$ satisfy the requirement by default, so we only need to ensure the new edges can form a (fractional) matching satisfying the relative sizes: for this, 
we have relative size of \emph{New} vertices at top is $\eps$ and the \emph{Old} \underline{and} matched vertices at the bottom is $(1-value(\bG)) \cdot (1-\eps) \cdot s$. We thus need to pick 
\[
	s = \frac{\eps}{(1-\eps) \cdot (1-value(\bG))},
\]
to make these two quantities equal and ensure the matching constraint. 
\item \textbf{Ban constraint:} we add edges of $\bG$ to $\hbG$ as part of $\hbEE{2},\ldots,\hbEE{P+1}$ (i.e., by shifting their group-index by one); this ensures that the ``range'' of bans introduced by these edges will be localized to \emph{Old} vertices 
at the top; since $\bG$ was a proper blueprint, we obtain that the ban constraint of these edges are also satisfied; see~\Cref{fig:sqrt2-ban1}. 

The new edges will be added to $\hbEE{1}$ and so introduce ``global'' bans between the two blocks (by global, 
we mean not respecting \emph{Old} and \emph{New} partitions). But, by the choice of 
which vertices in the bottom blocks are matched (namely, the ``mirror image'' of unmatched vertices at the top for \emph{Old} vertices, and none of \emph{New} vertices of bottom), we ensure that at most one endpoint of each banned pair is matched as required; see~\Cref{fig:sqrt2-ban2}.  

\end{itemize}

\begin{figure}[h!]
	\centering
	\subcaptionbox{Bans introduced by edges of $\bG$ are localized to \emph{Old} vertices at the top.\label{fig:sqrt2-ban1}}%
	[.48\linewidth]{
		\includegraphics[scale=0.5]{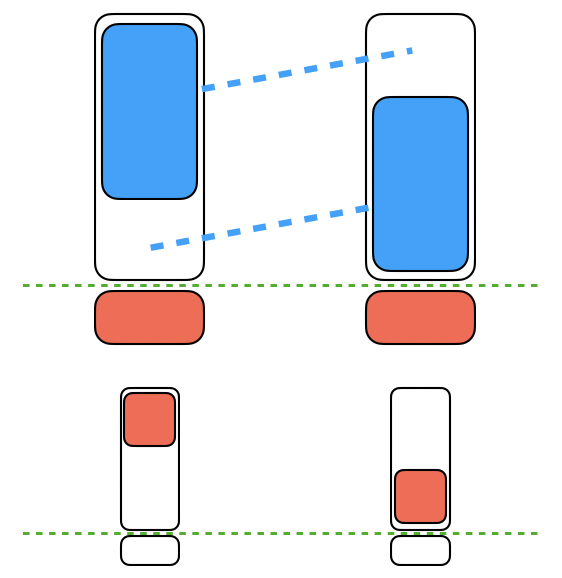}
	} 
	\hspace{0.2cm} 
	\subcaptionbox{\textbf{Step 2}: Bans introduced by new edges are handled by the construction.\label{fig:sqrt2-ban2}}%
	[.48\linewidth]{
		\includegraphics[scale=0.5]{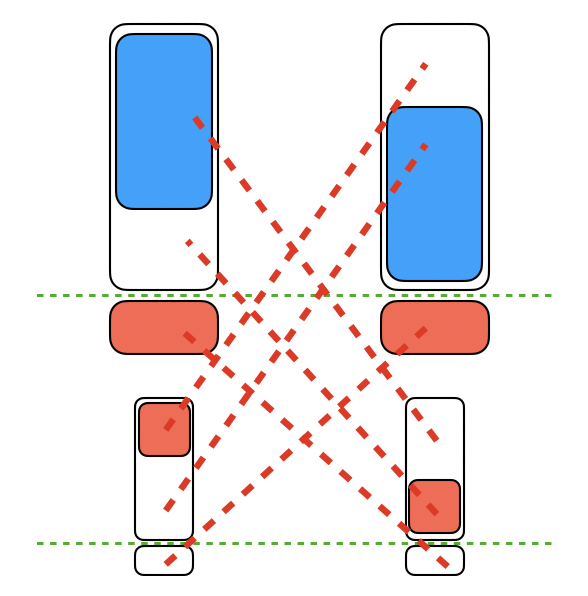}
	} 
	\caption{An intuitive explanation of the bans in the blueprint $\hbG$ of~\Cref{lem:2-sqrt2-amplification}. The dashed edges depict the  banned pairs. 
	The ban constraint requires at most one endpoint of each dashed edge to be matched.
	}
	\label{fig:sqrt2-ban}
\end{figure}

Finally, we should calculate $value(\hbG)$ and compare it with $value(\bG)$. By~\Cref{def:blueprint-value-approx}, the value of $\hbG$ is the fraction of its matched vertices over the relative sizes of the vertices. 
Since we need $value(\hbG) > value(\bG)$, but \emph{Old} top vertices achieve a value of exactly $value(\bG)$, we need to ensure the added vertices overall match at a higher rate; this translates 
to having 
\[
	\frac{4\eps}{2\eps+2s} > value(\bG);
\]
we have added $2(s+\eps)$ vertices in total to $\bG$ and have matched $4\eps$ of them ($2\eps$ of \emph{New} vertices in the top block and another $2\eps$ of their matched pairs). 
Plugging in the value of $s$  we get 
\[
	2\eps > \eps \cdot value(\bG) + \frac{\eps}{1-value(\bG)} \cdot value(\bG) + \Theta(\eps^2),  
\]
which can be re-written as 
\[
	value(\bG)^2 - 4 \cdot value(\bG) + 2 > \Theta(\eps). 
\]
Given the quadratic polynomial on the left is strictly positive for any $value(\bG) < 2-\sqrt{2}$, we can pick $\eps$ small enough to satisfy this inequality and obtain that $value(\hbG) > value(\bG)$, as desired.

\subsubsection{Construction of the blueprint $\hbG$ in~\Cref{lem:2-sqrt2-amplification}} 

We now formalize the blueprint $\hbG=(\hbL,\hbR,\hbE,\hbw)$ with parameters $(\hP,\hC,\hB_L,\hB_R,\hsizes_L,\hsizes_R)$. The reader may find it useful to refer to~\Cref{fig:2-sqrt2} throughout this proof. 

\medskip

\begin{figure}[h!]
	\centering
	\includegraphics[scale=0.5]{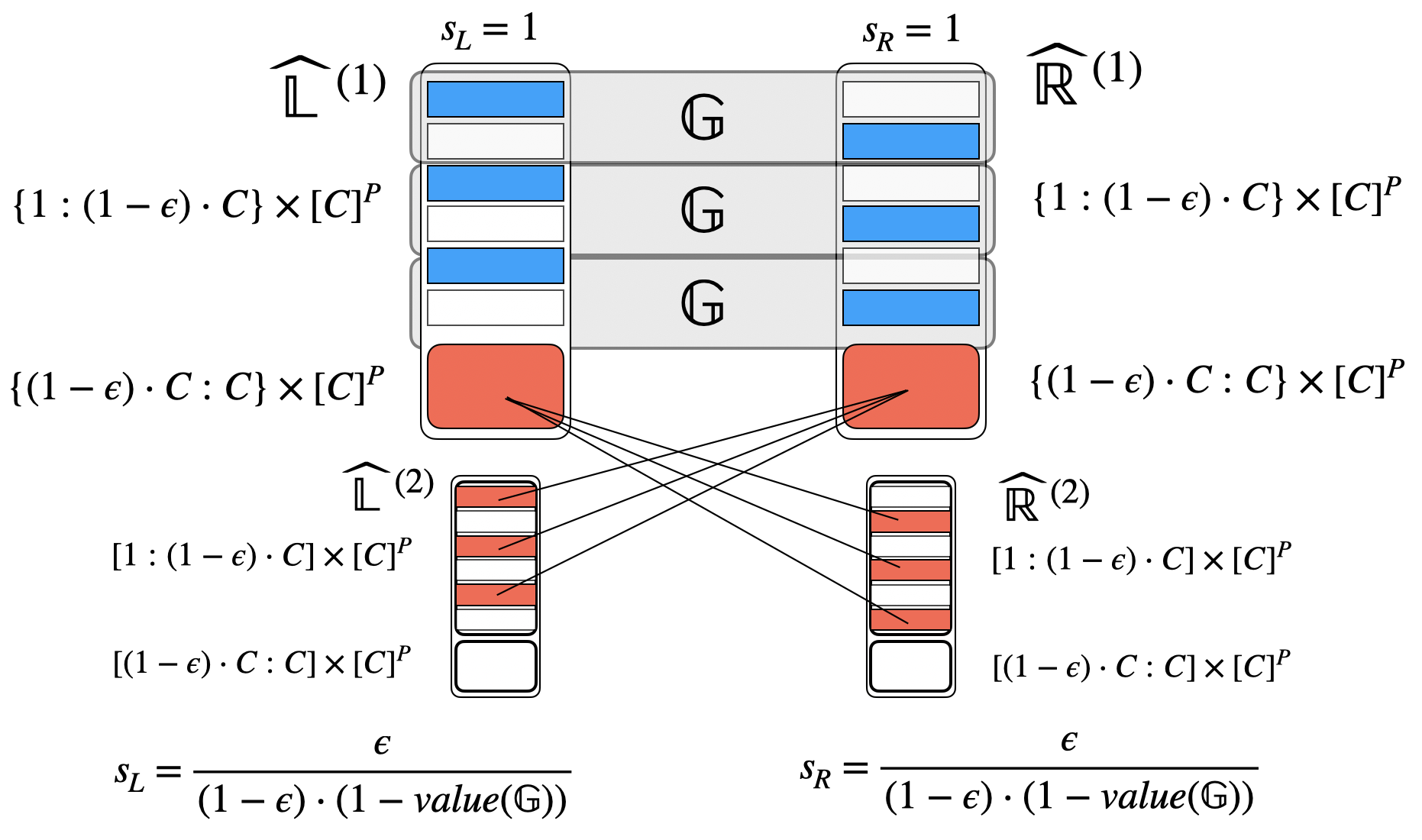}
	\caption{An illustration of the proper blueprint $\hbG$ in the proof of~\Cref{lem:2-sqrt2-amplification}. The shaded (red) vertices at the bottom part are $\Free_L$ and $\Free_R$.}\label{fig:2-sqrt2}
\end{figure}

\bigskip

We start with setting up the parameters.

\paragraph{Parameters of $\hbG$.} Define
\[
	\hP = P+1~,~\hC = C~,~\hB_L = \hB_R = 2~,~ \hsizes_L = \hsizes_R = (1,s), 
\]
where 
\begin{align}
	s := \frac{\eps}{(1-\eps) \cdot (1-value(\bG))},  \label{eq:2-sqrt2-s}
\end{align}
for the parameter $\eps$ chosen earlier at the beginning of the proof. 

\paragraph{Vertices of $\hbG$.} The choice of parameters fixes the vertices of $\hbG$ to also be
\[
	\hbL = \hbLL{1} \sqcup \hbLL{2} \quad\text{and}\quad \hbR = \hbRR{1} \sqcup \hbRR{2}, 
\]
where each of the blocks is now identified with $[C]^{P+1}$. We further define
\[
	\Old := \set{1,\ldots,(1-\eps)\cdot C}  \qquad \text{and} \qquad \New := \set{(1-\eps)\cdot C+1,\ldots,C} .
\]
These sets will be useful when defining the edges of the graph. 

\paragraph{Edges of $\hbG$ and the fractional assignments $\hbw$.} This is the main step of the construction. We will have two different types of edges as defined below (see also~\Cref{fig:2-sqrt2}):  

\begin{itemize}
	\item \textbf{Type 1 edges:}  For every $i \in \Old$, connect 
	\[
		\text{$\hbLL{1}(i,x)$ to $\hbRR{1}(i,y)$ in $\hbEE{p+1}$ iff $(\bL(x),\bR(y))$ is an edge in $\bEE{p}$ of $\bG$}. 
	\]
	We further set the value of any edge $\be$ in this step to be
	\begin{align}
		\hbw(\be) := 1. \label{eq:sqrt2-we1}
	\end{align}
	When considering $\hbLL{1}(i,*)$ and $\hbRR{1}(i,*)$, the remaining labels of vertices are in $[C]^P$ and thus there is a natural bijection 
	between these two sets and $\bL$ and $\bR$. The edges in this step are simply ``embedding'' a copy of $\bG$ between each of $\hbLL{1}(i,*)$ and $\hbRR{1}(i,*)$ -- we only ``shift'' the partition of each 
	edge by one, i.e., edges in $\bEE{p}$ are added to $\hbEE{p+1}$ for $p \in [P]$. 

	\item \textbf{Type 2 edges:} Add an edge to $\hbEE{1}$ between every pairs of vertices in,\footnote{We note that the use of $\bR$ to define $\hbLL{2}$  in $\Free_L$ (and similarly $\bL$ in $\Free_R$) 
	is intentional and not a typo.} 
	\begin{align*}
		&\underbrace{\set{\hbLL{2}(i,x) \mid \text{$i \in \Old$ and $\bR(x)$ has no edge in $\bG$}}}_{:=\Free_L} \times \set{\hbRR{1}(j,y) \mid j \in \New}, \\
		\intertext{as well as between}
		&\underbrace{\set{\hbRR{2}(i,x) \mid \text{$i \in \Old$ and $\bL(x)$ has no edge in $\bG$}}}_{:=\Free_R} \times \set{\hbLL{1}(j,y) \mid j \in \New}.
	\end{align*}
	We set the fractional value of any edge $\be$ in this step to be
	\begin{align}
		\hbw(\be) := \frac{s}{\eps \cdot C^{P+1}}. \label{eq:sqrt2-we2}
	\end{align}
\end{itemize}

\noindent
This concludes the construction of the blueprint $\hbG$.

\subsubsection{Proving $\hbG$ is proper and concluding the proof of~\Cref{lem:2-sqrt2-amplification}}

We now need to ensure that the blueprint $\hbG$ we defined is indeed a proper blueprint, namely, verify its matching and ban constraints, and finally calculate its value to conclude the proof. 
We do so in the following three claims. 

\begin{claim}\label{clm:2-sqrt2-matching-constraint}
	Blueprint $\hbG$ satisfies the matching constraint. 
\end{claim}
\begin{proof}
We only prove the bounds for vertices in $\hbL$. By symmetry, the same exact bounds also hold for vertices in $\hbR$. 
There are four types of vertices in $\hbL$ (based on the block they belong to and the value of their first index); we go over each one separately: 
\begin{itemize}
	\item $\hbLL{1}(i,x)$ for $i \in \Old$ and $x \in [C]^P$: each vertex here has either one or zero edge depending on whether $x$ is matched or not in the simple blueprint $\bG$. Since 
	$\sizes_L(1)=1$ and the fractional value of each edge is $1$ by~\Cref{eq:sqrt2-we1}, we have that these vertices satisfy the matching constraint. 
	\item $\hbLL{1}(i,x)$ for $i \in \New$ and $x \in [C]^P$: each vertex here has edges to all vertices in $\Free_R$ defined in the construction. Thus, for a vertex $\bv = \hbLL{1}(i,x)$ for $i \in \New$, 
	given the fractional values of these edges in~\Cref{eq:sqrt2-we2}, we have 
	\begin{align*}
		\sum_{\be \ni \bv} \hbw(\be) &= \card{\Free_R} \cdot \frac{s}{\eps \cdot C^{P+1}} \\
		&= (1-value(\bG)) \cdot \card{\Old} \cdot C^{P} \cdot \frac{s}{\eps \cdot C^{P+1}} \tag{since $(1-value(\bG))$ fraction of vertices in $\hbRR{2}(i,*,\ldots,*)$ are in $\Free_R$ for any $i \in \Old$} \\
		&= (1-value(\bG)) \cdot (1-\eps) \cdot \frac{s}{\eps} \tag{as $\card{\Old} = (1-\eps) \cdot C$} \\
		&= (1-value(\bG)) \cdot (1-\eps) \cdot \frac{1}{\eps} \cdot \frac{\eps}{(1-\eps) \cdot (1-value(\bG))} \tag{by the choice of $s$ in \Cref{eq:2-sqrt2-s}} \\
		&= 1. 
	\end{align*}
	Since $\sizes_L(1) = 1$, these vertices also satisfy the matching constraint. 
	
	\item $\hbLL{2}(i,x)$  for $i \in \Old$ and $x \in [C]^P$: each vertex here either has no edges, if it is not in $\Free_L$, or has edges to all vertices $\hbRR{1}(j,y)$ for $j \in \New$ and $y \in [C]^P$. In the first case it satisfies the matching
	constraint immediately. For the second case, for any such vertex $\bv$, we have, 
	\begin{align*}
		\sum_{\be \ni \bv} \hbw(\be) &= \card{\set{\hbRR{1}(j,y) \mid j \in \New, y \in [C]^P}} \times \frac{s}{\eps \cdot C^{P+1}} \tag{by~\Cref{eq:sqrt2-we2} for the fractional weight of the edges} \\
		&= \eps \cdot C \cdot C^{P} \times \frac{1}{\eps \cdot C^{P+1}} \cdot s \tag{as $\card{\New} = \eps \cdot C$} \\ 
		&= s. 
	\end{align*}
	Since $\sizes_L(2) = s$, we have that these vertices also satisfy the matching constraint. 
	
	\item $\hbLL{2}(i,x)$  for $i \in \New$ and $x \in [C]^P$: these vertices have no edges in the construction and thus satisfy the matching constraint immediately. 
\end{itemize}

\noindent
In summary, the blueprint $\hbG$ satisfies the matching constraint. \Qed{clm:2-sqrt2-matching-constraint}

\end{proof}

\begin{claim}\label{clm:2-sqrt2-ban-constraint}
	Blueprint $\hbG$ satisfies the ban constraint. 
\end{claim}
\begin{proof}
We consider two different cases, based on the type of the edge introducing the ban constraint. 
\begin{itemize}
	\item \textbf{Type 1 edges:} Consider any edge $\be \in \hbEE{2} \sqcup \ldots \sqcup \hbEE{P+1}$. Let $\be = (\bu,\bv)$ where $\bu = \hbLL{1}(i,x)$ and $\bv = \hbRR{1}(i,y)$ for some $i \in \Old$ and $x,y \in [C]^P$ (by construction, we know there
	are no edges between $\hbLL{1}(i,x)$ and $\hbRR{1}(j,y)$ when $i \neq j$). 
	
	Any ban introduced by $\be$ will be between vertices of $\hbLL{1}(i,*,\ldots,*)$ and $\hbRR{1}(i,*,\ldots,*)$ (since we are adding an edge in $\hbEE{p}$ for $p > 1$ and the 
	first index of the labels of $\bu$ and $\bv$ is $i$). But then the subgraph between these two sets of vertices is exactly as in $\bG$ and since $\bG$ was a proper blueprint, this subgraph will also satisfy all its (internal) ban constraints. 
	
	\item \textbf{Type 2 edges:} Since we only have edges in $\hbEE{1}$ between $\hbLL{1}$ and $\hbRR{2}$, and $\hbLL{2}$ to $\hbRR{1}$, the bans introduced by them are also only between these two pairs of blocks. 
	Let us focus on the bans between $\hbLL{1}$ and $\hbRR{2}$; the other case also holds by symmetry. 
	
	By definition, any edge in $\hbEE{1}$ introduces the same bans constraints between $\hbLL{1}$ to $\hbRR{2}$: 
	\begin{itemize}
	\item for any $i \in [C]$ and any
	$x \in [C]^P$, at most one of $\hbLL{1}(i,x)$ or $\hbRR{2}(i,x)$ can have any incident edge. 
	\end{itemize}
	This holds in our construction because: 
	\begin{itemize}
		\item If $\hbRR{2}(i,x) \in \Free_R$, it means $i \in \Old$ and $\bL(x)$ has no edges in $\bG$ and since $\hbLL{1}(i,x)$ (for $i \in \Old$) has the same edges as $\bG$, then $\hbLL{1}(i,x)$ also has no edges in $\hbG$ as required. 
		\item If $\hbRR{2}(i,x) \notin \Free_R$, then this vertex has no edges in $\hbG$. 
	\end{itemize}
\end{itemize}
\noindent
In summary, the blueprint $\hbG$ satisfies the ban constraint also. \Qed{clm:2-sqrt2-ban-constraint}
\end{proof}

At this point, we established that $\hbG$ is indeed a proper blueprint by~\Cref{clm:2-sqrt2-matching-constraint} and~\Cref{clm:2-sqrt2-ban-constraint}. 
We thus only need to calculate the value of $\hbG$. The proof of the following claim is almost purely calculations (quite similar to the calculations done in the high level intuition earlier). 
As such, we postpone this proof to~\Cref{app:clm:2-sqrt2-value} to move on from this warm-up to the main proof. 

\begin{claim}[``Value of $\hbG$ increases'']\label{clm:2-sqrt2-value}
	\[
		value(\hbG) > value(\bG) + \delta^3. 
	\]
\end{claim}
\noindent
\Cref{lem:2-sqrt2-amplification} now follows immediately from~\Cref{clm:2-sqrt2-matching-constraint,clm:2-sqrt2-ban-constraint,clm:2-sqrt2-value}.

\subsection{Our Main Blueprint Construction}\label{sec:main-blueprint}

We now prove~\Cref{thm:main-blueprint}, namely, show that 
\[
	\tvalue \geq \frac{5-\sqrt{10}}{3}. 
\]
Similar to the proof of~\Cref{thm:2-sqrt2-blueprint}, we prove this theorem also using an amplification lemma. 

\begin{lemma}\label{lem:main-amplification}
	For any $\delta \in (0,1)$ and any proper blueprint $\bG$ with $value(\bG) < \frac{5-\sqrt{10}}{3}-\delta$, there exists another proper blueprint $\hbG$ such that 
	\[
	value(\hbG) > value(\bG) + \delta^3.
	\] 
\end{lemma}

\noindent
As before, there is nothing special about $\delta^3$-term in~\Cref{lem:main-amplification} and it is simply chosen to simplify the calculations. 
\Cref{thm:main-blueprint} follows from~\Cref{lem:main-amplification} exactly the same way~\Cref{thm:2-sqrt2-blueprint} followed from~\Cref{lem:2-sqrt2-amplification}.

\subsubsection{A high level intuition of the construction}\label{sec:main-high-level}

Unlike the previous subsection, we will not go over a step by step construction of the new blueprint 
as it will be redundant in light of the relatively short proof of the construction itself. But, it will be helpful to provide the following rationale for the 
seemingly arbitrary gadget we use in the proof of~\Cref{lem:main-amplification} to demystify its construction. 

It helps to think of~\Cref{lem:2-sqrt2-amplification} as a recursive way of creating a blueprint: 
we started with a blueprint $\bG$ and then added two different types of vertices to it. A new set $\bX$ that was directly appended to $\bG$, by increasing the parameter $P$ by one, 
and a new block of vertices $\bY$. The new blueprint then ensured that all of $\bX$ is matched to $(1-value(\bG))$ fraction of $\bY$ (essentially, by taking $\eps \rightarrow 0$ here). At this stage, we ``pack'' $\bG,\bX,\bY$ 
into a new blueprint and recurse on that and continue improving the value until it reaches $2-\sqrt{2}$. 

A second look at this approach shows that it is quite \emph{wasteful} actually: even though $\bX$ was fully matched, $\bY$ is barely matched (especially as we increase the value of the starting blueprint through the recursion). 
But, the moment we pack $\bX,\bY$ and $\bG$ together, we \emph{lose} this extra information: we simply achieve a blueprint that overall has a (relatively) high value although there are certain blocks inside it that are barely used (but we no longer have
access to them separately). 

Our goal in the new blueprint is to fix this wastefulness. We would like to ``work more'' on $\bY$ and add more vertices and blocks to it, to increase its overall value closer to $value(\bG)$, 
before packing it with $\bX$ and $\bG$ into a single blueprint. This step of increasing value of $\bY$ is reminiscent of increasing value of a blueprint on its own; the important difference however is that $\bY$ is now part of a 
bigger blueprint consisting of $\bG$ and $\bX$ and we cannot \emph{freely} work with it. 

Nevertheless, we can apply a similar strategy as increasing $value(\bG)$ in the proof of~\Cref{lem:2-sqrt2-amplification}, 
to also increase the value of $\bY$ here. The existence of $\bG$ and $\bX$ in the picture does not allow us to use~\Cref{lem:2-sqrt2-amplification} directly, nor even implement this step recursively. Instead, the new gadget we add
can be seen as applying multiple steps of the recursion of~\Cref{lem:2-sqrt2-amplification} to the blueprint of $\bY$ in one go, while ensuring the resulting blueprint can be plugged into $\bG$ and $\bX$. We simply optimize the number of times we apply this recursive strategy (this is the parameter $K$ to be defined in the proof) so the \emph{overall} value of the entire blueprint $\bG,\bX,\bY$, and not only $\bY$, is as large as possible\footnote{Technically speaking, we have \emph{not} fully optimized the choice of $K$ but rather picked a value that is close to optimal that makes the calculations quite easier; we further discuss this and other optimizations in~\Cref{sec:concluding}.}.

\subsubsection{Construction of the blueprint $\hbG$ in~\Cref{lem:main-amplification}} 

We now provide the proof of~\Cref{lem:main-amplification}. 

Throughout, we fix $\eps > 0$ to be a sufficiently small constant as a function of $\delta$ such that $1/\eps$ is an integer. 
We also pick two other integer parameters $K,t \geq 1$ such that 
\begin{align}
	\text{$t$ is sufficiently large} \quad , \quad \text{$(1-value(\bG)) \cdot t$ is an integer,} \quad \text{and} \quad K := (1-value(\bG)) \cdot t. \label{eq:main-tkeps}
\end{align}

As before, we assume that $value(\bG) \geq 1/2$ (otherwise start with the blueprint of~\Cref{fig:2/3-blueprint}), and that $\bG$ is a simple proper blueprint (by~\Cref{prop:simplify-blueprint}) and hence is $\bG = (\bL,\bR,\bE)$ with parameters $(P,C)$. 
We can also apply~\Cref{lem:increase-C-blueprint} to $\bG$ to ensure that its parameter $C$ is such that it is a multiple of $(K+t)! \cdot (1/\eps)$ (recall that $1/\eps$ is an integer). 

We go over setting up the parameters of the blueprint and then define its vertices and edges. 
\Cref{fig:main-blueprint} provides an illustration of this blueprint and various sets and their connections defined in the construction.

\paragraph{Parameters of $\hbG$.} Using the parameters in~\Cref{eq:main-tkeps}, define: 
\[
	\hP = P+K+1~,~\hC = C~,~\hB_L = \hB_R = K + 2~,~ \hsizes_L = \hsizes_R = (s_1,s_2,\ldots,s_{K+2}), 
\]
where 
\begin{alignat}{2}
	s_i &:= \frac{K+t-i+1}{t} \cdot \frac{1}{value(\bG)}, \quad && \forall i \in [K] \notag \\
	s_{K+1} &:=  K+t, \notag \\
	s_{K+2} &:= \frac{1}{\eps} \cdot (1-\eps) \cdot (1-value(\bG)) \cdot t.	 \label{eq:main-s}
\end{alignat}

\clearpage

\begin{figure}[t!]
	\centering
	\includegraphics[scale=0.5]{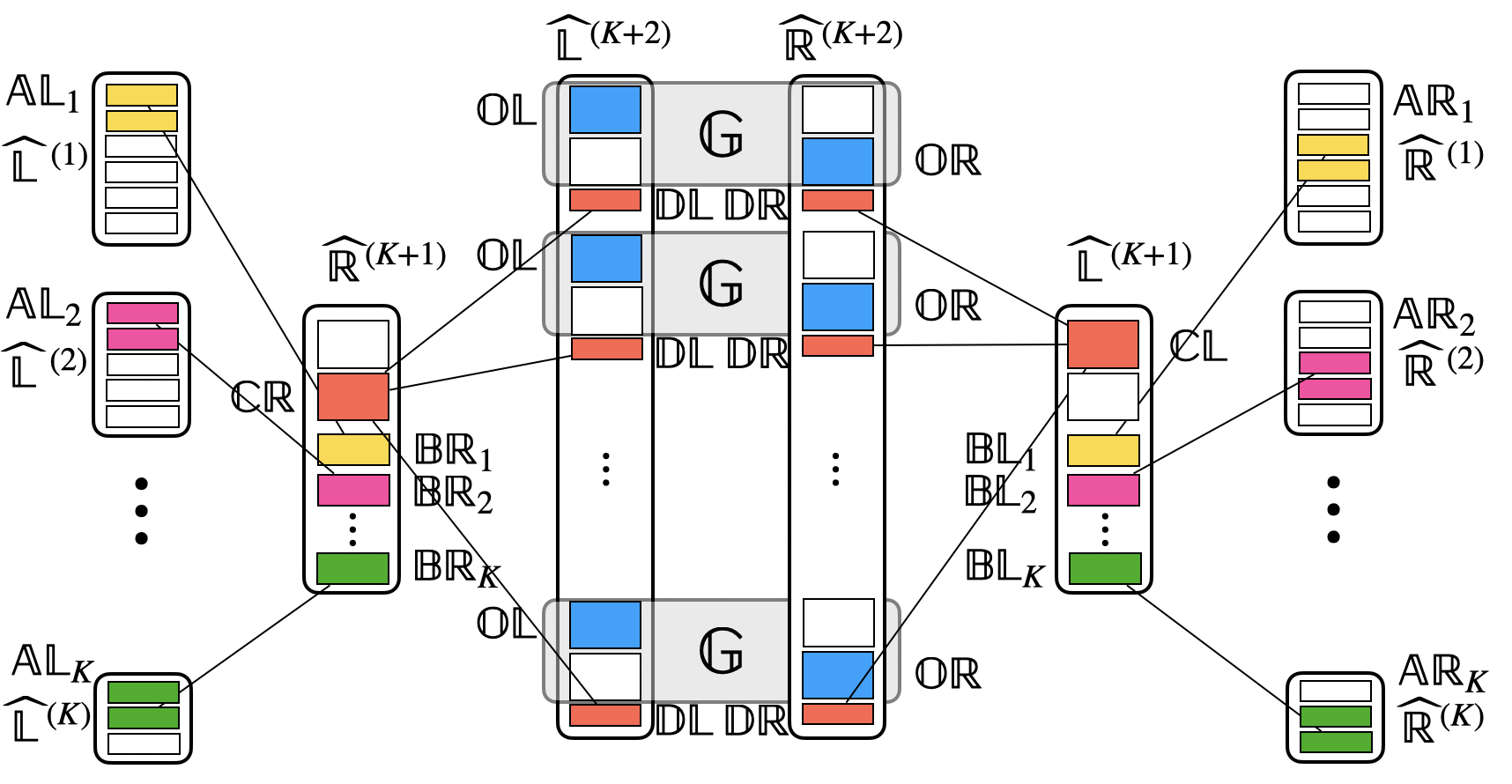}
	\caption{An illustration of $\hbG$ in the proof of~\Cref{lem:main-amplification} and its various internal sets and their connections.}\label{fig:main-blueprint}
\end{figure}

\paragraph{Vertices of $\hbG$.} The choice of parameters fixes the vertices of $\hbG$ to also be
\[
	\hbL = \hbLL{1} \sqcup \ldots \sqcup \hbLL{K+2} \quad\text{and}\quad \hbR = \hbRR{1} \sqcup \ldots \sqcup \hbRR{K+2}, 
\]
where each of the blocks is now identified with $[C]^{\hP} = [C]^{P+K+1}$. We further define: 
\begin{alignat}{2}
	\Top_i &:= \set{1,2,\ldots,\frac{K+t-i}{K+t-i+1} \cdot C} \quad,\quad &&\Bot_i := \set{\frac{K+t-i}{K+t-i+1} \cdot C+1,\ldots,C} ~ \forall i \in [K] \notag \\
	\Old &:= \set{1,\ldots,(1-\eps) \cdot C} \quad,\quad &&\New := \set{(1-\eps) \cdot C + 1 , \ldots, C} \notag \\
	\Free_L &:= \set{x \in [C]^P : \bL(x)~\text{has no edge in}~\bG} ~,~ &&\Free_R := \set{x \in [C]^P : \bR(x)~\text{has no edge in}~\bG} \notag \\
	\Match_L &:= \set{x \in [C]^P : \bL(x)~\text{has an edge in}~\bG} ~,~ &&\Match_R := \set{x \in [C]^P : \bR(x)~\text{has an edge in}~\bG} \notag \\
	\label{eq:main-top-bot}
\end{alignat}
Given $C$ is a multiple of $(K+t)! \cdot (1/\eps)$, all sets above are well-defined. %

\paragraph{Edges of $\hbG$ and the fractional assignments $\hbw$.} We have $K+2$ different types of edges:  

\begin{itemize}[leftmargin=5pt]
	\item \textbf{Type $i$ for $i \in [K]$:} Add an edge to $\hbEE{i}$ between every pairs of vertices in
	\begin{align*}
		&\overbrace{\hbLL{i}(\underbrace{*,\ldots,*}_{i-1},\Top_{i},\ldots,\Top_K,*,\Match_L)}^{\bAL_i :=} ~\times~ \overbrace{\hbRR{K+1}(\Top_1,\ldots,\Top_{i-1},\Bot_i,\underbrace{*,\ldots,*}_{\hP-i})}^{\bBR_i :=}, \\
	\intertext{as well as between}
		&\overbrace{\hbLL{K+1}(\Top_1,\ldots,\Top_{i-1},\Bot_i,\underbrace{*,\ldots,*}_{\hP-i})}^{\bBL_i :=}  ~\times~ \overbrace{\hbRR{i}(\underbrace{*,\ldots,*}_{i-1},\Top_{i},\ldots,\Top_K,*,\Match_R)}^{\bAR_i :=}.
	\end{align*}
	For every $i \in [K]$ and any type $i$ edge $\be$, we set the fractional weight as 
	\begin{align}
		\hbw(\be) := \frac{s_i}{\card{\bBR_i}} = \frac{s_i}{\card{\bBL_i}}. \label{eq:main-type-i}
	\end{align}
	\item \textbf{Type $K+1$:} Add an edge to $\hbEE{K+1}$ between every pairs of vertices in
	\begin{align*}
		&\overbrace{\hbLL{K+1}(\Top_1,\Top_2,\ldots,\Top_K,\Old,\Free_R)}^{\bCL:=} ~\times~ \overbrace{\hbRR{K+2}(\underbrace{*,\ldots,*}_K,\New,\underbrace{*,\ldots,*}_P)}^{\bDR:=} \\ 
		\intertext{as well as between}
		&\overbrace{\hbLL{K+2}(\underbrace{*,\ldots,*}_K,\New,\underbrace{*,\ldots,*}_P)}^{\bDL:=}  ~\times~ \overbrace{\hbRR{K+1}(\Top_1,\Top_2,\ldots,\Top_K,\Old,\Free_L)}^{\bCR:=}. 
	\end{align*}
	We set the fractional weight of any type $K+1$ edge $\be$ as 
	\begin{align}
		\hbw(\be) := \frac{s_{K+1}}{\card{\bDR}} = \frac{s_{K+1}}{\card{\bDL}}.  \label{eq:main-type-k+1}
	\end{align}
	\item \textbf{Type $K+2$:} For every $w \in [C]^K$, $z \in \Old$, and $x,y \in [C]^P$, add an edge between
	\[
		\text{$\hbLL{K+2}(w,z,x)$ to $\hbRR{K+2}(w,z,y)$ in $\hbEE{p+K+1}$ iff $(\bL(x),\bR(y))$ is an edge in $\bEE{p}$ of $\bG$.}
	\]
	We set the fractional weight of any edge $\be$ in this step to be 
	\begin{align}
		\hbw(\be) := s_{K+2}. \label{eq:main-we-k+2}
	\end{align}
	We refer to vertices matched in this step as $\bOL \subseteq \hbLL{K+2}$ and $\bOR \subseteq \hbRR{K+2}$. 
\end{itemize}

Before moving on, we list some simple observation about the sets defined above and their sizes. 
We postpone the simple proofs of these observations to~\Cref{app:obs:main-disjoint} and~\Cref{app:obs:main-sizes}. 

\begin{observation}\label{obs:main-disjoint}
	The sets $\bAL_i,\bAR_i,\bBL_i,\bBR_i$ for $i \in [K]$, and $\bCL,\bCR,\bDL,\bDR,\bOL,\bOR$ are all pairwise disjoint. 
\end{observation}

\begin{observation}\label{obs:main-sizes}
	We have 
	\begin{alignat*}{2}
		\card{\bAL_i} &= \card{\bAR_i} =  \frac{t}{K+t-i+1} \cdot value(\bG) \cdot C^{\hP} \qquad && \forall i \in [K] \\
		\card{\bBL_i} &= \card{\bBR_i} = \frac{1}{K+t} \cdot C^{\hP} \qquad && \forall i \in [K] \\
		\card{\bCL} &= \card{\bCR} = \frac{t}{K+t} \cdot (1-\eps) \cdot (1-value(\bG)) \cdot C^{\hP} \qquad \\
		\card{\bDL} &= \card{\bDR} = \eps \cdot C^{\hP}.
	\end{alignat*}
\end{observation}

\subsubsection{Proving $\hbG$ is proper and concluding the proof of~\Cref{lem:main-amplification}}

We now prove that $\hbG$ is a proper blueprint in the following two claims and then bound its value. 

\begin{claim}\label{clm:main-matching-constraint}
	Blueprint $\hbG$ satisfies the matching constraint.
\end{claim}
\begin{proof}
We only prove the constraint for vertices in $\hbL$. By symmetry, the same exact bounds also hold for vertices in $\hbR$. 
We consider each type of vertices separately (throughout, we use~\Cref{obs:main-disjoint} implicitly in various places and not mention it each time): 
\paragraph{Vertices in $\bAL_i$ for $i \in [K]$:} Any vertex $\bv \in \bAL_i$ is only incident on type $i$ edges 
	and is connected to all vertices in $\bBR_i$. We have, 
	\begin{align*}
		\sum_{\be \ni \bv} \hbw(\be) &= \card{\bBR_i} \cdot  \frac{s_i}{\card{\bBR_i}} = s_i. \tag{by the value of fractional weights in~\Cref{eq:main-type-i}}
	\end{align*}
	Since we have $\bAL_i \subseteq \hbLL{i}$ and $\hsizes_L(i)= s_i$, vertex $\bv$ satisfies the matching constraint. 
	
\paragraph{Vertices in $\bBL_i$ for $i \in [K]$:} Any vertex $\bv \in \bBL_i$ is only incident on type $i$ edges and is connected to all vertices in $\bAR_i$. We have, 
	\begin{align*}
		\sum_{\be \ni \bv} \hbw(\be) &= \card{\bAR_i} \cdot  \frac{s_i}{\card{\bBR_i}} \tag{by the value of fractional weights in~\Cref{eq:main-type-i}} \\
		&= \frac{t}{K+t-i+1}  \cdot value(\bG) \cdot s_i \cdot (K+t) \tag{by~\Cref{obs:main-sizes}} \\
		&= s_{K+1} \tag{by the definition of $s_i$ and $s_{K+1}$ in~\Cref{eq:main-s}}.
	\end{align*}
	Since $\bBL_i \subseteq \hbLL{K+1}$ and $\hsizes_L(K+1) = s_{K+1}$, vertex $\bv$ satisfies the matching constraint. 
	
\paragraph{Vertices in $\bCL$:} Any vertex $\bv \in \bCL$ is only incident on type $K+1$ edges 
	and is connected to all vertices in $\bDR$. We have, 
	\begin{align*}
		\sum_{\be \ni \bv} \hbw(\be) &= \card{\bDR} \cdot \frac{s_{K+1}}{\card{\bDR}} = s_{K+1}. \tag{by the value of fractional weights in~\Cref{eq:main-type-k+1}}
	\end{align*}
	As before, since $\bCL \subseteq \hbLL{K+1}$, we are done. 
	
\paragraph{Vertices in $\bDL$:} Any vertex $\bv \in \bDL$ is only incident on type $K+1$ edges 
	and is connected to all vertices in $\bCR$. We have, 
	\begin{align*}
		\sum_{\be \ni \bv} \hbw(\be) &= \card{\bCR} \cdot \frac{s_{K+1}}{\card{\bDR}} \tag{by the value of fractional weights in~\Cref{eq:main-type-k+1}} \\
		&= \frac{t}{K+t} \cdot (1-\eps) \cdot (1-value(\bG)) \cdot \frac{s_{K+1}}{\eps} \tag{by~\Cref{obs:main-sizes}} \\
		&= s_{K+2}. \tag{by the definition of $s_{K+1}$ and $s_{K+2}$ in~\Cref{eq:main-s}}
	\end{align*}
	Since $\bDL \subseteq \hbLL{K+2}$ and $\hsizes_L(K+2) = s_{K+2}$, we are done in this case also. 
	
\paragraph{Vertices in $\bOL$:} Any vertex $\bv \in \bOL$ is only incident on type $K+2$ edges 
	and has the same edge as some vertex in $\bG$. Since we have rescaled the weight of this edge by $s_{K+2}$ and $\hsizes_L(K+2) = s_{K+2}$, we obtain that $\bv$ also satisfies the matching constraint. 

Any remaining vertex has no edges in the blueprint and thus satisfies the matching constraint by default.
In summary, the blueprint $\hbG$ satisfies the matching constraint. \Qed{clm:main-matching-constraint}
\end{proof}

\begin{claim}\label{clm:main-ban-constraint}
	Blueprint $\hbG$ satisfies the ban constraint.
\end{claim}
\begin{proof}
We prove this claim by going over all different types of edges in the blueprint.  

\paragraph{Type $i \in [K]$ edges:} Consider any type $i$ edge $\be$ between $\bAL_i$ and $\bBR_i$ (the other case of $\bBL_i$ and $\bAR_i$ is symmetric). This edge belongs to $\hbEE{i}$ and is between $\hbLL{i}(x)$ and $\hbRR{K+1}(y)$
	for some $x,y \in [C]^{\hP}$. By~\Cref{def:proper-blueprint}, the edge $\be$ introduces a ban between all vertices
	\[
		\hbLL{i}(x_1,\ldots,x_{i-1},z_1,\ldots,z_{\hP-i+1}) \quad \text{and} \quad \hbRR{K+1}(y_1,\ldots,y_{i-1},z_1,\ldots,z_{\hP-i+1}),
	\]
	for any choices of $(z_1,\ldots,z_{\hP-i+1})$ in $[C]^{\hP-i+1}$ -- the ban implies that between each pair, at most one can have any incident edges. 
	
	The only vertices $\hbRR{K+1}(y)$ that appear in these bans have 
	\[
		y_1 \in \Top_1~,~y_2 \in \Top_2~,~\ldots~,~ y_{i-1} \in \Top_{i-1}. 
	\]
	Fix any such vertex $\bu$. For $\bu$ to have any edge in $\hbG$ it should belong to one of these two sets: 
	\begin{itemize}
		\item for at least one index $j \in [i:K]$, the $j$-th index of $\bu$ is in $\Bot_j$ (so $\bu$ has a type $j$ edge); \emph{or}, 
		\item the last $P$ indices of $\bu$ belong to $\Free_L$ (so $\bu$ has a type $K+1$ edge);
	\end{itemize}
	(if $\bu$ is matched by a type $<i$ edge, it is not part of the banned pairs imposed by a type $i$ edge).  
	
	On the other hand, any vertex $\bv \in \hbLL{i}$ can only have an edge if it is in $\bAL_i$
	and thus:
	\begin{itemize}
		\item	for all indices $j \in [i:K]$, the $j$-th index of $\bv$ belongs to $\Top_j$; \emph{and}, 
		\item the last $P$ indices of $\bv$ belong to $\Match_L$.
	\end{itemize}
	Given the two conditions above are complementary, we have that among the banned pairs above, at most one endpoint of each pair can have any edge, ensuring the ban constraint of type $i$ edges is respected. 
	
\paragraph{Type $K+1$ edges:} Consider any type $K+1$ edge $\be$ between $\bCL$ and $\bDR$ (the other case of $\bDL$ and $\bCR$ is symmetric). This edge belongs to $\hbEE{K+1}$ and is 
	between $\hbLL{K+1}(x)$ and $\hbRR{K+2}(y)$ for some $x,y \in [C]^{\hP}$. By~\Cref{def:proper-blueprint}, the edge $\be$ introduces a ban between all vertices 
	\[
		\hbLL{K+1}(x_1,\ldots,x_{K},z_1,\ldots,z_{\hP-K}) \quad \text{and} \quad \hbRR{K+2}(y_1,\ldots,y_{K},z_1,\ldots,z_{\hP-K}),
	\]
	for any choices of $(z_1,\ldots,z_{\hP-K}) \in [C]^{\hP-K}$. Specifically, the only vertices $\hbLL{K+1}(x)$ that appear in these bans have 
	\[
		x_1 \in \Top_1~,~x_2 \in \Top_2~,~\ldots~,~x_K \in \Top_K. 
	\]
	Fix any such vertex $\bu$. For $\bu$ to have any edge in $\hbG$, it can only be part of $\bCL$ (vertices in any $\bBL_i$ require $x_i$ to be in $\Bot_i$ instead), 
	and thus should satisfy that: 
	\begin{itemize}
		\item index $K+1$ of $\bu$ is in $\Old$; \emph{and},
		\item the last $P$ indices of $\bu$ are in $\Free_R$;
	\end{itemize}

	On the other hand, any vertex $\bv \in \hbRR{K+2}$ belongs to one of the following two sets: 
	\begin{itemize}
		\item index $K+1$ of $\bv$ is in $\New$ (so $\bv$ has a type $K+1$ edge); \emph{or}, 
		\item the last $P$ indices of $\bv$ are in $\Match_R$ (so $\bv$ has a type $K+2$ edge).
	\end{itemize}
	Given the two conditions above are complementary, we have that among the banned pairs above, at most one endpoint of each pair can have any edge, ensuring the ban constraint of type $K+1$ edges is respected. 
	
\paragraph{Type $K+2$ edges:} Any edge $\be$ of type $K+2$ is between two vertices $\hbLL{K+2}(x)$ and $\hbRR{K+2}(y)$ and belongs to $\hbEE{K+1+p}$ for some $p \in [P]$. 
	By construction, we know $(x_1,\ldots,x_{K+1}) = (y_1,\ldots,y_{K+1})$ and that $x_{K+1}=y_{K+1} \in \Old$. 

	Any ban introduced by this edge will be only between vertices of 
	\[
	\hbLL{K+2}(x_1,\ldots,x_{K+1},\underbrace{*,\ldots,*}_{P}) \quad \text{and} \quad \hbRR{K+2}(x_1,\ldots,x_{K+1},\underbrace{*,\ldots,*}_P).
	\] 
	But then the subgraph between these two sets is exactly as in $\bG$ and thus since $\bG$ was a proper blueprint, this subgraph will also satisfy all its (internal) ban constraints. 

In summary, the blueprint $\hbG$ satisfies the ban constraint also. \Qed{clm:main-ban-constraint}
\end{proof}

We have established that $\hbG$ is indeed a proper blueprint. It thus remains to bound its value. 

\begin{claim}[``Value of $\hbG$ increases'']\label{clm:main-value}
	\[
		value(\hbG) > value(\bG) + \delta^3. 
	\]
\end{claim}
\begin{proof}
We have, 
\begin{align*}
	value(\hbG) &= \frac{2 \cdot \sum_{\be \in \hbE} \hbw(\be)}{C^{\hP} \cdot \paren{\sum_{b \in [B_L]} \hsizes_L(b)+\sum_{b' \in [B_R]}\hsizes_R(b')}} \tag{by the definition of $value(\cdot)$ in~\Cref{def:blueprint-value-approx}}\\
	&= \frac{\sum_{\be \in \hbE} \hbw(\be)}{C^{\hP} \cdot \paren{\sum_{b \in [K+2]} s_{b}}} \tag{as $B_L = B_R$ and $\hsizes_L(b) = \hsizes_R(b)$ for all $b$} \\
	&=  \frac{\sum_{i=1}^{K+2} \sum_{\be:~type~i} \hbw(\be)}{C^{\hP} \cdot \paren{\sum_{b \in [K+2]} s_{b}}} \tag{by partitioning across the types of edges} \\
	&= \frac{\paren{\sum_{i=1}^{K} \card{\bAL_i} \cdot s_i + \card{\bBL_i} \cdot s_{K+1}} + \card{\bCL} \cdot s_{K+1} + \card{\bDL} \cdot s_{K+2} + C^{K+1} \cdot (1-\eps) \cdot \card{\bE} \cdot s_{K+2} }{C^{\hP} \cdot \paren{\sum_{b \in [K+2]} s_{b}}} 
	\tag{by the matching constraint for the type $<K+2$ edges and using~\Cref{eq:main-we-k+2} for type $K+2$ edges} \\
	&=  \frac{\paren{\sum_{i=1}^{K} 2} + 2t \cdot (1-\eps) \cdot (1-value(\bG)) + (1-\eps)^2 \cdot value(\bG) \cdot \frac{1}{\eps} \cdot (1-value(\bG)) \cdot t }{\paren{\sum_{i=1}^{K} \frac{K+t-i+1}{t} \cdot \frac{1}{value(\bG)}} + K+t + \frac{1}{\eps} \cdot (1-\eps) \cdot (1-value(\bG)) \cdot t } \tag{by~\Cref{eq:main-s} and~\Cref{obs:main-sizes}} \\
	&= \frac{2K + 2t \cdot (1-\eps) \cdot (1-value(\bG)) + (1-\eps)^2 \cdot value(\bG) \cdot \frac{1}{\eps} \cdot (1-value(\bG)) \cdot t }{\paren{\frac{K}{value(\bG)} + \sum_{i=1}^{K} \frac{K-i+1}{t} \cdot \frac{1}{value(\bG)}} + K+t + \frac{1}{\eps} \cdot (1-\eps) \cdot (1-value(\bG)) \cdot t } \tag{by simplifying the terms} \\
	&=  \frac{2K + 2t \cdot (1-\eps) \cdot (1-value(\bG)) + (1-\eps)^2 \cdot value(\bG) \cdot \frac{1}{\eps} \cdot (1-value(\bG)) \cdot t }{{\frac{K}{value(\bG)} + \frac{K \cdot (K+1)}{2t} \cdot \frac{1}{value(\bG)}} + K+t + \frac{1}{\eps} \cdot (1-\eps) \cdot (1-value(\bG)) \cdot t } \tag{by calculating the sum in the denominator}. 
\end{align*}

We can plug in the value of $K = (1-value(\bG)) \cdot t$ in~\Cref{eq:main-tkeps} and have that $value(\hbG)$ is
\begin{align*}
	&\frac{ 2 \cdot (1-value(\bG)) \cdot t + 2t \cdot (1-\eps) \cdot (1-value(\bG)) + (1-\eps)^2 \cdot value(\bG) \cdot \frac{1}{\eps} \cdot (1-value(\bG)) \cdot t }{{\frac{(1-value(\bG)) \cdot t}{value(\bG)} + \frac{(1-value(\bG)) \cdot ((1-value(\bG)) \cdot t+1)}{2} \cdot \frac{1}{value(\bG)}} + (2-value(\bG)) \cdot t + \frac{1}{\eps} \cdot (1-\eps) \cdot (1-value(\bG)) \cdot t }. 
\end{align*}
This in turn can be simplified to, after a division by $t$, 
\begin{align*}
		&\frac{ 2 \cdot (1-value(\bG))+ 2 \cdot (1-\eps) \cdot (1-value(\bG)) + (1-\eps)^2 \cdot value(\bG) \cdot \frac{1}{\eps} \cdot (1-value(\bG))}{{\frac{(1-value(\bG))}{value(\bG)} + \frac{(1-value(\bG))^2}{2 \cdot value(\bG)} + \frac{1}{2t \cdot value(\bG)}} + (2-value(\bG))+ \frac{1}{\eps} \cdot (1-\eps) \cdot (1-value(\bG))},
\end{align*}
where the only term that still has a $t$-factor is in the denominator (coming from $+1$ part of the $K \cdot (K+1)/2$ term earlier). 

Since our goal is to establish $value(\hbG) > value(\bG) + \delta^3$, we would now require that 
\begin{align*}
	&2 \cdot (1-value(\bG))+ 2 \cdot (1-\eps) \cdot (1-value(\bG)) + (1-\eps)^2 \cdot value(\bG) \cdot \frac{1}{\eps} \cdot (1-value(\bG)) \\
	& \hspace{0.5\linewidth} >  \\
	 &{{(1-value(\bG))} + \frac{(1-value(\bG))^2}{2} + \frac{1}{2t} + value(\bG) \cdot (2-value(\bG))+ value(\bG) \cdot \frac{1}{\eps} \cdot (1-\eps) \cdot (1-value(\bG))} \\
	 &\hspace{0.5\linewidth} + \frac{\delta^3}{\eps} \cdot \eta,  
\end{align*}
for some absolute constant $\eta > 0$ since the denominator of the earlier term is upper bounded $\eta/\eps$ given $1/2 \leq value(\bG) \leq 2/3$. 
After simplifying, the above inequality becomes equivalent to 
\begin{align*}
	&2 \cdot (1-value(\bG))+ 2 \cdot (1-\eps) \cdot (1-value(\bG)) - value(\bG) \cdot (1-value(\bG)) - \eps \cdot (1-value(\bG))\\
	& \hspace{0.5\linewidth} >  \\
	 &{(1-value(\bG))} + \frac{(1-value(\bG))^2}{2} + \frac{1}{2t} + value(\bG) \cdot (2-value(\bG)) +  \frac{\delta^4}{\eps} \cdot \eta.
\end{align*}
At this point, all existing terms with the exception of $3\eps \cdot (1-value(\bG))$ in the LHS and $\frac{1}{2t}$ in the RHS are bounded away from $0$ and $1$; we can thus use a loose approximation 
for those terms also (again using the fact $1/2 \leq value(\bG) \leq 2/3$) and focus on obtaining 
\begin{align*}
	&3 \cdot (1-value(\bG)) - value(\bG) \cdot (3-2 \cdot value(\bG)) - \frac{(1-value(\bG))^2}{2} > \frac{\delta^4}{\eps} \cdot \eta + 3\eps + \frac{1}{2t}. 
\end{align*}
This is equivalent to 
\begin{align*}
	3\cdot value(\bG)^2-10\cdot value(\bG) + 5 > 2\frac{\delta^4}{\eps} \cdot \eta + 6\eps + \frac{1}{t}. 
\end{align*}

Considering the LHS as a quadratic polynomial $P(x)$ in $x$ by defining $x:= value(\bG)$, we get $P(x) = 3x^2-10x+5$ whose roots are $(5/3 - \sqrt{10}/3 , 5/3+\sqrt{10}/3)$ and 
the polynomial is non-negative for $x \leq 5/3-\sqrt{10}/3$. Moreover, for $x \leq 5/3-\sqrt{10}/3-\delta$, we have that $P(x) \geq 2\sqrt{10}\delta$. 
Thus, whenever $value(\bG) \leq (5-\sqrt{10})/3-\delta$, as in the statement of~\Cref{lem:main-amplification}, we have 
\[
		3\cdot value(\bG)^2-10\cdot value(\bG) + 5 \geq 2\sqrt{10}\delta,
\]
and thus we only need to have 
\[
		2\sqrt{10}\delta \geq 2\frac{\delta^4}{\eps} \cdot \eta + 6\eps + \frac{1}{t}, 
\]
to conclude the proof. This can be achieved by setting $t \geq 1/6\eps$ and $\eps = 2\eta\delta^3/\sqrt{10}$ and picking $\delta$ to be sufficiently small. 

Thus, 
we proved that if $value(\bG) \leq (5-\sqrt{10})/3-\delta$, then $value(\hbG) > value(\bG)+\delta^3$. \Qed{clm:main-value}
\end{proof}

\Cref{lem:main-amplification} now follows from~\Cref{clm:main-matching-constraint,clm:main-ban-constraint,clm:main-value}.

\clearpage


\section{Concluding Remarks}\label{sec:concluding}

We have proved the following lower bound for semi-streaming matching, formalizing~\Cref{res:main}. 

\begin{theorem}\label{thm:main-streaming}
	Any single-pass streaming algorithm that given every $n$-vertex bipartite graph $G$ can find an $\alpha$-approximate maximum matching of $G$ with high constant
	probability requires $n^{1+\Omega(1/\log\log{n})}$ space for any constant $\alpha$ satisfying
	\[
		\alpha > \frac{8-2\sqrt{10}}{3} \sim 0.55848155\cdots. 
	\]
\end{theorem}
\Cref{thm:main-streaming} follows immediately from~\Cref{thm:framework} (our framework) and~\Cref{thm:main-blueprint} (our main blueprint construction).

The main question that is left open by our work is the one we started with: is 0.5-approximation the best possible for the semi-streaming matching problem? Our~\Cref{thm:main-streaming} makes progress on this 
question by further narrowing the gap on what the \emph{right} approximation ratio can be. However as we stated earlier, we find the main contribution of our work, not in the particular numerical improvement of~\Cref{thm:main-streaming} over~\cite{Kapralov21}, but rather the lower bound framework we put forward in establishing this result. We conclude the paper by discussing several remarks about this framework. 

\paragraph{A plan of attack.}  Our framework presents a concrete plan of attack for settling the semi-streaming complexity of the matching problem: 
\begin{quote}
	\emph{Is there a sequence of proper blueprints with approximation ratio approaching $1/2$?} %
\end{quote}
At present, we have no evidence toward a positive or negative answer to this question. A positive answer, combined with our~\Cref{thm:framework}, immediately settles the longstanding open question
of the semi-streaming matching problem. But, even a negative answer will be quite interesting and can potentially guide us toward new approaches for proving a lower bound for the problem 
or alternatively designing better semi-streaming algorithms. 

Given this question bypasses any streaming arguments, we find it potentially (much) easier to address than directly considering the semi-streaming matching problem. 
Moreover, our results in this paper, at the very least, provide a proof of concept that the blueprints can already lead to stronger lower bounds than prior work~\cite{Kapralov21} with the added benefit 
of employing considerably less technical arguments and thus they merit further study.  

\paragraph{Automating blueprint construction.} At its core, determining the best value/approximation ratio of a (proper) blueprint is simply a \emph{constant-size} optimization problem. 
Thus, in principle, one can ``just'' run a computer program for some large values of parameters (say, only $P,C$ for simple blueprints) to estimate this quantity. Unfortunately, for this we require $P$ to be sufficiently large
to get close to a (potential) $0.5$-approximation lower bound: our~\Cref{lem:framework} proves a lower bound on the $(P+1)$-communication game of matching while
it is known that for any \emph{fixed} $P$, there is a better-than-$0.5$-approximation protocol for this problem~\cite{LeeS17,BehnezhadK22}. Given the exponential dependence of the problem space on $P$, so far, we have not been
able to generate blueprints automatically this way (although, early on we did manage to use computer programs to find small blueprints like the ones in~\Cref{fig:blueprint3} and even its more optimized version with \emph{value} approaching $4-2\sqrt{3} \sim 0.535$).
 
In general, it will be quite interesting to develop tools for generating blueprints automatically and to see what type of approximation ratios one can rule out this way. 
In light of the above discussion, this is perhaps more feasible to apply to \emph{fixed} values of $P$ to prove lower bounds for the $(P+1)$-player one-way communication complexity of the matching 
problem studied in several works~\cite{GoelKK12,AssadiB19,AssadiB21b,BehnezhadK22,DerakhshanGR25}. 

\emph{Remark.} As we showed in~\Cref{sec:equiv-blueprints}, 
simple blueprints can capture the power of all blueprints, with the caveat that the reduction increases the value of $P$ by one. This was inconsequential for us but 
can be problematic if the goal is to prove a lower bound on communication complexity for a \emph{fixed} number of players mentioned above. Nevertheless, the increment in $P$ in~\Cref{prop:simplify-blueprint} 
is more of a technicality and in fact, in the resulting instances of~\Cref{thm:framework}, the first player receives no edges (as $\bEE{1} = \emptyset$ after the reduction) and thus can be eliminated. 
Therefore, one can still use not-simple blueprints of parameter $P$ to prove a lower bound for the $(P+1)$- and not only $(P+2)$-player communication complexity of matchings. 

\paragraph{Optimizing the proof of~\Cref{thm:main-blueprint}.} The constant we get in~\Cref{thm:main-blueprint} (and subsequently in~\Cref{thm:main-streaming}) is in no way sacrosanct; in fact, it is not even fully optimized for our particular approach in
proving~\Cref{thm:main-blueprint} and its~\Cref{lem:main-amplification}. Specifically, as we alluded to in~\Cref{sec:main-high-level}, the gadget we used in proving~\Cref{lem:main-amplification} can be seen as using the approach of our 
weaker blueprint in~\Cref{lem:2-sqrt2-amplification} on the blueprint $\bY$ mentioned there. However, we can also recursively use the approach of~\Cref{lem:main-amplification} itself to increase the value of the blueprint $\bY$ instead. 
While we could numerically estimate the approximation ratio of the blueprint constructed this way to be $\sim 0.548$, calculating the exact value analytically is quite cumbersome. 
And, since it is clear to us this particular way of creating blueprints is not able to get below, say, a
$\sim 0.540$ bound, we have opted to stick with the simpler version of~\Cref{lem:main-amplification} that admits a simple analytical solution.

Regardless of this however, it is an interesting question---and in line with our earlier ones---to develop better analytical tools for analyzing 
recursive constructions of blueprints and in particular the effect of various gadgets and transformations on their values.

\paragraph{Other variants and related problems.} Finally, while our goal in this paper has been solely to make progress on the semi-streaming matching problem, 
our results and approach have or can potentially have implications for other problems as well. For instance, our results can be used, with a minor modification of~\Cref{thm:framework}, to show that blueprints 
also prove a lower bound for the one-way communication complexity and streaming complexity of the minimum vertex cover problem studied in~\cite{DerakhshanGR25};\footnote{The only modification to the proof we need is to define the graph $G$ by picking $\delta_p$-fraction of edges (depending which stage $p \in [P]$ the edge arrives) instead of removing them; then, 
for the final matching $\Mstar$, at most some $\delta$ fraction of edges are in the graph, but the algorithm needs to pick one vertex incident on any edge that it is not sure is \emph{missing} from the graph, which by the same argument as before (applied to the complement graph) will be $(1-\delta)$ fraction of edges of $\Mstar$; we omit the standard details; see, e.g.~\cite{DerakhshanGR25} a similar argument.} as such, we can also prove a lower bound of $(\frac{8-2\sqrt{10}}{3})^{-1} \sim 1.790$-approximation for the streaming vertex cover problem. 

It will be quite interesting to see if blueprints can be used to also make progress on proving lower bounds for the following variants of the semi-streaming matching problem: 
\begin{itemize}
\item \emph{random-order} streaming (which admits a $1-\Theta(1/\log{n})$-lower bound~\cite{AssadiB21}), 
\item \emph{two-pass} algorithms (which admits a $8/9$-approximation lower bound~\cite{KonradN24}), or
\item \emph{multi-pass} algorithms for $(1-\eps)$-approximation (which admits an $\Omega(\log{(1/\eps)})$-pass lower bound~\cite{AssadiS23}\footnote{The lower bound of~\cite{AssadiS23} is proven \emph{conditionally} (under the plausible hypothesis that RS graphs with $n^{1+\Omega(1)}$ edges exist). But, recent results in~\cite{AssadiKNS24,AssadiBKNS25} imply an unconditional version of this lower bound; see also~\cite{KonradN24} that proved an unconditional version
for two-pass algorithms improving the conditional result of~\cite{Assadi22}.}).
\end{itemize}
We leave further applications of blueprints to similar problems as our final open question.

\section*{Acknowledgement}

We thank Peter Kiss for conversations about connections of our work with the online preemptive matching literature, and Trevor Brown for
providing us with computational resources for generating blueprints efficiently in the early stages of this project. Finally, we are thankful to Richard Peng for connecting the authors in Fall 2023, which led to this collaboration.

Sepehr Assadi would like to also thank Soheil Behnezhad, Aaron Bernstein, Niv Buchbinder, Deeparnab Chakrabarty, Michael Kapralov, Sanjeev Khanna, Christian Konrad, Lap Chi Lau, Thatchaphol Saranurak, and David Wajc for various discussions
 about the semi-streaming matching problem over the years. 

Part of this work was conducted while Sepehr Assadi was visiting the Simons Institute for the Theory of Computing as part of the Sublinear Algorithms program.

\section*{AI Acknowledgement} 

No AI tools were used in any capacity in the research or preparation of this paper, except possibly for minor typo editing.

\bibliographystyle{halpha-abbrv}
\bibliography{new}

\clearpage
\appendix

\part*{Appendix}

\section{Proof of the Lower Bound for Lossy Compression}\label{app:info}

We present the proof of~\Cref{prop:compression}, restated below, in this appendix. 

\begin{proposition*}
	Fix any graph $\Gbase$ on $m$ edges, any $\delta \in (0,1/2)$, and any $s \geq \log{(m+1)}$. Consider any compression scheme $\Phi$ of size $s$ for $\GG = \GG(\Gbase,\delta)$. 
	For $G$ sampled from $\GG$, 
	\[
		\Exp_{G \sim \GG} \card{\belong{\Ebase}{\Phi(G)}} \leq  \frac{16s}{\delta}.
	\]
\end{proposition*}

Before getting to the proof, we need a quick refresher on basic information theory tools. 

\subsection{Refresher on Entropy and its Basic Properties}

For a random variable $A$, we use $\supp{A}$ to denote the support of $A$ and $\distribution{A}$ to denote its distribution. 
When it is clear from the context, we may abuse the notation and use $A$ directly instead of $\distribution{A}$, for example, write 
$a \sim A$ to mean $a \sim \distribution{A}$, i.e., $a$ is sampled from the distribution of random variable $A$.

We denote the \emph{Shannon Entropy} of a random variable $A$ by
$\en{A}$, which is defined as: 
\begin{align}
	\en{A} := \sum_{a \in \supp{A}} \Pr\paren{A = a} \cdot \log{\paren{1/\Pr\paren{A = a}}} \label{eq:entropy}
\end{align} 

\noindent
The \emph{conditional entropy} of $A$ conditioned on $B$ is denoted by $\en{A \mid B}$ and defined as:
\begin{align}
\en{A \mid B} := \Ex_{b \sim B} \bracket{\en{A \mid B = b}}, \label{eq:cond-entropy}
\end{align}
where 
$\en{A \mid B = b}$ is defined by using the distribution of $A$ conditioned on the event $B = b$. 

We shall use the following basic properties of entropy. 
Proofs of these properties follow from convexity of the entropy function
and Jensen's inequality and can be found in~\cite[Chapter~2]{CoverT06}. 

\begin{fact}\label{fact:it-facts}
  Let $A$, $B$, and $C$ be (possibly correlated) random variables.
   \begin{enumerate}
  \item \label{part:uniform} $0 \leq \en{A} \leq \log{\card{\supp{A}}}$. The right (resp. left) equality holds
    iff $\distribution{A}$ is uniform over its support $\supp{A}$ (resp. is deterministic). 
    \item \label{part:sub-additivity} \emph{Subadditivity of entropy}: $\en{A,B \mid C}
    \leq \en{A \mid C} + \en{B \mid  C}$.
   \item \label{part:chain-rule} \emph{Chain rule for entropy}: $\en{A,B \mid C} = \en{A \mid C} + \en{B \mid C,A}$.
   \end{enumerate}
\end{fact}

Furthermore, for any $\delta \in (0,1)$, we write $H_2(\delta)$ to denote the \emph{binary entropy} of $\delta$, i.e., the entropy of the Bernoulli random variable with mean $\delta$. 
We further define the \emph{KL-divergence} between two Bernoulli random variables with means $\delta,\delta' \in (0,1)$ as 
\[
	D_2\paren{\delta' \, || \, \delta} := \delta' \cdot \log{\frac{\delta'}{\delta}} + (1-\delta') \cdot \log{\frac{(1-\delta')}{(1-\delta)}}. 
\]
We use an ad hoc lower bound on the value of KL-divergence which for our particular setting is quadratically stronger than the standard Pinsker's inequality. 
\begin{fact}\label{fact:kl-bernoulli}
	For any $\delta \in (0,1)$, 
	$
		D_2\paren{\frac\delta2 \, || \, \delta} \geq \frac{\delta}{8}. 
	$
\end{fact}
\begin{proof}
	By definition,
	\[
		D_2\paren{\frac\delta2 \, || \, \delta} = \frac{\delta}{2} \cdot \log{\frac{1}{2}} + \paren{1-\frac{\delta}{2}} \cdot \log{\paren{1 + \frac{\delta/2}{1-\delta}}}.
	\]
	Using $\ln{(1+x)} \geq \frac{x}{1+x}$ for $x \geq 0$ with $x = \frac{\delta/2}{1-\delta}$, the second term is at least $\frac{\delta}{2\ln{2}}$. Hence,
	\[
		D_2\paren{\frac\delta2 \, || \, \delta} \geq \frac{\delta}{2} \cdot \paren{\frac{1}{\ln{2}} - 1} \geq \frac{\delta}{8}. \qed
	\]
	
\end{proof}

\subsection{Proof of~\Cref{prop:compression}}

	Let $X \in \set{0,1}^{\Ebase}$ be a random variable over the distribution $\GG$ such that for any edge $e \in \Ebase$, $X_e = 1$ iff $e$ is sampled in $G$ and otherwise $X_e=0$. Additionally, let $S$ be the random variable
	for the summary $\Phi(G)$. Finally, we use $B(S)$ to denote the edges of $\Ebase$ that belong to the summary $\Phi(G)$, i.e., $B(S) = \belong{\Ebase}{S}$. We first have 
	\begin{align}
		\en{X \mid S} \geq \en{X} - \en{S} \geq \log{{\binom{m}{\delta m}}} - s \geq m \cdot H_2(\delta) - \log{(m+1)} - s \geq m \cdot H_2(\delta) - 2s; \label{eq:first-one}
	\end{align}
	the first inequality follows from chain-rule (\itfacts{chain-rule}) and non-negativity (\itfacts{uniform}) of entropy, second one uses the fact that $X$ is uniform over its support and support of $S$ has size $2^s$ (and applies~\itfacts{uniform}), the third one
	uses the standard lower bound for binomial coefficient of the form $\binom{a}{b} \geq 2^{a \cdot H_2(b/a)}/(a+1)$, and the last one uses the lower bound on $s$ in the statement. 
	
	Fix any summary $\phi$ and for any $e \in \Ebase$, define 
	\[
		q_e := \Pr\paren{X_e = 0 \mid S = \phi}. 
	\]
	Since $X$ has exactly $\delta m$ zeros, we have $\sum_{e} q_e = \delta m$. At the same time, by the subadditivity of entropy (\itfacts{sub-additivity}), 
	\[
		\en{X \mid S=\phi} \leq \sum_{e} H_2(q_e) = m \cdot H_2(\delta) - \sum_{e} D_2\paren{q_e \,||\, \delta}, 
	\]
	because
	\begin{align*}
		\sum_e D_2\paren{q_e \,||\, \delta} &= \sum_e q_e \cdot \log{\frac{q_e}{\delta}} + (1-q_e) \cdot \log{\frac{(1-q_e)}{(1-\delta)}} \tag{by the definition of KL-divergence} \\
		&= \sum_e \paren{q_e \cdot \log{\frac{1}{\delta}} + (1-q_e) \cdot \log{\frac{1}{(1-\delta)}} - H_2(q_e)} \tag{by the definition of binary entropy} \\
		&= \sum_e \paren{\delta \cdot \log{\frac{1}{\delta}} + (1-\delta) \cdot \log{\frac{1}{(1-\delta)}} - H_2(q_e)} \tag{by linearity and since $\sum_{e} q_e = \delta  m = \sum_e \delta$} \\
		&= m \cdot H_2(\delta) - \sum_e H_2(q_e) \tag{by the definition of binary entropy}. 
	\end{align*}
	Given that for $e \in B(\phi)$, we have $q_e \leq \delta/2$ and by~\Cref{fact:kl-bernoulli} and monotonicity of $D_2(\cdot \,||\, \delta)$ in $[0,\delta]$, 
	\[
		\en{X \mid S=\phi} \leq m \cdot H_2(\delta) - \sum_{e \in B(\phi)} D_2\paren{\frac\delta2 \,||\, \delta} \leq m \cdot H_2(\delta) - \card{B(\phi)} \cdot \frac{\delta}{8}. 
	\]
	We thus have 
	\[
		m \cdot H_2(\delta) - 2s \leq \en{X \mid S} = \Exp_{\phi \sim S}\bracket{\en{X \mid S=\phi}} \leq m \cdot H_2(\delta) - \Exp\card{B(\phi)} \cdot \frac{\delta}{8}, 
	\]
	which implies that 
	\[
		\Exp_{G \sim \GG} \card{\belong{\Ebase}{\Phi(G)}} = \Exp_{\phi \sim S} \card{B(\phi)} \leq \frac{16s}{\delta}, 
	\]
	concluding the proof. \Qed{prop:compression}

\clearpage 

\section{A Construction of Somewhat-Dense ERS Graphs}\label{app:ers}

We present a proof of~\Cref{prop:base-c-rs} restated below. 

\begin{theorem*}
	Fix any integer $C \geq 2$ and parameter $\delta \in (0,1/(40C))$. There exists an integer $n_0 := n_0(C,\delta)$ such that for infinitely many values of $n \geq n_0$, 
	there are $(C,r,t)$-ERS graphs with $n$ vertices on each side of the bipartition with parameters $t = n^{\Omega(1/\log\log{n})}$ and $r = n/C - \delta \cdot n$. 
\end{theorem*}

We fix the following parameters throughout the proof: 
\begin{itemize}
    \item parameter $\delta \in (0,1/40C)$ such that $1/\delta$ is an integer;
    \item sufficiently large integer $d \geq 1$ and $p := d^2$ such that $1/\delta \mid d$ (so $\delta \cdot d$ is an integer) and $C \cdot (1+10C\delta) \cdot d \mid p$; given $p=d^2$, this
    is equivalent to specifying that there exists an integer $k$ where $ k\cdot C(1+10C \delta) = d$,
    which is satisfied by infinitely many integers $d$ such that $1/\delta \mid d$.
\end{itemize} 
In the following, we fix $\delta$ and allow $d$ to go to infinity with respect to $\delta$ and thus our asymptotic notation is with respect to $d$ and suppresses dependence on $\delta$.

We have the following claim using random sets. It immediately follows from standard results in coding-theory and its simple proof is provided here for completeness. 
\begin{claim}\label{clm:random-sets-base-C}
	There exists $t := 2^{\Omega(d)}$ sets $S_1, \ldots, S_t \subseteq [d]$ such that for $i \neq j \in [t]$: 
	\[
		\card{S_i} = \delta \cdot d \qquad \text{and} \qquad  \card{S_i \cap S_j} \leq 4 \delta^2 \cdot d.
	\]
\end{claim}
\begin{proof}
	We pick the sets $S_1,\ldots,S_t$ randomly, meaning that each $S_i$ is a $(\delta \cdot d)$-subset of $[d]$ chosen uniformly and independently. Thus, the sets satisfy the desired size constraint trivially. We now bound the intersection size of 
	any pairs of sets $S_i,S_j$ for $i \neq j$ and union bound over all choices to conclude the proof. 
	
	Fix the elements in $S_i$ and consider the process of sampling $S_j$ which is independent of $S_i$. Each element of $S_i$ will belong to $S_j$ with probability $\delta$ and thus 
	\[
		\Exp\card{S_i \cap S_j} = \card{S_i} \cdot \delta = \delta^2 \cdot d. 
	\]
	Moreover, the distribution of $S_i \cap S_j$ is a hypergeometric random variable and thus by Chernoff bound for negatively correlated random variables 
	\[
		\Pr\paren{\card{S_i \cap S_j} \geq 4\Exp\card{S_i \cap S_j}} \leq \exp\paren{-2\Exp\card{S_i \cap S_j}} = \exp\paren{-2\delta^2 \cdot d}. 
	\]
	Thus, as long as $t$ satisfies 
	\[
		{{t}\choose{2}} < \exp\paren{2\delta^2 \cdot d}
	\]
	we can union bound over all pairs $i \neq j \in [t]$ and conclude existence of the desired collection of sets. This implies setting $t = 2^{\Omega(d)}$ concludes the proof. \Qed{clm:random-sets-base-C}
	
\end{proof}

We now define our ERS graph. Fix a collection of sets $S_1,\ldots,S_t$ as in~\Cref{clm:random-sets-base-C} for the rest of the proof (these sets will subsequently define the matching-groups $M_1,\ldots,M_t$ in the ERS graph). 

Define the vertices as $L = R = [p]^d$, and let $n = \card L = \card R = p^d$. 
Define the integers:
\[
B = d, \quad W = 10C\delta \cdot d, \quad P = B+W, \quad \text{and} \quad Q = C \cdot P. 
\]
For every $i \in [t]$ and vector $v \in [p]^d$, we define 

\begin{align*}
    \weight_i(v) &= \sum_{k \in S_i} v_k
\end{align*}
and
\begin{align*}
    \ccolor_i(v) &= 
    \begin{cases} 
        x & \text {if } \weight_i(v) - (x-1) \cdot P \bmod Q \in [0, B) \text{ for } x \in [C]\\
        0 & \text{otherwise}
    \end{cases}.
\end{align*}

In other words, for a vector $v$ with $\weight_i(v) \bmod P \in [B, P)$, we assign it color $0$. Then, we assign the vectors $v$ with $\weight_i(v) \bmod Q$ in $[0,B)$ color $1$, vectors $v$ with weights $\weight_i(v) \bmod Q$ in $[P, P+B)$ color $2$, and so on. 

From here on, we use vectors in $[p]^d$ and vertices (of either $L$ or $R$) interchangeably, so we will talk about the $\weight_i(\cdot)$ or $\ccolor_i(\cdot)$ of vertices. We now define the decomposition of vertices for each $S_i$ into $C$ groups on each side (which forms the decomposition of the matching-group $M_i$). For each $x \in [C]$, define:
\begin{align*}
L_{i,x} &= \set{v \in [p]^d : \ccolor_i(v) = x} \qquad \text{and} \qquad R_{i,x} = \set{v \in [p]^d : \ccolor_i(v) = x}.
\end{align*}

We then define the matching-groups $M_i$ and their matchings $M_{i,x,y}$ for $i \in [t]$ and $x,y \in [C]$. 
Let $\mathbb{1}_{S_i} \in [p]^d$ be the indicator vector for $S_i$, and for each $x, y \in [C]$, define: 
\begin{align*}
M_{i,x,y} &= \set{(u, v) \in L_{i,x} \times [p]^d : v = u+\mathbb{1}_{S_i} \cdot \paren{(y-x) \cdot (10C+1/\delta) + Q/(\delta d)}}.
\end{align*}
Note that $Q/(\delta d) = C(1+10C\delta)/\delta = C/\delta + 10C^2$ is an integer since $1/\delta$ is an integer by our choice of parameters.

The following claim allows us to establish the \textbf{decomposition condition} of each $M_i$. 

\begin{claim}\label{clm:base-C-decomp}
For any $i \in [t]$ and $x, y \in [C]$, the matching $M_{i,x,y}$ is subset of $L_{i,x} \times R_{i,y}$.
\end{claim}
\begin{proof}
    We have $\weight_i(\mathbb 1_{S_i}) = \delta d$ as $|S_i| = \delta d$ and each element of $S_i$ contributes $1$ to $\weight_i(\cdot)$. Then, from linearity of the function $\weight_i(\cdot)$, we have that for all $x, y \in [C]$,
	\begin{align*}
		\weight_i(\mathbb 1_{S_i}((y-x)(10C+1/\delta) + Q/(\delta d))) &= ((y-x) \cdot (10C+1/\delta) + Q/(\delta d)) \cdot \delta d\\
        &= (y-x) \cdot (10C\delta d + d) + Q \\
        &= (y-x) \cdot (W+B) + Q \\
        &= (y-x) \cdot P + Q.
	\end{align*}
	Since $Q \equiv 0 \pmod{Q}$, this means the weight increases by $(y-x) \cdot P \bmod Q$. Thus, for all $u \in [p]^d$ for which $\ccolor_i(u) = x$, we have $\ccolor_i(u + \mathbb 1_{S_i}((y-x)(10C+1/\delta) + Q/(\delta d))) = x + (y-x) = y$, as long as that vector is still in $[p]^d$. Given the definition of $R_{i,y}$ for $x, y \in [C]$, we can conclude the proof. \Qed{clm:base-C-decomp}
	
\end{proof}

The next claim similarly allows us to establish the \textbf{inducedness condition} of each $M_i$. 
\begin{claim}\label{clm:base-C-induced}
For any $i \neq j \in [t]$, any edge of $M_j$ that is between vertices $L_i$ and $R_i$ of $M_i$ is between $L_{i,x}$ and $R_{i,x}$ for some $x \in [C]$. 
\end{claim}
\begin{proof}
    For any $i \neq j$, we know $|S_i \cap S_j| \leq 4 \delta^2 d$ by~\Cref{clm:random-sets-base-C}, which means $\weight_i(\mathbb 1_{S_j}) \leq 4 \delta^2 d$, and consequently, for any $x, y \in [C]$,
    \begin{align*}
    \card{\weight_i(\mathbb{1}_{S_j}((y-x)(10C+1/\delta) + Q/(\delta d)))} &\leq \card{y-x} \cdot (10C+1/\delta) \cdot 4 \delta^2 d + Q/(\delta d) \cdot 4\delta^2 d \\
    &\leq 4C\delta d \cdot (1+10C\delta) + 4C\delta d \cdot (1+10C\delta) \\
    &< 5C\delta d + 5C\delta d \\
    &= W,
    \end{align*}
    where we used $\card{y-x} \leq C$, $Q/(\delta d) = C/\delta + 10C^2 = (C/\delta)(1+10C\delta)$, and $1 + 10C\delta < 5/4$ since $\delta < 1/(40C)$.

    Consider an edge $(u,v) \in M_{j,x,y}$ for $x, y \in [C]$. By the definition of edges in $M_j$, the linearity of $\weight_i(\cdot)$, and the last equation, we have \begin{align*}
    	\card{\weight_i(v) - \weight_i(u)} &= \card{\weight_i(\mathbb 1_{S_j}((y-x)(10C+1/\delta) + Q/(\delta d)))} \\
        &< W.
    \end{align*}
        
    For any two vertices $u, v$ that have different nonzero $\ccolor_i(\cdot)$, their $\weight_i(\cdot)$ values differ by at least $W$ by the definition of $\ccolor_i(\cdot)$. This in turn implies that if $\ccolor_i(u) = x$ for some $x \in [C]$, then $\ccolor_i(v)$ is either $x$ or $0$ -- we emphasize that we are determining `color' of these vertices with respect to $i$ and not $j$. This concludes the proof by the way the decompositions were defined. \Qed{clm:base-C-induced}
\end{proof}

Finally, we also need to bound the \textbf{sizes of each of the matchings} $M_{i,x,y}$ for $i \in [t]$ and $x,y \in [C]$. The following two claims address this. 

\begin{claim}\label{clm:base-C-size1}
For any $x \in [C]$, we have $\card{\set{v \in [p]^d \mid \ccolor_i(v) = x}} \geq (n/C) \cdot (1-10C\delta)$.
\end{claim}
\begin{proof}
    By the definitions of $Q, P, W, B$, we find that
    \begin{align*}
        Q &= CP = C(B+W) = C(d+10C\delta d).
    \end{align*}
    Recall that $Q = C(d+10C\delta d) \mid p$ by our choice of parameters, so the last coordinate of a vector determines its $\weight_i(\cdot) \bmod Q$ if the other coordinates are fixed, and each value of the $\weight_i(\cdot) \bmod Q$ is achieved by equally many last coordinates, and thus, equally many vectors. The $\ccolor_i(\cdot)$ of a vector is determined by its $\weight_i(\cdot) \bmod Q$, and exactly $B$ of $Q$ values in $[0,Q)$ cause a vector with that $\weight_i(\cdot) \bmod Q$ to be colored $x$. Thus, the proportion of vectors colored $x$ is \[B/Q = d / (C(d+10C\delta d)) = 1/(C(1+10C\delta)) \geq (1-10C\delta)/C,\] concluding the proof. \Qed{clm:base-C-size1}
\end{proof}

\begin{claim}\label{clm:base-C-size2}
	For any $i \in [t]$ and $x,y \in [C]$, we have $\card{M_{i,x,y}} \geq (n/C) \cdot (1-10C\delta) - o(n)$.
\end{claim}
\begin{proof}
    Let $i \in [t]$ and $x, y \in [C]$. Consider a vertex $u$ in $L_{i,x}$ that is not matched in $M_{i,x,y}$. Then, there exists a coordinate $k$ such that $(u + \mathbb{1}_{S_i}((y-x)(10C+1/\delta) + Q/(\delta d)))_k \not \in [p]$. Call this condition being ``outside'' on the $k$th coordinate.

    For each $k \in [d]$, the number of vectors $u$ that are outside on the $k$th coordinate is upper bounded by the number of vectors whose $k$-th coordinate is at most $\card{y-x}(10C+1/\delta) + Q/(\delta d)$ away from either $0$ or $p+1$. This number is bounded by $p^{d-1} \cdot \mathcal O_\delta(1) = \mathcal O_\delta(p^{d-1})$. Summing over $k \in [d]$, the number of total vectors outside some coordinate is $O(p^{d-1} \cdot d) = o(p^d) = o(n)$, concluding the proof. \Qed{clm:base-C-size2}
\end{proof}

Finally, we establish the \textbf{disjointness} of matching-groups.
\begin{claim}\label{clm:base-C-disjoint}
For any $i \neq j \in [t]$, we have $M_i \cap M_j = \emptyset$.
\end{claim}
\begin{proof}
    For any edge $(u,v) \in M_{i,x,y}$, we have $v - u = \mathbb{1}_{S_i} \cdot ((y-x)(10C+1/\delta) + Q/(\delta d))$, which is a nonzero scalar multiple of $\mathbb{1}_{S_i}$ since $(y-x)(10C+1/\delta) + Q/(\delta d) \neq 0$ (as $\card{y-x} \leq C-1 < C$ and $Q/(\delta d) = C/\delta + 10C^2 > (C-1)(10C + 1/\delta)$). Thus, $v - u$ is nonzero precisely on the coordinates in $S_i$.

    Suppose for contradiction that $(u,v) \in M_i \cap M_j$. Then $v - u$ is nonzero precisely on $S_i$ and also precisely on $S_j$, which implies $S_i = S_j$. But $\card{S_i \cap S_j} \leq 4\delta^2 d < \delta d = \card{S_i}$ by~\Cref{clm:random-sets-base-C}, so $S_i \neq S_j$, a contradiction. \Qed{clm:base-C-disjoint}
\end{proof}

We are now ready to conclude the proof of~\Cref{prop:base-c-rs}. 
\begin{proof}[Proof of~\Cref{prop:base-c-rs}]
	We construct our graphs on $n=p^d = d^{2d}$ vertices on each side of the bipartition as explained above. This implies that for any fixed $\delta \in (0,1/40C)$, we have 
	\[
		d = \Omega\paren{\frac{\log{n}}{\log\log{n}}}, 
	\]
	which in turn means for the choice of $t = 2^{\Omega(d)}$ guaranteed by~\Cref{clm:random-sets-base-C},
	\[
		t = n^{\Omega(1/\log\log n)}. 
	\]
	We now verify the main properties of the graph according to \Cref{def:base-rs}: 
	\begin{itemize}[leftmargin=5pt]
    \item \textnormal{\textbf{Size:}} For $i \in [t], x, y \in [C]$, we get $\card{M_{i,x,y}} \geq (n/C) \cdot (1-10C\delta) - o(n) \geq n/C - 11\delta n$ by~\Cref{clm:base-C-size1,clm:base-C-size2} and since $C$ is a constant.
    \item \textnormal{\textbf{Decomposition:}} For $i \in [t]$ and $x, y \in [C]$, the matching $M_{i,x,y}$ is a matching between $L_{i,x}$ and $R_{i,y}$ by~\Cref{clm:base-C-decomp}.
    \item \textnormal{\textbf{Inducedness:}} For any $i \in [t]$, any edge of $E \setminus M_i$ that is between vertices $L_i$ and $R_i$ of $M_i$ can only be between some pair $L_{i,x}$ and $R_{i,x}$ for some $x \in [C]$ by~\Cref{clm:base-C-induced}.
    \item \textnormal{\textbf{Disjointness:}} For any $i \neq j \in [t]$, we have $M_i \cap M_j = \emptyset$ by~\Cref{clm:base-C-disjoint}.
\end{itemize}
This concludes the proof of~\Cref{prop:base-c-rs} after re-parameterizing $\delta \leftarrow \delta/11$. \Qed{prop:base-c-rs}
\end{proof}


\clearpage

\section{Omitted Proofs}\label{app:omitted}

\subsection{From~\Cref{sec:blueprints}}\label{app:omitted-blueprints}

\subsubsection{Proof of~\Cref{lem:increase-C-blueprint}}\label{app:lem:increase-C-blueprint}

\begin{lemma*}
Let $\bG = (\bL,\bR,\bE,\bw)$ be a proper blueprint with parameters $(P,C,B_L,B_R,\sizes_L,\sizes_R)$ and $k \geq 1$ be an integer.
Then, there exists a proper blueprint $\hbG=(\hbL,\hbR,\hbE,\hbw)$ with parameters $(P,k \cdot C,B_L,B_R,\sizes_L,\sizes_R)$ such that $value(\hbG) = value(\bG)$.
\end{lemma*}
\begin{proof}
	In the following, we use integers in $[k \cdot C]$ and pairs in $[k] \times [C]$ interchangeably for the ease of notation (using any arbitrary bijection). The parameters of $\hbG$ are already chosen in the lemma statement, 
	which determines the vertices of $\hbLL{b}$ and $\hbRR{b'}$ for all $b \in [B_L]$ and $b' \in [B_R]$. What remains in the definition of $\hbG$ is to specify the edges $\hbE$. For every $p \in [P]$, define
	\begin{align*}
		\hbEE{p} &:= ~\text{all edges $\paren{\hbLL{b}((i_1,x_1), \ldots, (i_P,x_P)), \hbRR{b'}((i_1,y_1), \ldots, (i_P,y_P))}$} \\
		& \qquad \text{where} \qquad i \in [k]^P, \quad x, y \in [C]^P, \quad b \in [B_L], \quad b' \in [B_R],~\text{and} \quad (\bLL{b}(x), \bRR{b'}(y)) \in \bEE{p}.
	\end{align*}
    For each edge $\be$ of $\bE$, there are $k^P$ edges $\hbe \in \hbE$ that are added as a result of $\be$ in the definition above. For each such $\hbe$, we set $\hbw(\hbe) = \bw(\be)$. When comparing the $value(\hbG)$ with $value(\bG)$ in \Cref{def:blueprint-value-approx}, both the numerator and denominator of $value(\hbG)$ are $k^P$ times those of $value(\bG)$, so $value(\hbG) = value(\bG)$.

	It thus remains to verify that $\hbG$ is a proper blueprint.
	\begin{itemize}
		\item \textnormal{\textbf{Matching constraint:}} We only prove the constraint holds for any vertex in $\hbL$, the same constraint hold for vertices in $\hbR$ by symmetry. Consider vertex $\bv = \hbLL{b}((i_1, x_1), \ldots, (i_P, x_P))$, $b \in [B_L]$, $i \in [k]^P$, $x \in [C]^P$ that has at least one incident edge in $\hbG$. Each edge $\hbe$ incident to $\hbLL{b}((i_1, x_1), \ldots, (i_P, x_P))$ corresponds to a unique edge $\be$ incident to $\bLL{b}(x)$ in $\bG$, and $\be$ satisfies $\hbw(\hbe) = \bw(\be)$. By the matching constraints of $\bG$,
        \[ \sizes_L = \sum_{\be \ni \bLL{b}(x)} \bw(\be) = \sum_{\hbe \ni \bv} \hbw(\hbe), \]
        as required of the matching constraints of $\hbG$.
		\item \textnormal{\textbf{Ban constraints:}} Consider any edge $\hbe \in \hbEE{p}$ for $p \in [P]$. Suppose
		\[
		\hbe = \paren{\hbLL{b}((i_1,x_1), \ldots, (i_P,x_P)), \hbRR{b'}((i_1,y_1), \ldots, (i_P,y_P))},
		\]
		where $b \in [B_L]$, $b' \in [B_R]$, $i \in [k]^P$, and $x, y \in [C]^P$. The ban constraints induced by $\hbe$ is that for any choices of $(j_p, z_p), \ldots, (j_P, z_P) \in [k] \times [C]$, the vertices
		\begin{align*}
			&\hbLL{b}((i_1, x_1), \ldots, (i_{p-1}, x_{p-1}), (j_p, z_p), \ldots, (j_P, z_P)) \text{ and } \\
			&\hbRR{b'}((i_1, y_1), \ldots, (i_{p-1}, y_{p-1}), (j_p, z_p), \ldots, (j_P, z_P))
		\end{align*}
		are banned together. The former vertex has an incident edge in $\hbG$ if and only if $\bLL{b}(x_{<p} \conc (z_p, \ldots, z_P))$ has an incident edge in $\bG$, and the latter vertex has an incident edge in
		$\hbG$ if and only if $\bRR{b'}(y_{<p} \conc (z_p, \ldots, z_P))$ has an incident edge in $\bG$. These cannot happen simultaneously due to the ban constraint in $\bG$ from $(\bLL{b}(x), \bRR{b'}(y)) \in \bEE{p}$.
	\end{itemize}
	This proves that $\hbG$ is a proper blueprint, concluding the proof.
    \Qed{lem:increase-C-blueprint}
\end{proof}

\subsubsection{Proof of~\Cref{lem:integer-weights}}\label{app:lem:integer-weights}

\begin{lemma*}
    Let $\bG = (\bL,\bR,\bE,\bw)$ be a proper blueprint with parameters $(P,C,B_L,B_R,\sizes_L,\sizes_R)$ and $k \geq 1$ be a rational. Then, the blueprint $\hbG$, obtained from $\bG$ by changing nothing except scaling $\bw$, $\sizes_L$, and $\sizes_R$ each by $k$, is proper and satisfies $value(\hbG) = value(\bG)$.
\end{lemma*}
\begin{proof}
    Notice that the vertices and edges of $\hbG$ are the exact same as those of $\bG$. As the ban constraints have no relation with $\bw$, $\sizes_L$, or $\sizes_R$, we know that they are satisfied for $\hbG$. It only remains to check the matching constraint for $\hbG$. Consider any $b \in [B_l]$ (resp. $b' \in [B_R]$) and any vertex $\bu \in \bLL{b}$ (resp. $\bv \in \bRR{b'})$. We require that either $\bu$ (resp. $\bv$) is not incident on an edge of $\bE$ or that $\sum_{\be \ni \bu} k \cdot \bw(\be) = k \cdot \sizes_L(b)$ (resp. $\sum_{\be \ni \bv} k \cdot \bw(\be) = k \cdot \sizes_R(b')$). After cancelling $k$ from both sides of the equation, this condition becomes identical to the matching constraint of $\bG$, which we know holds since $\bG$ is assumed to be proper.

    Finally, the value of $\hbG$ is the same as the value of $\bG$, as we are scaling both the numerator and denominator in the value formula from \Cref{def:blueprint-value-approx} by $k$.
    \Qed{lem:integer-weights}
    
\end{proof}

\subsubsection{Proof of~\Cref{lem:remove-sizes-blueprint}}\label{app:lem:remove-sizes-blueprint}

\begin{lemma*}
Let $\bG = (\bL,\bR,\bE,\bw)$ be a proper blueprint with parameters $(P,C,B_L,B_R,\sizes_L,\sizes_R)$.
Then, there exists a proper blueprint $\hbG=(\hbL,\hbR,\hbE,\hbw)$ with parameters $(P,C,\hB_L,\hB_R,\hsizes_L,\hsizes_R)$, where $\hsizes_L = (1, \ldots, 1)~,~ \hsizes_R = (1, \ldots, 1)$, and $\hbw(\hbe) = 1$ for all edges $\hbe \in \hbE$, such that $value(\hbG) = value(\bG)$.
\end{lemma*}

\begin{proof}
    Without loss of generality, assume that $\sizes_L$, $\sizes_R$, and $\bw$ are integral, as otherwise we can scale them using \Cref{lem:integer-weights}. The parameters of $\hbG$ that are not yet specified are $\hB_L$ and $\hB_R$; set them to be $\hB_L = \sum_{b \in [B_L]} \sizes_L(b)$ and $\hB_R = \sum_{b' \in [B_r]} \sizes_R(b')$. The idea is to correspond each block of $\bG$ with relative size $s$ to $s$ different blocks of $\hbG$. We will use integers in $[\hB_L]$ and pairs in $\set{(b, i) : b \in [B], i \in [\sizes_L(b)]}$ interchangeably for the ease of notation (using any arbitrary bijection). The same goes with integers in $[\hB_R]$ and $\set{(b', i') : b' \in [B], i' \in [\sizes_R(b')]}$. Under this notation, the block $\hbLL{b, i}$ (resp. $\hbRR{b', i'}$) will be the $i$-th (resp. $i'$-th) block corresponding to $\bLL{b}$ (resp. $\bRR{b'}$).

    Consider $b \in [B_L]$ and $x \in [C]^P$. For each edge $\be$ incident to $\bLL{b}(x)$, we allocate it a subset $I_L(\be)$ of size $\bw(\be)$ from $[\sizes_L(b)]$ in such a way that across all the edges $\be$ incident to $\bLL{b}(x)$, the sets $I_L(\be)$ are pairwise disjoint and partition $[\sizes_L(b)]$. This can always be done as the matching constraint of $\bG$ ensures $\sum_{\be \ni \bLL{b}(x)} \bw(\be) = \sizes_L(b)$. We will similarly define $I_R(\be)$ for each $\be$ incident to a vertex $\bRR{b'}(y)$, $b' \in [B_R]$, $y \in [C]^P$. Finally, let $I(\be)$ be an arbitrary perfect matching between $I_L(\be)$ and $I_R(\be)$, i.e., $I(\be)$ is a set of $|I_L(\be)| = |I_R(\be)|$ pairs in $I_L(\be) \times I_R(\be)$, where each element of either $I_L(\be)$ or $I_R(\be)$ is in at least one pair.

    We are now ready to define the edges of $\hbG$. For every $p \in [P]$ define
    \begin{align*}
		&& \hbEE{p} := \quad \text{all edges} \quad & (\hbLL{b, i}(x), \hbRR{b', i'}(y)) \\
		&& \text{where} \quad & x, y \in [C]^P, \quad b \in [B_L], \quad b' \in [B_R], \\
        && & \be = (\bLL{b}(x), \bRR{b'}(y)) \in \bEE{p}, \quad (i, i') \in I(\be).
	\end{align*}
    By \Cref{def:blueprint-value-approx}, the value of $\hbG$ is
    \[ value(\hbG) = \frac{2 \sum_{\hbe \in \hbE} \hbw(\hbe)}{C^P \left( \sum_{b \in \hB_L} \hsizes_L(b) + \sum_{b' \in \hB_R} \hsizes_R(b') \right)}. \]
    The numerator works out to be $2\sum_{\be \in \bE} \bw(\be)$ as we add exactly $\bw(\be)$ edges in $\hbG$ for each edge $\be$ in $\bG$. The denominator is exactly $C^P \paren{\sum_{b \in [B_L]} \sizes_L(b) + \sum_{b' \in [B_R]} \sizes_R(b')}$ by the way we defined $\hB_L$ and $\hB_R$. Thus, the value of $\hbG$ is exactly that of $\bG$, and it only remains to verify that $\hbG$ is proper.

    \begin{itemize}
        \item \textnormal{\textbf{Matching constraint:}} As $\hsizes_L$, $\hsizes_R$ are all ones and $\hbw(\hbe) = 1$ for all $\hbe \in \hbE$, showing that $\hbG$ satisfies the matching constraint is equivalent to showing $\hbE$ forms a matching in the standard sense. For each $\be \in \bE$, let $\hbE_\be$ be the set of edges in $\hbG$ that are added as a result of $\be$ in the definition above. Certainly $\hbE_\be$ is a matching. Moreover, $I_L(\be_1)$ and $I_L(\be_2)$ are disjoint and $I_R(\be_1)$ and $I_R(\be_2)$ are disjoint for any distinct $\be_1, \be_2 \in \bE$, which, combined with the matching constraint on $\bG$, means that $\hbE_{\be_1}$ and $\hbE_{\be_2}$ share no endpoints. Hence, $\hbE$ is a union of vertex disjoint matchings, and thus it itself is a matching.
        \item \textnormal{\textbf{Ban constraints:}} Consider any edge $\hbe \in \hbEE{p}$ for $p \in [P]$. Suppose
        \[ \hbe = (\hbLL{b, i}(x), \hbRR{b', i'}(y)), \]
        where $b \in [B_L], b' \in [B_R], x, y \in [C]^P$, and $(i, i') \in I(\be)$ where $\be = (\bLL{b}(x), \bRR{b'}(y))\in \bEE{p}$. The ban constraints induced by $\hbe$ is that for any choices of $z_p, \ldots, z_P \in [C]$, the vertices
        \[ \hbLL{b, i}(x_1, \ldots, x_{p - 1}, z_p, \ldots, z_P) \quad \text{ and } \quad \hbRR{b', i'}(y_1, \ldots, y_{p - 1}, z_p, \ldots, z_P) \]
        are banned together. The former vertex can have an incident edge in $\hbG$ only if $\bLL{b}(x_{<p} \conc (z_p, \ldots, z_P))$ has an incident edge in $\bG$, and the latter vertex can have an incident edge in
		$\hbG$ only if $\bRR{b'}(y_{<p} \conc (z_p, \ldots, z_P))$ has an incident edge in $\bG$. These cannot happen simultaneously due to the ban constraint in $\bG$ from $(\bLL{b}(x), \bRR{b'}(y)) \in \bEE{p}$.
    \end{itemize}
    This proves that $\hbG$ is a proper blueprint, as desired.
    \Qed{lem:remove-sizes-blueprint}
    
\end{proof}

\subsection{From~\Cref{sec:construct-blueprints}}\label{app:omitted-construct-blueprints}

\subsubsection{Proof of~\Cref{clm:2-sqrt2-value}}\label{app:clm:2-sqrt2-value}
\begin{claim*}[``Value of $\hbG$ increases'']
	\[
		value(\hbG) > value(\bG) + \delta^3. 
	\]
\end{claim*}
\begin{proof}
 We have, 
\begin{align*}
	value(\hbG) &= \frac{2 \cdot \sum_{\be \in \hbE} \hbw(\be)}{C^{P+1} \cdot \paren{(1+s)+(1+s)}} \tag{by the definition of $value(\cdot)$ in~\Cref{def:blueprint-value-approx}}\\
	&= \frac{\sum_{\be:~\text{Type 1 edges}} \hbw(\be) + \sum_{\be:~\text{Type 2 edges}} \hbw(\be)}{C^{P+1} \cdot (1+s)} \tag{by partitioning between different types of added edges}\\
	&= \frac{\card{\Old} \cdot \card{\bE} + \card{\set{\hbLL{1}(i,*): i \in \New}} + \card{\set{\hbRR{1}(i,*): i \in \New}}}{C^{P+1} \cdot (1+s)}, 
	\intertext{as we have added $\card{\bE}$ many edges between each $\hbLL{1}(i,*,\ldots,*),\hbRR{1}(i,*,\ldots,*)$ for $i \in \Old$
	and further matched all vertices in these sets whenever $i \in \New$. These were all edges in $\hbG$. We can thus continue and have}
	&=\frac{\card{\Old} \cdot value{(\bG)} \cdot C^P + 2 \cdot \card{\New} \cdot C^P}{C^{P+1} \cdot (1+s)} \tag{by the definition of $value(\bG)$ and sizes of the involved sets}\\
	&=\frac{(1-\eps) \cdot value{(\bG)} + 2 \cdot \eps}{(1+s)} \tag{as $\card{\Old} = (1-\eps) \cdot C$ and $\card{\New} = \eps \cdot C$}
\end{align*}
	We would like to have $value(\hbG) > value(\bG) + \delta^3$, which given the above, means we need to prove 
	\begin{align*}
		(1-\eps) \cdot value{(\bG)} + 2 \cdot \eps &> value(\bG) \cdot (1+s) + \delta^3 \cdot (1+s). 
	\end{align*}
	The rest of the proof is just calculations. 
	By plugging in the value of $s$ in~\Cref{eq:2-sqrt2-s}, and upper bounding the $\delta^3(1+s)$ term (since $value(\bG)$ is bounded away from $0$ and $1$), we need to show 
	\begin{align*}
		(1-\eps) \cdot value{(\bG)} +2 \cdot \eps \geq value(\bG) +\frac{\eps \cdot value(\bG)}{(1-\eps) \cdot (1-value(\bG))} + 2\delta^3,
	\end{align*} 
	for small enough $\eps > 0$. This in turn simplifies to, after moving the first term of RHS to LHS and dividing by $\eps$, 
	\begin{align*}
		 -value{(\bG)} + 2 \geq \frac{value(\bG)}{(1-\eps) \cdot (1-value(\bG))} + \frac{2\delta^3}{\eps}.
	\end{align*}
	Given $value(\bG) < 2-\sqrt{2} < 2/3$ for small enough $\eps > 0$, the above follows from 
	\begin{align*}
		 -value{(\bG)} + 2 \geq \frac{value(\bG)}{(1-value(\bG))} + 4\eps + \frac{2\delta^3}{\eps},
	\end{align*}
	which in turn follows from 
	\begin{align*}
		- value{(\bG)} + value(\bG)^2 + 2-2 \cdot value(\bG) \geq value(\bG) + 4\eps + \frac{2\delta^3}{\eps},
	\end{align*}
	which is the same as 
	\begin{align*}
		value(\bG)^2 - 4\cdot value(\bG) + 2 \geq 4\eps + \frac{2\delta^3}{\eps}. 
	\end{align*}
	Considering the LHS as quadratic polynomial $x^2-4x+2$, we get its roots are $(2-\sqrt{2},2+\sqrt{2})$, 
	and the polynomial is non-negative for $x \leq 2-\sqrt{2}$ and when $x \leq 2-\sqrt{2}-\delta$, 
	we have that $x^2-4x+2 \geq 2\sqrt{2}\delta$. Thus, whenever $value(\bG) \leq 2-\sqrt{2}-\delta$, as in the statement of~\Cref{lem:2-sqrt2-amplification}, 
	\[
		value(\bG)^2 - 4\cdot value(\bG) + 2 \geq 2\sqrt{2}\delta,
	\]
	and thus we only need to have 
	\[
		2\sqrt{2}\delta \geq 4\eps + \frac{2\delta^3}{\eps}, 
	\]
	to conclude the proof. This can be achieved by setting $\eps = \sqrt{2}\delta/4$ and picking $\delta$ to be sufficiently small. Thus, 
	we proved that if $value(\bG) \leq 2-\sqrt{2}-\delta$, then $value(\hbG) > value(\bG)+\delta^3$. \Qed{clm:2-sqrt2-value}
	
\end{proof}

\subsubsection{Proof of~\Cref{obs:main-disjoint}}\label{app:obs:main-disjoint}

\begin{observation*}
	The sets $\bAL_i,\bAR_i,\bBL_i,\bBR_i$ for $i \in [K]$, and $\bCL,\bCR,\bDL,\bDR,\bOL,\bOR$ are all pairwise disjoint. 
\end{observation*}
\begin{proof}
	We only prove the bounds for the sets in $\hbL$; the other subsets in $\hbR$ are disjoint from these obviously and the rest can be proven by symmetry. 
	
	$\bAL_i$'s are disjoint from each other and the rest as they are subsets of $\hbLL{i}$, respectively, and no other sets uses these blocks. 
	
	$\bBL_i$'s as well as $\bCL$ are subsets of $\hbLL{K+1}$. Vertices in $\bBL_i$ have their first $i-1$ indices in $\Top_i$ but their $i$-th index in $\Bot_i$ and 
	so are disjoint from $\bBL_j$ for $j \neq i$. They are also disjoint from $\bCL$ because the latter vertices have their first $K$ indices in $\Top_1,\ldots,\Top_K$. 
	
	Finally, $\bDL$ as well as $\bOL$ are subsets of $\hbLL{K+2}$. They are pairwise disjoint because vertices in $\bDL$ have 
	their $(K+1)$-th index in $\New$ whereas vertices in $\bOL$ have that index in $\Old$. \Qed{obs:main-disjoint} 
\end{proof}

\subsubsection{Proof of~\Cref{obs:main-sizes}}\label{app:obs:main-sizes}

\begin{observation*}
	We have 
	\begin{alignat*}{2}
		\card{\bAL_i} &= \card{\bAR_i} =  \frac{t}{K+t-i+1} \cdot value(\bG) \cdot C^{\hP} \qquad && \forall i \in [K] \\
		\card{\bBL_i} &= \card{\bBR_i} = \frac{1}{K+t} \cdot C^{\hP} \qquad && \forall i \in [K] \\
		\card{\bCL} &= \card{\bCR} = \frac{t}{K+t} \cdot (1-\eps) \cdot (1-value(\bG)) \cdot C^{\hP} \qquad \\
		\card{\bDL} &= \card{\bDR} = \eps \cdot C^{\hP}.
	\end{alignat*}
\end{observation*}
\begin{proof}
	As before, we prove the bounds for the sets in $\hbL$, and infer the rest by symmetry. 
	
	For $\bAL_i$ for $i \in [K]$: 
	\begin{align*}
	 \card{\bAL_i} &= \paren{\prod_{j=i}^{K} \frac{\card{\Top_j}}{C}} \cdot \frac{\card{\Match_L}}{C^P} \cdot C^{\hP} \tag{by the definition of $\bAL_i$} \\
		&= \paren{\prod_{j=i}^{K} \frac{K+t-j}{K+t-j+1}} \cdot  value(\bG) \cdot C^{\hP} \tag{by definition of $\Top_j$ for $j \in [K]$ and $value(\bG)$ in~\Cref{def:blueprint-value-approx}} \\
		&= \frac{t}{K+t-i+1} \cdot value(\bG) \cdot C^{\hP} \tag{as the multiplicative terms cancel out telescopically}. 
	\end{align*}
	
	For $\bBL_i$ for $i \in [K]$: 
	\begin{align*}
	 \card{\bBL_i} &= \paren{\prod_{j=1}^{i-1} \frac{\card{\Top_j}}{C}} \cdot \frac{\card{\Bot_i}}{C} \cdot C^{\hP} \tag{by the definition of $\bBL_i$} \\
		&= \paren{\prod_{j=1}^{i-1} \frac{K+t-j}{K+t-j+1}} \cdot \frac{1}{K+t-i+1}\cdot C^{\hP} \tag{by definition of $\Top_j,\Bot_j$ for $j \in [K]$} \\
		&= \frac{1}{K+t} \cdot C^{\hP} \tag{as the multiplicative terms cancel out telescopically}. 
	\end{align*}
	
	For $\bCL$: 
	\begin{align*}
		\card{\bCL} & = \prod_{i=1}^{K} \paren{\frac{\card{\Top_i}}{C}} \cdot \frac{\card{\Old}}{C} \cdot \frac{\card{\Free_R}}{C^P} \cdot C^{\hP} \tag{by definition of $\bCL$} \\
		&= \prod_{i=1}^{K} \paren{\frac{K+t-i}{K+t-i+1}} \cdot (1-\eps) \cdot (1-value(\bG)) \cdot C^{\hP} \tag{by definition of $\Top_j$ for $j \in [K]$, $\Old$, and $value(\bG)$ in~\Cref{def:blueprint-value-approx}} \\
		&= \frac{t}{K+t} \cdot (1-\eps) \cdot (1-value(\bG)) \cdot C^{\hP} \tag{as the multiplicative terms cancel out telescopically}. 
	\end{align*}
	
	For $\bDL$: 
	\begin{align*}
		\card{\bDL} = \frac{\New}{C} \cdot C^{\hP} = \eps \cdot C^{\hP} \tag{by the definition of $\bDL$ and $\New$}. 
	\end{align*}
	\Qed{obs:main-sizes}
\end{proof}

\end{document}